%% file: gta-arxiv.tex
\documentclass[a4paper,UKenglish,thm-restate]{lipics-v2021}

\usepackage{float}
\usepackage{microtype}
\usepackage{wrapfig}
\usepackage{tikz}
\usetikzlibrary{decorations.pathreplacing,decorations.markings,calligraphy}
\usetikzlibrary{arrows,automata}
\usetikzlibrary{positioning}

\usepackage{graphicx}

\newcommand{\TS}{\mathbb{TS}}

\nolinenumbers

\input{declarations.tex}
\title{A Unified Model for Real-Time Systems: Symbolic Techniques and Implementation}

\titlerunning{A Unified Model for Real-Time Systems}

\author{S Akshay}{Department of CSE, Indian Institute of Technology Bombay, Mumbai, India}
{akshayss@cse.iitb.ac.in}{https://orcid.org/0000-0002-2471-5997}{Supported in part by DST/SERB Matrics Grant MTR/2018/000744.}

\author{Paul Gastin}{Université Paris-Saclay, ENS Paris-Saclay, CNRS, LMF, 91190,
	Gif-sur-Yvette, France \and CNRS, ReLaX, IRL 2000, Siruseri, 
  India}{paul.gastin@ens-paris-saclay.fr}{https://orcid.org/0000-0002-1313-7722}{Partially supported by ANR project Ticktac (ANR-18-CE40-0015).}

\author{R Govind}{Department of CSE, Indian Institute of Technology Bombay, Mumbai, India} 
{govindr@cse.iitb.ac.in}{https://orcid.org/0000-0002-1634-5893}{}

\author{Aniruddha R Joshi}{Department of CSE, Indian Institute of Technology Bombay, Mumbai, India} 
{aniruddhajoshi@cse.iitb.ac.in}{https://orcid.org/0000-0003-1884-7894}{}

\author{B Srivathsan}{Chennai Mathematical Institute, India
  \and CNRS, ReLaX, IRL 2000, Siruseri, 
  India} {sri@cmi.ac.in}{https://orcid.org/0000-0003-2666-0691}{}

\authorrunning{S. Akshay, P. Gastin, R. Govind, A. Joshi and  B. Srivathsan}
\Copyright{S. Akshay, P. Gastin, R. Govind, A. Joshi and  B. Srivathsan}

\keywords{Real-time systems, Timed automata, Event-clock automata, Clocks, Timers, Verification, Zones, Simulations, Reachability}

\category{}
 \ccsdesc[500]{Theory of computation~Timed and hybrid models}
 \ccsdesc[300]{Theory of computation~Quantitative automata}
 \ccsdesc[100]{Theory of computation~Logic and verification}

\relatedversion{} 

\supplement{}

\funding{This work was supported by DST/CEFIPRA/INRIA Project EQuaVE.}

\hideLIPIcs  

\begin{document}

\maketitle

 \begin{abstract}
   In this paper, we consider a model of \emph{\GTAfull} (\GTA) with
   two kinds of clocks, \emph{history} and \emph{future}, that can express 
   many timed features succinctly, including timed automata, event-clock
   automata with and without diagonal constraints, and automata with timers.
   
   Our main contribution is a new simulation-based zone algorithm for checking
   reachability in this unified model.  While such algorithms are known to exist for
   timed automata, and have recently been shown for event-clock automata without diagonal
   constraints, this is the first result that can handle event-clock automata with
   diagonal constraints and automata with timers.  We also provide a prototype
   implementation for our model and show experimental results on several benchmarks.  To the best of our
   knowledge, this is the first effective implementation not just for our unified
   model, but even just for automata with timers or for event-clock automata (with
   predicting clocks) without going through a costly translation via timed automata.
   Last but not least, beyond being interesting in their own right, \GTAfull\ can be used
   for model-checking event-clock specifications over timed automata models.
 \end{abstract}

 \section{Introduction}\label{sec:intro}

 The idea of adding real-time dynamics to formal verification models started as a hot topic
 of research in the 1980s~\cite{de1992real,AlurPhD}.  Over the years, timed
 automata~\cite{DBLP:conf/icalp/AlurD90, AD94} has emerged as a leading model for
 finite-state concurrent systems with real-time constraints.  Timed automata make use of
 \emph{clocks}, real-valued variables which increase along with time.  Constraints over
 clock values can be used as guards for transitions, and clocks can be reset to $0$ along
 transitions.  It is notable that the early works in this area made use of \emph{timers} to
 deal with real-time~\cite{DBLP:conf/podc/KoymansVR83, Dill89, DBLP:conf/sosp/BernsteinH81}.
 Timers are started by setting them to some initial value within a given interval.  Their
 values decrease with time, and an \emph{timeout} event can be used in transitions to detect
 the instant when the timers become $0$.
 Quoting from~\cite{AlurPhD}, the shift from timers to clocks in timed automata, as we know
 them today, is attributed to the fact that: ``\emph{apart from some technical conveniences
 in developing the emptiness algorithm and proving its correctness, the reformulation
 allows a simple syntactic characterization of determinism for timed automata}''.  Over the
 last thirty years, the study of timed automata has led to the development of rich theory
 and industry-strength verification tools.  The use of clocks has also allowed for the
 extension of the model to more complex constraints and assignments to clocks in
 transitions~\cite{Bouyer04,DBLP:journals/tcs/BouyerDFP04}.  Furthermore, considering more
 sophisticated rates of evolution for clocks gives the yet another well-established model
 of hybrid automata~\cite{DBLP:conf/hybrid/AlurCHH92}.
 
  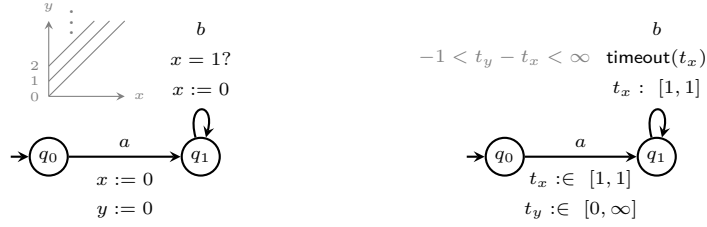
\begin{figure}[t]
    \centering
    \begin{tikzpicture}[state/.style={circle, inner sep=2pt, minimum
        size = 3pt, draw, thick}]
      \begin{scope}[every node/.style={state}]
        \node (0) at (0,0) {\scriptsize $q_0$}; \node (1) at (2,0)
        {\scriptsize $q_1$};
      \end{scope}
      \begin{scope}[thick, ->,>=stealth, auto]
        \draw (-0.5, 0) to (0); \draw (0) to node {\scriptsize $a$}
        (1); \draw (1) to [loop above] (1);
      \end{scope}
      \node at (1, -0.3) {\scriptsize $x := 0$}; \node at (1, -0.7)
      {\scriptsize $y := 0$}; \node at (2, 1.3) {\scriptsize $x = 1?$};
      \node at (2, 0.9) {\scriptsize $x := 0$}; \node at (2, 1.7)
      {\scriptsize $b$};
 
      \begin{scope}[gray]
        \draw [->, >=stealth] (0, 0.8) -- (1, 0.8); \draw [->,
        >=stealth] (0, 0.8) -- (0, 1.8); \draw (0, 0.8) -- (1, 1.8);
        \draw (0, 1) -- (0.8, 1.8); \draw (0, 1.2) -- (0.6, 1.8); \draw
        (0.3, 1.5) node [above] {$\vdots$} (0.3, 1.8); \node [right] at
        (1, 0.8) {\tiny $x$}; \node [above] at (0, 1.8) {\tiny $y$};
        \node [left] at (0, 0.8) {\tiny $0$}; \node [left] at (0, 1)
        {\tiny $1$}; \node [left] at (0, 1.2) {\tiny $2$};
      
      \end{scope}
 
      \begin{scope}[xshift=6cm]
        \begin{scope}[every node/.style={state}]
          \node (0) at (0,0) {\scriptsize $q_0$}; \node (1) at (2,0)
          {\scriptsize $q_1$};
        \end{scope}
        \begin{scope}[thick, ->,>=stealth, auto]
          \draw (-0.5, 0) to (0); \draw (0) to node [above]
          {\scriptsize $a$} (1); \draw (1) to [loop above] (1);
        \end{scope}
        \node at (1, -0.3) {\scriptsize $t_x :\in~ [1,1]$}; \node at (1,
        -0.7) {\scriptsize $t_y :\in~ [0, \infty]$}; \node at (2, 1.3)
        {\scriptsize $\timeout(t_x)$}; \node at (2, 0.9) {\scriptsize
          $t_x :~ [1, 1]$}; \node at (2, 1.7) {\scriptsize $b$}; \node
        [gray] at (0, 1.3) {\scriptsize $-1 < t_y - t_x < \infty$} ;
      \end{scope}
    \end{tikzpicture}
    \caption{An automaton with clocks on left, and timers on
      right for same constraints.}
    \label{fig:timers-vs-clocks}
  \end{figure}
  
 When it comes to the reachability problem, timers do have some nice properties.  Let us
 explain with an example.  Figure~\ref{fig:timers-vs-clocks} shows a timed automaton on the
 left, and an automaton with timers on the right, for the set of words $ab^*$ such that the
 time between every consecutive letters is $1$.  The timed automaton sets clock $x$ to $0$
 and checks for the guard $x = 1?$ to enforce the timing constraint.  The automaton with
 timers, on the right, sets a timer $t_x$ to $1$, and asks for its expiry in the immediate
 next action.  Clock $y$ and timer $t_y$ are not necessary for the required timing
 property, but we add them to illustrate a different aspect that we will describe now.  To
 solve the reachability problem, a symbolic enumeration of the state space is performed.
 In the timed automaton, at state $q_1$, the enumeration gives constraints $y - x = n$ for
 every $n \ge 0$.  Starting from $y - x = n$ and executing $b$ gives $y - x = n + 1$, due
 to the combination of guard $x = 1?$ and reset $x := 0$.  This shows that a na\"ive
 symbolic enumeration is not bound to terminate.  The question of developing finite
 abstractions for timed automata has been a central problem of study which started in the
 late 90s and continues till date (see recent
 surveys~\cite{DBLP:conf/formats/BouyerGHSS22,DBLP:journals/siglog/Srivathsan22}).
 Such an issue does not occur with timers.  In the automaton with timers on the right,
 $t_x$ is set to $1$ and $t_y$ is set to some arbitrary value in the transition to $q_1$.
 This gives $-1 \le t_y - t_x \leq \infty$ for the set of all possible timer values.  When
 $t_x$ times out, the value of $t_y$ could still be any value from $0$ to $\infty$.  When
 $t_x$ is set to $1$ again, the set of possible timer values still satisfies the same
 constraint $-1 \le t_y - t_x \leq \infty$ leading to a fixed point with a finite reachable
 state space.  The fact that symbolic enumeration terminates on an automaton with timers
 was already observed in~\cite{Dill89}.  To our knowledge, later works on timed automata
 reachability never went back to timers, and there is no tool support that we know of to
 deal with models with timers directly.  We find this surprising given that timers occur
 naturally while modeling real-time systems and moreover they enjoy this finiteness
 property.
 
 In addition to clocks and timers, \emph{event-clocks} are another special type of clock
 variables that are used to deal with timing constraints~\cite{AFH99}, which are attached
 to events.  An event-recording clock for event $a$ maintains the time since the previous
 occurrence of $a$, whereas an event-predicting clock for $a$ gives the time to the next
 occurrence of $a$.
 Event-clocks have been used in the model of event-clock automata (ECA), and also in the
 logic of event-clocks~\cite{RaskinS99}.  These works argue that
 event-clocks can express typical real-time requirements.  Theoretically, ECA can be
 determinized, and hence complemented.  Therefore, model-checking an event-clock (logic or
 automaton) specification $\varphi$ over a timed automaton $\Aa$ can be reduced to
 reachability on the product of $\Aa$ and the ECA for $\neg \varphi$.  This makes
 event-clocks a convenient feature in specifications.
 
 Recently, a symbolic enumeration algorithm for ECA was proposed~\cite{Concur22}.  It was
 noticed that when restricted to event-predicting clocks, the symbolic enumeration
 terminates without any additional checks (similar to the case of timers), whereas for the
 combination involving event-recording clocks, one needs simulation techniques from the
 timed automata literature.  The same work showed how to adapt the best known simulation
 technique from timed automata into the setting of ECA. However, as discussed above, for
 model-checking we need a model containing both conventional clocks, timers and
 event-clocks.  To our knowledge, no tool can directly work on such models.
 
 Our goal in this work is to provide a one stop solution to real-time verification, be it
 reachability analysis or model-checking (over event-clock specifications), be it using
 models with clocks, or models with timers.  We consider a unified model of a timed automaton over variables that can simulate normal clocks, timers and event-clocks.
 Here are our key contributions:
 \begin{enumerate}
   \item We define a new model of \GTAfull\ (\GTA) which have two types
   of variables, called \emph{history} clocks and \emph{future} clocks.  History clocks
   generalize normal clocks as well as event-recording clocks, while future clocks
   generalize event-predicting clocks and timers.  However, unlike event-clocks, clocks in
   \GTA\ are not necessarily associated with events.  We also consider a generic syntax that
   allows for diagonal constraints between variables.
   
  \item We show undecidability of reachability for \GTA, and study a \emph{safe subclass}
  that makes the model decidable.  Safe \GTA\ already subsume timed automata, event-clock
  automata (with diagonal constraints) and automata with timers.
 
  \item We adapt state-of-the-art symbolic enumeration techniques from timed automata
  literature to safe \GTA.  While we make use of ideas presented in~\cite{Dill89} and
  \cite{Concur22}, these works do not contain diagonal constraints between variables.  Our
  main technical and theoretical innovation lies in a new termination analysis of the
  symbolic enumeration in the presence of diagonal constraints.  Surprisingly, we show that
  the enumeration terminates as long as the diagonal constraints are restricted to usual
  clocks and event-clocks, but not timers.
  
  \item We develop a prototype implementation of our model and algorithm in
  \textsc{Tchecker}, an open-source platform for timed automata analysis, and show
  promising results on several existing and new benchmarks.  To the best of our knowledge,
  our tool is the first that can handle event-clock automata, a model that till date has
  been the subject of many theoretical results.
 \end{enumerate}
 
 \noindent{\em Related works.}  
 In the work that first introduced ECA, 
 a translation from ECA to a timed automaton was also proposed.
 However, this translation is not efficient: in the worst case, this translation incurs a
 blowup in the number of clocks and states.  In~\cite{GeeraertsRS11, GeeraertsRS14}, an
 extrapolation approach using maximal constants has been studied for ECA. However,
 it has been observed that simulation-based techniques are both more
 effective~\cite{Bouyer04,DBLP:conf/cav/BouyerCM16} and
 efficient~\cite{GastinMS18,GastinMS19,GastinMS20,AkshayGP21} than extrapolation for
 checking reachability.  Recently, \cite{Concur22} proposed a zone-based reachability
 algorithm for diagonal-free ECA, using simulations for finiteness, but there was no accompanying implementation. 
 Diagonal constraints have long been known to allow succinct modeling~\cite{Chevalier-Bouyer} for the class of timed-automata, but only recently a zone-based algorithm that directly works on such automata, was proposed.
 ECA with diagonals are more expressive than ECA~\cite{DBLP:conf/fsttcs/BozzelliMP19a}.
 In this work, we propose a zone-based algorithm for a unified model that subsumes ECA with diagonals.
 
 The use of history clocks and prophecy clocks in ECAs is in the same spirit as past and future modalities in temporal logics - this makes ECAs an attractive model for writing timed specifications.  Indeed, this has also led to a development of various temporal logics with event-clocks~\cite{DSouzaT04,AkshayBG13,RaskinS99}.  
 ECA with diagonal constraints have been well-studied, such as in the context of timeline based planning~\cite{DBLP:conf/fsttcs/BozzelliMP19a,DBLP:journals/tcs/BozzelliMP22}. 
 Finally, while there has been substantial advances in the theory of ECA, to the best of our knowledge, the only tool that handles ECA is \textsc{Tempo}~\cite{Tempo}, and even this tool is restricted to just history clocks.
 
  \medskip\noindent{\em Structure of the paper.}
  % \notego{Modified}
 In Section~\ref{sec:prelims} we start by defining the generalized model.
 Section~\ref{sec:examples} examines its expressiveness, while
 Section~\ref{sec:reachability} deals with the reachability problem and the safe subclass.
 Section~\ref{sec:simulation} develops the symbolic enumeration technique, while
 Section~\ref{sec:DBM} explains how distance graphs can be extended to this setting.
 In Section~\ref{sec:dagger}, we discuss some of the properties of distance graphs of reachable zones.
 Section~\ref{sec:finiteness} is dedicated to finiteness.  Finally, we provide our experimental results in Section~\ref{sec:experiments} and conclude with
 Section~\ref{sec:conclusion}.

\section{\GTAFULL}\label{sec:prelims}

In this section we introduce the unified model.  While we build on classical ideas from timed automata, almost every aspect is extended and below we highlight these changes. 

\subsection{Extending clocks and constraints}
We define $X = X_H \uplus X_F$ to be a finite set of real-valued variables called \emph{clocks}, where $X_H$ is the set of \emph{history clocks}, and $X_F$ is the set of \emph{future clocks}.

Let $\overline{\mathbb{R}}=\mathbb{R}\cup \{-\infty,+\infty\}$ denote the set of
all real numbers along with $-\infty$ and $+\infty$.  
The usual $<$ order on reals is extended to deal with $\{-\infty, +\infty\}$ as: $-\infty
< c < +\infty$ for all $c \in \mathbb{R}$ and $-\infty < \infty$.  Similarly,
$\overline{\mathbb{Z}}=\mathbb{Z}\cup \{-\infty,+\infty\}$ denotes the set of all integers
along with $-\infty$ and $+\infty$.  Let $\Rpos$ (resp.  $\Rneg$) be the set of
non-negative (resp.  non-positive) reals.

  \begin{definition}[Weights]\label{def:weights}
    Let $\mathcal{C} = \{(\leqlt, c) \mid c \in \overline{\mathbb{R}} \text{ and } {\leqlt} \in \{\leq, <\}\}$, called the set of weights.
  \end{definition}
  
Let $X \cup \{0\}$ be the set obtained by extending the clocks of \GTA\ with the special constant clock $0$.   Note that this clock will always have the value 0. Let $\Phi(X)$ denote a set of clock constraints generated by the following
  grammar: 
  $$\varphi ::= x - y \leqlt c \mid \varphi \land \varphi$$
  where $x,y \in X \cup \{0\}$, $(\leqlt,c) \in \mathcal{C}$ and $c \in \overline{\mathbb{Z}}$. 
  The introduction of the special constant clock $0$ allows us to treat constraints with just a single clock as special cases.
Note that the constraint $x\leqlt c$ is equivalent to $x-0\leqlt c$ and the constraint
$c\leqlt x$ is equivalent to $0-x\leqlt -c$.
We often write $x=c$ as a shorthand for $x\leq c \wedge c\leq x$.
The base constraints of the form $x - y \leqlt c$ will be called \emph{atomic constraints}.  
The atomic constraints of the form $x - y \leqlt c$ such that $x \neq 0 \neq y$ are called \emph{diagonal constraints}, and all other atomic constraints are called \emph{non-diagonal constraints}.   

\subsection{Extending valuations}
We first recall and discuss the \emph{extended algebra} of weights that was proposed in a recent work on event-clock automata~\cite{Concur22,eca-simulations-arxiv}.
To evaluate the constraints allowed by $\Phi(X)$, we extend addition on real numbers with the convention that $(+\infty)+\alpha = \alpha+(+\infty) = +\infty$ for all $\alpha\in\overline{\RR}$ and $(-\infty)+\beta = \beta+(-\infty) = -\infty$, as long as $\beta\neq+\infty$.  
We also extend the unary minus operation from real numbers to
$\overline{\mathbb{R}}$ by setting $-(+\infty)=-\infty$ and $-(-\infty)=+\infty$.  Abusing notation, we write $\beta-\alpha$ for $\beta+(-\alpha)$. 
Notice that with this definition of extended addition, the minus operation does not distribute over addition.\footnote{
Notice that $-(a+b)=(-a)+(-b)$ when $a$ or $b$ is finite or when $a=b$. But, when $a=+\infty$ and $b=-\infty$ then $-(a+b)=-\infty$ whereas $(-a)+(-b)=+\infty$.}. 

  We first highlight a few more important features of the definition of extended addition operation.

\begin{remark}\cite{eca-simulations-arxiv}\label{rem:extended-addition}
  This extended addition has the following properties that are easy to check:
  \begin{enumerate}
    \item $(\overline{\mathbb{R}},+,0)$ is a monoid with $0$ as neutral element.  In particular, the extended addition is associative.

    \item $(\overline{\mathbb{R}},+,0)$ is not a group, since $-\infty$ and $+\infty$ have
    no opposite values.  Note that, $\alpha+(-\alpha)=0$ when $\alpha\in\mathbb{R}$ is
    finite but $\alpha+(-\alpha)=+\infty$ when $\alpha\in\{-\infty,+\infty\}$.  As a
    consequence, in an equation $\alpha+\beta=\alpha+\gamma$, we can cancel $\alpha$ and
    deduce $\beta=\gamma$ when $\alpha$ is finite, but not when $\alpha$ is infinite.
    
    \item The order $\leq$ is monotone on $\overline{\mathbb{R}}$: $b\leq c$ implies
    $a+b\leq a+c$, but the converse implication only holds when $a$ is finite.

    \item The strict order $<$ is only monotone with respect to finite values: when $a$ is
    finite, $b<c$ iff $a+b<a+c$.

    \item For all $a,b\in\overline{\mathbb{R}}$ and $(\leqlt,c)\in\mathcal{C}$, we have
    $a\leqlt b$ iff $-b\leqlt -a$.  Further, $a-b\leqlt c$ implies $a\leqlt b+c$. The
    converse of the latter statement holds when $b$ is finite.  Note that the
    converse may be false when $b$ is infinite.\footnote{
    For instance, if $a<+\infty=b$ then $a<b+(-\infty)$, but $a-b=-\infty\not<-\infty$.
    
    If $a<+\infty$ and $b=-\infty$ then $a<b+\infty$, but $a-b=+\infty\not<+\infty$.
    
    If $a=b\in\{-\infty,+\infty\}$ and $c$ is finite then $a\leq b+c$, but
    $a-b=+\infty\not\leq c$.}
  \end{enumerate}  
\end{remark}

\begin{definition}[Valuation]
  A valuation of clocks is a function
  $v\colon X \cup \{0\} \mapsto \overline{\RR}$ which
  maps the special clock $0$ to 0,
  history clocks to $\Rpos\cup\{+\infty\}$ 
  and future clocks to $\Rneg \cup \{-\infty\}$. 
  We denote by $\V(X)$ or simply by $\V$ the set of valuations over $X$.
  We say that clock $x$ is \emph{defined} (resp.\ \emph{undefined}) in $v$
  when $v(x)\in\mathbb{R}$ (resp.\ $v(x)\in\{-\infty,+\infty\}$).
\end{definition}

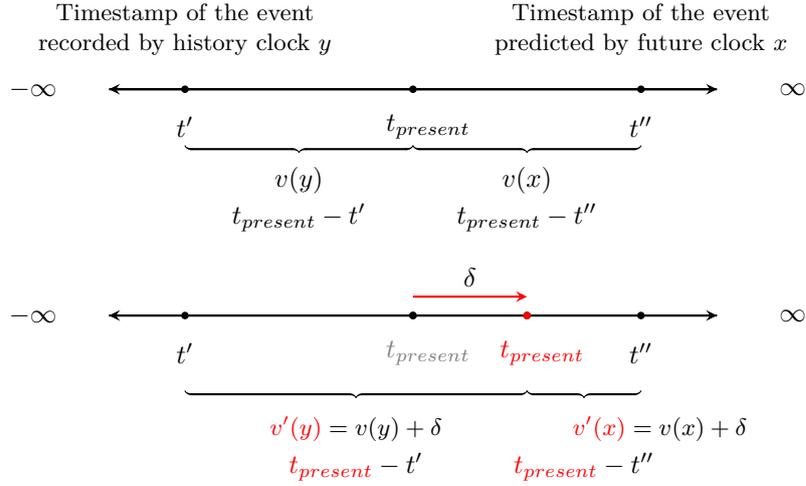
\begin{figure}
  \centering
  \begin{tikzpicture}[state/.style={draw, thick, circle, inner sep=2pt}]

    \begin{scope}[->, >=stealth, thick]
      \draw (0,4) to (8,4);
      \draw (8,4) to (0,4);
  
      \draw (0,7) to (8,7);
      \draw (8,7) to (0,7);
    \end{scope}
  
    \node at (1,4) [circle,fill,inner sep=1pt]{};
    \node at (7,4) [circle,fill,inner sep=1pt]{};
  
    \node at (1,8)  {\small \textcolor{black}{Timestamp of the event} };
    \node at (1,7.6)  {\small \textcolor{black}{recorded by history clock $y$} };
    \node at (7,8)  {\small \textcolor{black}{Timestamp of the event} };
    \node at (7,7.6)  {\small \textcolor{black}{predicted by future clock $x$} };
  
    \node at (-1,7)  {$- \infty$ };
    \node at (9,7)  {$\infty$ };
    \node at (1,6.5)  { $t'$ };
    \node at (7,6.5)  { $t''$ };
    \node at (1,7) [circle,fill,inner sep=1pt]{};
    \node at (7,7) [circle,fill,inner sep=1pt]{};
    \node at (4,7) [circle,fill,inner sep=1pt]{};
  
    \node at (4,4) [circle,fill,inner sep=1pt]{};
    \node at (4.2,6.5)  { $t_{present}$ };
  
    \draw [thick, decorate,
    decoration = {calligraphic brace}] (4,6.25) --  (1,6.25);
    \draw [thick, decorate,
    decoration = {calligraphic brace}] (7,6.25) --  (4,6.25);
  
    \node at (2.5,5.8)  { \textcolor{black}{$v(y)$} };
    \node at (2.5,5.3)  { $t_{present} - t'$ };
    \node at (5.5,5.8)  { \textcolor{black}{$v(x)$} };
    \node at (5.5,5.3)  { $t_{present} - t''$ };    
  
    \node at (-1,4)  {$- \infty$ };
    \node at (9,4)  {$\infty$ };
    \node at (1,3.5)  { $t'$ };
    \node at (7,3.5)  { $t''$ };
    \node [red] at (5.7,3.5)  {  \textcolor{red}{$t_{present}$} };
    \node at (4,4) [circle,fill,inner sep=1pt]{};
    \node at (4.2,3.5)  { \textcolor{gray}{$t_{present}$}};
    \node at (5.5,4) [circle,fill,inner sep=1pt]{};
  
      \begin{scope}[->, >=stealth, thick]
        \draw [thick,red] (4,4.25) to (5.5,4.25);
      \end{scope}
      \node at (4.75,4.5) {$\d$};
  
      \node [red] at (5.5,4) [circle,fill,inner sep=1pt]{};

      \draw [thick, decorate,
      decoration = {calligraphic brace}] (5.5,3) --  (1,3);
      \draw [thick, decorate,
      decoration = {calligraphic brace}] (7,3) --  (5.5,3);
  
      \node at (3.25,2.5)  {\small \textcolor{black}{$\textcolor{red}{v'(y)} = v(y) + \delta$} };
      \node at (3.25,2)  { $ \textcolor{red}{t_{present}} - t'$ };
      \node at (7.25,2.5)   { \small \textcolor{black}{$\textcolor{red}{v'(x)} = v(x) + \delta$} };
      \node at (6.25,2)  { $ \textcolor{red}{t_{present}} - t''$ };  
  \end{tikzpicture}
    \captionof{figure}{Representation of valuations in \GTAfull. Here, $v' = v+ \d$.}
  \label{fig:gta-valuation}
\end{figure}

\begin{definition}\label{def:diagonal-semantics}
  Let $x,y\in X\cup\{0\}$ be clocks (including 0) and let $(\leqlt,c)$ be a weight.
  For valuations $v\in\V$, define $v\models y-x \leqlt c$ as $v(y)-v(x)\leqlt c$.
  We say that a valuation $v$ satisfies a constraint $\varphi$, denoted as 
  $v\models\varphi$, when $v$ satisfies all atomic constraints in $\varphi$.
\end{definition}

\begin{remark}
  From Definition~\ref{def:diagonal-semantics}, we easily check that the constraint
  $y-x\leqlt c$ is equivalent to \emph{true} (resp.\ \emph{false}) when
  $(\leqlt,c)=(\leq,+\infty)$ (resp.\ $(\leqlt,c)=(<,-\infty)$).
  Constraints that are equivalent to \emph{true} or \emph{false} will be called trivial, whereas all others are non-trivial constraints.
  
  If $(\leqlt,c)\neq(\leq,+\infty)$ then $v\models y-x \leqlt c$ never holds when $v(x)=-\infty$.  
  
  Also, if $v(x)=v(y)\in\{-\infty,+\infty\}$ then $v\models y-x \leqlt c$ only holds for $(\leqlt,c)=(\leq,+\infty)$.
  
  For a non-trivial constraint $y-x \leqlt c$, i.e., $(\leqlt,c)\notin 
  \{(<,-\infty),(\leq,+\infty)\}$,
  we have 
  \begin{itemize}
    \item $v\models y-x\leqlt c$ iff  $v(y)<+\infty=v(x)$ or ($v(x)$ is finite and $v(y) \leqlt v(x) + c$).
    
    \item $v\models y-x \leq -\infty$ iff $v(y)<+\infty=v(x)$ or $v(y)=-\infty<v(x)$.
    
    \item $v\models y-x < +\infty$ iff $v(x)\neq-\infty$ and $v(y)\neq+\infty$.
    \qed
  \end{itemize}
\end{remark}

We abuse notation and  for $Y \subseteq X$, we define 
$Y \triangleleft c$ as $\displaystyle\bigwedge_{y \in Y} y \triangleleft c$, and
$Y = c$ as $\displaystyle\bigwedge_{y \in Y} y = c$.

We denote by $v + \d$ the valuation obtained from valuation $v$ by increasing by $\d
\in \Rpos$ the value of all clocks in $X$.
Note that, from a given valuation, not all time elapse result in valuations since future clocks need to stay at most $0$. 
For example, from a valuation with $v(x) = -3$ and $v(y) = -2$, where $x,y$ are future clocks, one can elapse at most $2$ time units. 

\subsection{Extending resets}
For history clocks, the reset operation sets the clock to 0.  For future clocks, the reset
operation says that all constraints on the clock must be discarded, i.e., the
clock is {\em released}.  Given that the set of clocks is partitioned into history clocks
and future clocks, we use the same notation $[R]v$ to talk about the change of clocks in
$R$, whether it be reset/release, which operates differently depending on whether the
clock is a history or future clock.
More precisely, given a set of clocks $R \subseteq X$, we define $R_{F} = R \cap X_{F}$ as the set of future clocks in $R$ and $R_{H} = R \cap X_{H}$ as the set of history clocks in $R$. We then define $[R]v$ as follows:
$$[R]v := \{v' \in \V \mid v'(x) = 0~\forall ~x \in R_{H} \text{ and } v'(x) = v(x) ~\forall~ x \not\in R\}$$
Observe that, {\em the release operation} is implicit: each future clock in $R$ could take any value (not necessarily the same) from $[-\infty,0]$ in $[R]v$.
Note that $[R]v$ is a singleton when $R$ contains only history clocks - this corresponds
exactly to the reset operation in timed automata. In this case, we simply write $v'=[R]v$ 
instead of $\{v'\}=[R]v$.
When $R$ contains only future clocks, $[R]v$ is the set of valuations obtained by
releasing each clock in $R$ (i.e., setting each clock in $R$ non-deterministically to some
value in $[-\infty,0]$, while keeping the value of other clocks unchanged).
For $W\subseteq\V$, we let $[R]W=\bigcup_{v\in W}[R]v$.
We have $[R'\cup R'']W=[R']([R'']W)$.

\subsection{Extending guards and transitions}
Before we define \GTA, let us focus on the language to specify transitions.  In normal
timed automata, as shown in Figure~\ref{fig:ta-transition}, a transition reads a letter,
checks a guard and then resets a subset of (history) clocks.  The guard $g$ can capture
multiple constraints by allowing a conjunction of atomic constraints and resetting the
subset $R\subseteq X$ of clocks corresponds to resetting them one by one.  But in any one
transition only a pair of guard, reset is performed and one cannot interleave them.

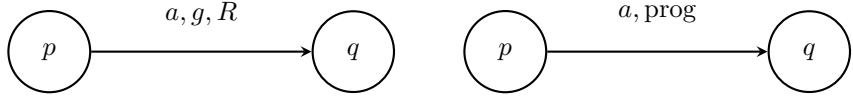
\begin{figure}
  \centering
  \begin{tikzpicture}[state/.style={draw, thick, circle, inner sep=1pt}]
    \begin{scope}[scale=1, every node/.style={scale=1}]
  
    \begin{scope}[state/.style={draw, thick, circle, inner sep=8pt}]
      \node [state] (p) at (2, 6) { $p$};
      \node [state] (q) at (6, 6) { $q$};
    \end{scope}
  
    \begin{scope}[->, >=stealth, thick]
      \draw (p) to (q);
    \end{scope}
  
    \node at (4, 6.5) { $a, g, R$};

    \begin{scope}[xshift=6cm]
      \begin{scope}[state/.style={draw, thick, circle, inner sep=8pt}]
        \node [state] (p) at (2, 6) { $p$};
        \node [state] (q) at (6, 6) { $q$};
      \end{scope}
    
      \begin{scope}[->, >=stealth, thick]
        \draw (p) to (q);
      \end{scope}
      \node at (4, 6.5) { $a, \prog$}; 
    \end{scope}

    \end{scope}
  \end{tikzpicture}
  \captionof{figure}{A transition of TA (left) and of a \GTA\ (right)}
  \label{fig:ta-transition}
\end{figure}

We generalize this to our setting with history and future clocks but also to allow
arbitrary interleaving of guards and changes\footnote{To model this with a TA one may use 
a sequence of multiple transitions without delays in-between.}.  Let us formalize this.
An {\em instantaneous timed program} is generated by the following grammar:
\begin{align}
  \prog  &:= \guard \mid \chg \mid \prog; \prog
  \qquad\text{ where,}  \\
  & \quad \guard =g\in \Phi(X)\\
  & \quad \chg   = [R] \text{ for some } R\subseteq X
\end{align}
While $\guard$ and $\chg$ are atomic programs, $\prog;\prog$ refers to sequential
composition.  The set of all programs generated by the above grammar will be denoted
$\Prog$.  Then on a transition, we simply have a pair of letter label and an 
instantaneous timed program,
e.g., $(a, \prog)$ in Figure~\ref{fig:ta-transition} (right).

The semantics for programs on a transition must generalize semantics for guards (defined
using satisfaction relation $\models$ above) and resets/release (defined using $[R]$
above).  But there is an obvious difference between these two: a guard may be crossed 
only if the valuation before the guard satisfies it, whereas a \emph{change} (reset or 
release) defines a relation between the valuations before and after the change.
To capture both in a uniform way, we define the semantics of programs as relations on
pairs of valuations.  Formally, for $v,v'\in \V$, $\prog\in\Prog$ we say $(v,v')\models
\prog$, more conveniently written as $v\xrightarrow{\prog} v'$, inductively:
\begin{itemize}
\item $v\xrightarrow{g}v'$ if $v\models g$ and $v'=v$,
\item $v\xrightarrow{[R]}v'$ if $v'\in [R]v$,
\item $v\xrightarrow{\prog_1;\prog_2} v'$ if $\exists v''\in \V$ such that $v\xrightarrow{\prog_1} v''$ and $v''\xrightarrow{\prog_2} v'$.
\end{itemize}

\subsection{Extending the automaton model}
Now, we have all the pieces necessary to define our generalized model.

\begin{definition}[\GTAfull]\label{defn:gta}
  A \emph{\GTA}\ $\A$ is given by a tuple
  $(Q, \Sigma, X, \Delta, (q_0, g_0), (Q_{f}, g_f))$, where
  \begin{itemize}
    \item $Q$ is a finite set of states,
    
    \item $\Sigma$ is a finite alphabet of actions,
    
    \item $X=X_F \uplus X_H$ is a set of clocks partitioned into future and history clocks,
    
    \item the initialization condition is a pair comprising of an initial state $q_0 \in Q$ 
    and an initial guard $g_0\in \Phi(X)$ which should be satisfied by initial valuations,
    
    \item similarly, the final condition is a pair comprising of a set of final states
    $Q_{f} \subseteq Q$ along with a final guard $g_f$ that must be satisfied by final
    valuations,
    
    \item $\Delta \subseteq (Q \times \Sigma \times \Prog \times Q)$ is a finite set of
    transitions.  $\Delta$ contains transitions of the form $(q, a, \prog, q')$, where $q$
    is the source state, $q'$ is the target state, $a$ is the action triggering the
    transition, and $\prog$ is the instantaneous timed program that is executed
    in sequence (from left to right) while firing the transition.
  \end{itemize}
\end{definition}

\begin{definition}[Semantics of \GTA]\label{defn:gta-semantics}
  The semantics of a \GTA\
  $\A=(Q, \Sigma, X, \Delta, (q_0, g_0), (Q_f, g_f))$ 
  is given by a transition system $\TS_{\A}$ whose states are \emph{configurations}
  $(q,v)$ of $\A$, where $q \in Q$ and $v\in\V$ is a valuation.
  \begin{itemize}
    \item   A configuration $(q,v)$ is initial if $q=q_{0}$ and $v\models g_0$.
      \item A configuration $(q,v)$ is accepting if $q\in Q_{f}$ and $v\models g_f$.
      \item   Transitions of $\TS_{\A}$ are of two forms:
  \begin{itemize}
    \item \textbf{Delay transition}:
    $(q,v) \xra{\delta} (q, v + \delta)$ if $(v + \delta) \models X_{F} \leq 0$.
    
    \item \textbf{Discrete transition}: $(q,v) \xra{t} (q',v')$ if 
    $t=(q,a,\prog,q')\in\Delta$ and $v\xrightarrow{\prog} v'$.
    \end{itemize}
  \end{itemize}
\end{definition}

Thus, a discrete transition $t=(q,a,\prog,q')$, where $\prog=\prog_1;\ldots;\prog_n$ can
be taken from $(q,v)$ if there are valuations $v_1,\ldots v_n$ such that
$v\xrightarrow{\prog_1}v_1 \xrightarrow{\prog_2}v_2 \ldots \xrightarrow{\prog_n} v_n=v'$.

A \emph{run} of a \GTA\ is a finite sequence of transitions from an initial configuration of $\TS_{\A}$.   A run is said to be \emph{accepting} if its last configuration is accepting.  

\section{Expressivity of \GTA\ and examples}\label{sec:examples}

The \GTA\ model defined above is rather expressive. Figure~\ref{fig:an-bn}
illustrates an example which accepts words of the form $a^n b^m$ with $m \le n$, where
each $a$ occurs at time $0$, after which $b$'s are seen one by one, with distance $1$
between them.  The history clock $x$ is used to ensure the timing constraint.  For every
$a$ that is read, the future clocks $y, z$ decrease by $1$.  Hence the future clocks $y,
z$ maintain the opposite of the number of $a$'s seen.  When the automaton starts reading
$b$, the future clocks also start elapsing time and since they cannot go above $0$, the
number of $b$'s is at most the number of $a$'s.  Such a language cannot be accepted by
timed automata since the untimed language obtained by removing the time stamps needs to be
regular in the case of timed automata.  The \GTA\ model is not only expressive, it is also
convenient for use.  To see this we now show that three classical models of timed systems
can be easily captured using \GTA. We also illustrate the modeling convenience provided by
\GTA\ in Section~\ref{sec:experiments} based on experiments.

\begin{figure}[htbp]
    \centering
      \scalebox{1}{
      \begin{tikzpicture}[state/.style={circle, draw, inner sep=5pt,
        thick}]
      \begin{scope}[every node/.style={state}]
        \node (0) at (-1,0) {\large \bf $q_0$};
        \node [double] (1) at (2,0) {\large \bf $q_1$};
      \end{scope}
      \begin{scope}[->, >=stealth, thick, auto]
        \draw (-1.75, 0) to (0);
        \draw (0) to [loop above] node {\large \bf $a, \prog_1$} (1);
        \draw (0) to node {\large \bf $b, \prog_2$} (1);
        \draw (1) to [loop above] node {\large \bf  $b,
          \prog_2$} (1); 
      \end{scope}
  
      \node [right] at (4, 1) {\large \bf History clocks: $\{x\}$, \quad Future clocks: $\{y, z \}$};
      \node [right] at (4, 0.5) {\large \bf $\prog_1: \langle x = 0; %~[x];
       ~[y]; ~y = z - 1; ~[z]; z = y \rangle$ };
      \node [right] at (4, 0) {\large \bf $\prog_2: \langle x = 1;
        ~[x] \rangle$};
      \node [right] at (4, -0.5) {\large \bf Initial condition:
        $y = z = 0$}; 
      \node [right] at (4, -1) {\large \bf Final condition:
        \emph{true} };
    \end{tikzpicture}
      }
    \caption{Example of a \GTA}
    \label{fig:an-bn}
  \end{figure}

\subsection{Timed automata}

Timed automata (TA) of Alur-Dill~\cite{AD94} can be modeled as a \GTA\ as follows:
\begin{itemize}
  \item The set of states of the \GTA\ is the same as the set of states of the TA.
  
  \item There are no future clocks in the \GTA\ and its history clocks are the clocks of the TA.

  \item Each transition of the form $q\xrightarrow{a,g,R} q'$ in a TA , where $g$ is a guard, $a$ a 
  letter and $R$ a subset of clocks to be reset, is replaced by a transition
  $q \xrightarrow{a, \prog} q'$ where $\prog= \langle g;[R] \rangle$.
  
  \item Initially, all clocks must be 0, captured by setting 
  $g_{0}=(X_{H}=0)$. 
  \item The final guard is empty: $g_f=\mathtt{True}$.
\end{itemize}

\subsection{Event-clock automata}
Event clock automata (ECA) of~\cite{AFH99} can be modeled as a \GTA\ as follows:
\begin{itemize}
  \item The set of states of the \GTA\ is the same as the set of states of the ECA.

  \item For each $a\in \Sigma$, the \GTA\ has a history clock $\history{a}$ and a future clock $\prophecy{a}$.
  
  \item Each transition of the form $q\xrightarrow{a,g} q'$ in a ECA, where $g$ is a
  guard of the ECA, $a$ a letter, is replaced by a transition $q \xrightarrow{a,\prog} q'$
  where $\prog:= \langle(\prophecy{a}=0); [\prophecy{a}]; g; [\history{a}]\rangle$.

  \item At initialization, history clocks must be undefined (set to $\infty$),
  captured by $g_{0}=(X_{H}=\infty)$. 
 
  \item At acceptance, all future clocks must be undefined, i.e., $g_{f}=(X_{F}=-\infty)$.
\end{itemize}

\subsection{Automata with timers}
The third model we consider is that of automata with timers.  Timers are timing constructs
that are started/intialized with a certain time value at some point/event and {\em count
down} to 0.  They measure the time from when they were started till the timer hits 0,
where the event of hitting 0 being called a {\em time-out} event.
However, they can be stopped using a {\em stop} event at any intermediate point instead and in which case the timer must be freed for reuse later.
Timers are a common construct in protocol specification, e.g., the ITU standard which uses timers rather than clocks~\cite{ITU} and
Mealy machines with timers~\cite{MMT}.   

In our setting, a timer can be seen as a specific instance of a future clock.  More
precisely Automata with timers ($\Atim$) can be modeled as \GTA\ as follows: \begin{itemize}
  \item The set of states of the \GTA\ is the same as the set of states of $\Atim$.

  \item The future clocks of \GTA\ are the timers of $\Atim$ and there are no history clocks.
  
  \item A transition of $\Atim$ with action $a$ from $q$ to $q'$ is encoded as a 
  $q\xrightarrow{a,\prog}q'$ where $\prog$ is defined as follows:
  \begin{itemize}
    \item if the transition starts timer $x$ with value $c\in\Rpos$, then 
    $\prog=\langle x=-\infty;[x];x=-c \rangle$.

    \item if the transition is guarded by $\mathtt{timeout}(x)$, then $\prog=\langle
    x=0;[x];x=-\infty \rangle$.
    
    \item  if the transition stops timer $x$, then $\prog=\langle [x];x=-\infty \rangle$.
  \end{itemize}
  \item Initially, the timers are undefined, captured by $g_{0}=(X_{F}=-\infty)$ and the final guard is empty, i.e., $g_{f}=\mathtt{True}$.  
\end{itemize}

We note that the timer above differs from a prophecy-event-clock (of ECA) though both are
future clocks.  Prophecy-clocks are released only when the event is seen, so at that point
the value of the prophecy-clock must be 0.  On the other hand timers can be stopped and
released even when their value is not 0.  This subtle difference has a surprising impact
when we allow diagonal guards as we will see shortly.

\section{The reachability problem for \GTA}
\label{sec:reachability}

We are interested in the \emph{reachability problem} for \GTA.  Formally, 

\begin{definition}[Reachability problem for \GTA]\label{defn:reach-problem}
  The reachability problem for a \GTA\ $\A$ is to decide whether $\A$ has an accepting run.
\end{definition}

For normal TA, the reachability problem is decidable and PSPACE complete as shown
in~\cite{AD94}.  This was shown using the so-called region abstraction, by proving the
existence of a finite time-abstract bisimulation.  However, this is not the case for \GTA.
As explained in the previous subsection, \GTA\ capture ECA, and as shown
in~\cite{GeeraertsRS11,GeeraertsRS14}, there exists ECA for which there is no finite
time-abstract bisimulation.  However, reachability is still decidable in the specific case
of ECA, as again shown in~\cite{AFH99}.
We note that for ECA model of~\cite{GeeraertsRS11,GeeraertsRS14} there are no diagonal
constraints.  In this case they show decidability via zone-extrapolation.
In~\cite{Concur22}, another approach for decidability via zone simulations is shown.  But
again even in this model diagonal constraints are disallowed.  Even more critically in
\GTA, we can capture timers and a priori we can have diagonal constraints even among
timers.  So, the question we ask is whether reachability is still decidable for \GTA.
Surprisingly, the answer is no.  The intuition is that with future clocks and diagonal
constraints, we get the ability to count (cf.\ Figure~\ref{fig:an-bn}).

\begin{theorem}
    Reachability for \GTA\ is undecidable.
\end{theorem}

\begin{proof}%[sketch]
  We will do a reduction from counter machines.  Given a counter machine, we will build a
  \GTA\ with one future clock $y_{C}$ for each counter $C$ and one extra future clock $z$.
  The reduction uses diagonal constraints between $z$ and the future clocks $y_{C}$.
  
  Initially and after each transition, the value of the future clock $z$ will be $0$. 
  Since a future clock has to be non-positive, time elapse is impossible. As an 
  invariant, the value of the future clock $y_{C}$ is the opposite of the value of counter
  $C$. The operations on counter $C$ are encoded with the following programs: \begin{align*}
    \mathtt{zero}_{C} & = \langle y_{C}=0 \rangle \\
    \mathtt{inc}_{C} & = \langle [z]; z=y_{C}-1; [y_{C}]; y_{C}=z; [z]; z=0 \rangle \\
    \mathtt{dec}_{C} & = \langle y_{C}\leq -1; [z]; z=y_{C}+1; [y_{C}]; y_{C}=z; [z]; z=0 \rangle 
  \end{align*}
  In the programs $\mathtt{inc}_{C}$ and $\mathtt{dec}_{C}$, each release of a future 
  clock is followed by a constraint which restricts the value non-deterministically 
  chosen during the release. For instance, $[z]; z=y_{C}-1$ is equivalent to 
  $z:=y_{C}-1$. Hence, the overall effect of $\mathtt{inc}_{C}$ is $y_{C}:=y_{C}-1$, 
  maintaining all other clocks unchanged, including the invariant $z=0$.
\end{proof}

Given this negative result, what can we do?  A careful observation of the proof tells us
that it is the interplay between diagonal constraints and arbitrary releases of future
clocks that leads to undecidability.  More precisely, the encoding depends on the
fact that clocks $z$ and $y_{C}$ which are used in diagonal constraints ($z=y_{C}-1$,
$z=y_{C}+1$ and $y_{C}=z$) may have arbitrary values when they are released.
This suggests a restricted subclass that we formalize next.

\subsection{$X_D$-Safe \GTA}

\begin{definition}~\label{def:safe-program}
  Let $X_D\subseteq X_F$ be a subset of future clocks.
  
  A program $\prog = \langle g_{1};[R_{1}]; g_{2}; [R_{2}];\ldots; g_{k}; [R_{k}]; g_{k+1}
  \rangle$ is $X_{D}$-safe if 
  \begin{itemize}
    \item  diagonal constraints between future clocks are restricted to clocks in $X_{D}$:
    if $x-y\leqlt c$ with $x,y\in X_{F}$ occurs in some $g_{i}$ then $x,y\in X_{D}$;
  
    \item  clocks in $X_{D}$ should be $0$ or $-\infty$ before being released:
    if $x\in X_{D}\cap R_{i}$ then $x=0$ or $x=-\infty$ occurs in $g_{i}$.
  \end{itemize}
  A \GTA\ $\Aa$ is $X_{D}$-safe if it only uses $X_{D}$-safe programs on its transitions
  and the initial guard $g_0$ sets each history clock to either $0$ or $\infty$.
\end{definition}

\begin{proposition}Timed automata, ECA (possibly with diagonal constraints) and Automata with timers (but without diagonal constraints) can all be captured by $X_D$-safe \GTA.\end{proposition}
\begin{proof}
  The proof follows by observing that in all the three cases, the safety condition holds.
  Timed automata do not have future clocks so the condition is vacuously true.  In ECA,
  event-predicting clocks are always checked for 0 before being released, hence they are
  safe as well with $X_D=X_{F}$.  Automata with timers without diagonal constraints are also
  trivially safe with $X_D=\emptyset$.  
\end{proof}

The importance of safety is the following theorem which is the center-piece of this
article.

\begin{theorem}
    Reachability for $X_D$-safe \GTA\ is decidable.
\end{theorem}

We will establish this theorem by showing a finite, sound and complete zone based
reachability algorithm for $X_D$-safe \GTA. If the given \GTA\ is not $X_D$-safe, then we
lose proof of termination (unsurprisingly, since the problem is undecidable), but we still
maintain soundness.  Thus, even for such \GTA\ when our algorithm does terminate it will give
 the correct answer.

\section{Symbolic enumeration}\label{sec:simulation}

We adapt the $\Gg$-simulation framework presented in~\cite{GastinMS20} for timed automata
with diagonal constraints to \GTA. Diagonal constraints offer succinct
modeling~\cite{Chevalier-Bouyer},
 but are quite challenging to handle efficiently in
zone-based algorithms, and have led to pitfalls in the past:~\cite{Bouyer04} showed that
the erstwhile algorithm based on zone-extrapolations that was implemented in tools is
incorrect for models with diagonal constraints; moreover no extrapolation based method can
work for automata with diagonal constraints.  The simulation framework by-passes this
impossibility result and is the state-of-the-art for timed automata with diagonal
constraints.  The framework was extended to event-clock automata without diagonal
constraints in~\cite{Concur22}.  We show that the ideas from~\cite{GastinMS20}
and~\cite{Concur22} can be suitably combined to give an effective procedure for safe
\GTAs.  This extension to \GTAs\ enables us to understand the mechanics of diagonal
constraints in future clocks.

The algorithm based on the $\Gg$-simulation framework involves:
\begin{enumerate}
  \item computation of a set of constraints at every state of the
    automaton by a \emph{static analysis} of the model,
  \item a symbolic enumeration using \emph{zones} to compute the
    \emph{zone graph},
  \item a \emph{simulation relation} between zones to ensure
    termination of the enumeration. 
\end{enumerate}

We will next adapt the static analysis to the \GTA\ setting.  The algorithm for the zone
graph computation and the implementation of the simulation relation over zones is taken
off-the-shelf from \cite{GastinMS20} and \cite{Concur22}, except for a minor adaptation to
include diagonal constraints involving future clocks.
What is absent, and requires a
non-trivial analysis, is the proof of termination.  Therefore, we will mainly focus on this
aspect 
and devote Section~\ref{sec:finiteness} for the termination argument.

\subsection{A concrete simulation relation for \GTA}\label{sec:concrete-simulation}

We fix a \GTA\ $\A=(Q, \Sigma, X, T, (q_0, g_0), (Q_{f}, g_{f}))$ for this section.  Our goal in this section is to define a simulation relation on the semantics of $\A$, i.e., on $\TS(\A)$. In the subsequent sections we will lift this to zones and show its finiteness. 
A simulation relation on $\TS(\A)$ is a reflexive,
  transitive relation $(q, v) \preceq (q, v')$ relating configurations
  with the same control state and (1) for every
  $(q, v) \xra{\delta} (q, v+\delta)$, we have
  $(q, v') \xra{\delta} (q, v'+\delta)$ and
  $(q, v+\delta) \preceq (q, v'+ \delta)$, (2) for every transition
  $t$, if $(q, v) \xra{t} (q_1, v_1)$ for some valuation $v_{1}$, then
  $(q, v') \xra{t} (q_1, v'_1)$ for some valuation $v'_{1}$ 
  with $(q_1, v_1) \preceq (q_1, v'_1)$.

  For any set $G$ of atomic constraints, we define a {\em preorder $\preceq_{G}$} on valuations by
$$
v\preceq_{G}v' \qquad \text { if }
\forall\varphi\in G,~\forall\delta\geq0,\qquad v+\delta\models\varphi
\implies v'+\delta\models\varphi \,.
$$
Notice that in the definition above, we \emph{do not} restrict $\delta$ to those such that
$v+\delta$ is a valuation: we may have $v(x)+\delta>0$ for some $x \in X_{F}$.  In usual
timed automata, this question does not arise, as elapsing any $\delta$ from any given
valuation always results in a valuation.  But this is crucial for the proof of
Theorem~\ref{thm:simulation} below.

Intuitively, the preorder above is a simulation w.r.t. the constraints in $G$ even after time
elapse.  But we need this to also be a simulation w.r.t. discrete transitions.  To achieve this,
the set of constraints $G$ should depend on the available discrete transitions.  In fact, we define a map $\G$ from states to set of
constraints, in such a way that it captures the simulation w.r.t. the discrete actions.  In
other words, our focus will be to choose state-dependent sets of constraints (given by the
map $\G$) depending on $\A$ such that the resulting preorder induces a simulation on
$\TS(\A)$.

As a first step towards this, we define, for any set $G$ of constraints and any program 
$\prog$, a set of constraints $G'=\pre{\prog}{G}$ such that, if $v\preceq_{G'}v'$ and 
$v\xra{\prog}v_{1}$ then there exists $v'\xra{\prog}v'_{1}$ such that 
$v_{1}\preceq_{G}v'_{1}$.
This set is defined inductively as follows ($G$ is a set of atomic constraints, $R$ is a
set of clocks, $g$ is an \emph{arbitrary} constraint, $y-x\leqlt c$ is an \emph{atomic}
constraint):
\begin{align*}
  \pre{\prog_1;\prog_2}{G} &= \pre{\prog_1}{\pre{\prog_2}{G}} \\
  \pre{g}{G} &= \asplit(g)\cup G \\
  \pre{[R]}{G} &= \bigcup_{\varphi\in G}\pre{[R]}{\{\varphi\}} \\
  \pre{[R]}{\{y-x\leqlt c\}} &=
  \begin{cases}
    \{y-x \leqlt c\} & \text{if } x,y \notin R \\
    \{y \leqlt c\} & \text{if } x \in R, y \notin R \\
    \{-x \leqlt c\} & \text{if } x \notin R, y \in R \\
    \emptyset & \text{if } x,y \in R 
  \end{cases}
\end{align*}
where $\asplit(g)$ is the set of atomic constraints occurring in $g$.

Now, the choice of suitable $G$ will be obtained by static analysis, on the lines of what was done for timed automata with diagonals~\cite{GastinMS18,GastinMS19,GastinMS20}, but adapted to our more powerful model.
More precisely, we define the map $\G$ from $Q$ to sets of atomic constraints as the least fixpoint of the set of equations:
\begin{align}
  \label{eq:fixpt}
  \G(q) = \{ x \leq 0 \mid x \in X_{F}\} \cup \bigcup_{q\xrightarrow{a,\prog} q'} \pre{\prog}{\G(q')}
\end{align}

Finally, based on $\preceq_G$ and the $\Gg(q)$ computation, we can define a preorder
$\preceq_{\A}$ between configurations of $\TS(\A)$ as $(q,v)\preceq_{\A}(q',v')$ if
$q=q'$ and $v\preceq_{\G(q)}v'$.

We will need the following technical lemma.

\begin{lemma}~\label{lem:pre-computation-general}
  Let $G$ be a set of atomic constraints and $G'=\pre{[R]}{G}$.  Let $R$ be a set of
  clocks.  Let $v_{1},v_{2}\in\V$ be valuations and let $v'_1\in [R]v_1$ and $v'_2\in
  [R]v_2$ be such that $v'_{2}\da_{R} = v'_{1}\da_{R}$.  Then, $v_1 \preceq_{G'} v_2$ 
  implies $v'_1 \preceq_{G} v'_2$.
\end{lemma}

\begin{proof}
  Since $v'_1\in [R]v_1$ and $v'_2\in [R]v_2$, we have $v'_1\da_{X\setminus R} =
  v_1\da_{X\setminus R}$ and $v'_{2}\da_{X\setminus R} = v_{2}\da_{X\setminus R}$.
  Moreover, from our assumption, we have $v'_{2}\da_{R} = v'_{1}\da_{R}$.
  
  We have $v\preceq_{G}v'$ iff $v\preceq_{\varphi}v'$ for all $\varphi\in G$.
  Hence, it is sufficient to prove the lemma when $G=\{\varphi\}$ where $\varphi$ is an 
  atomic constraint $y - x \leqlt c$. Let $G'=\pre{[R]}{G}$, which is either $\emptyset$ 
  when $x,y\in R$ or a singleton $\{\varphi'\}$. 
  Suppose $v_{1}\preceq_{G'}v_{2}$.
  \begin{itemize}
    \item Suppose that $x,y\in R$.  In this case, we have $v'_2(y)=v'_1(y)$ and
    $v'_{2}(x)=v'_{1}(x)$.  We then have $v'_1 \preceq_{\varphi} v'_2$ (we do not need any 
    hypothesis on $v_{1},v_{2}$).

    \item Suppose that $x,y\notin R$.  In this case, we have $\varphi'=\varphi$,
    $v'_1(y)=v_1(y)$ and $v'_{2}(y)=v_{2}(y)$, $v'_1(x)=v_1(x)$ and $v'_{2}(x)=v_{2}(x)$.
    Since $v_1 \preceq_{\varphi} v_2$, it follows that $v'_1 \preceq_{\varphi} v'_2$.
  
    \item Suppose that $x\in R$ and $y\notin R$. 
    In this case, we have $v_1 \preceq_{\varphi'} v_2$ with $\varphi'= y \leqlt c$.
    We need to show that $v'_1 \preceq_{\varphi} v'_2$. Let $\delta\geq0$ and assume that
    $v'_{1}+\delta\models\varphi$, i.e., $v'_1(y)-v'_1(x)\leqlt c$.
    We have to show that $v'_2(y) - v'_2(x) \leqlt c$.
    
    We have $v'_1(y)=v_1(y)$, $v'_{2}(y)=v_{2}(y)$ and $v'_{2}(x)=v'_{1}(x)\leq 0$.
    Let $\delta'=-v'_{1}(x)\geq 0$. We get $v_{1}+\delta'\models\varphi'$.
    We deduce that $v_{2}+\delta'\models\varphi'$, i.e., $v'_2(y)-v'_2(x)\leqlt c$ as
    desired.
    
    \item Suppose that $y\in R$ and $x\notin R$. 
    The proof is symmetric to the case above, and proceeds by similar arguments. 
    In this case, we have $v_1\preceq_{\varphi'} v_2$ where $\varphi'= -x \leqlt c$.
    We need to show that $v'_1 \preceq_{\varphi} v'_2$. Let $\delta\geq0$ and assume that
    $v'_{1}+\delta\models\varphi$, i.e., $v'_1(y)-v'_1(x)\leqlt c$.
    We have to show that $v'_2(y)-v'_2(x)\leqlt c$.
    
    We have $v'_1(x)=v_1(x)$, $v'_{2}(x)=v_{2}(x)$ and $v'_{2}(y)=v'_{1}(y)\leq 0$.
    Let $\delta'=-v'_{1}(y)\geq 0$. We get $-(v_{1}+\delta')(x)\leqlt c$, i.e.,
    $v_{1}+\delta'\models\varphi'$.
    We deduce that $v_{2}+\delta'\models\varphi'$, i.e., $v'_2(y)-v'_2(x)\leqlt c$ as
    desired.
    \qedhere
  \end{itemize}
\end{proof}

Now, let us prove that $\preceq_\A$ defined above is indeed a simulation relation.
\begin{theorem}\label{thm:simulation}
  The relation $\preceq_{\A}$ is a simulation on the transition system $\TS_{\A}$ of \GTA\ $\A$.
\end{theorem}

\begin{proof}
  Assume that $(q,v_{1})\preceq_{\A}(q,v_1')$, i.e., $v_{1}\preceq_{\G(q)}v_1'$.
  \begin{description}
    \item[Delay transition]  Assume that $(q,v_{1})\xra{\delta}(q,v_{1}+\delta)$ is a transition of $\TS_{\A}$. 
    Then, $v_{1}+\delta\models X_{F} \leq 0$.
    Since $\G(q)$ contains $x \leq0$ for all $x\in X_{F}$ and 
    $v_{1}\preceq_{\G(q)}v_{2}$, we deduce that 
    $v_1'+\delta\models X_{F}\leq0$.
    Therefore, $(q,v_1')\xra{\delta}(q,v_1'+\delta)$ is a transition in $\TS_{\A}$.
    It is easy to see that $v_{1}+\delta\preceq_{\G(q)}v_1'+\delta$.
  
    \item[Discrete transition]  
      Let $(q_1,v_1)\xra{a,\prog} (q,v)\in \TS_\A$ for some transition
      $t=(q_1,a,\prog,q)$ of $\A$.  Then we need to show that there exists $v'$ such that
      $(q,v)\preceq_\A(q,v')$ and $(q_1,v'_1)\xra{a,\prog} (q,v')$.  
      Wlog, we can assume that $\prog=\langle g_1;[R_1];\ldots g_k;[R_k]\rangle$ i.e., an alternating sequences of guards and changes (reset/release).  By
      definition, this means that there are $v_{2},\ldots,v_{k+1}$ with $v_{k+1}=v$ and
      $v_i\xrightarrow{g_i;[R_i]}v_{i+1}$ for all $1\leq i\leq k$.  This means that for 
      all $1\leq i\leq k$ we have $v_{i}\models g_{i}$ and $v_{i+1}\in[R_{i}]v_{i}$.

      Define $G_{k+1}=\G(q)$ and $G_{i}=\pre{\langle g_{i};[R_{i}] \rangle}{G_{i+1}}$ for
      $1\leq i\leq k$ so that $G_{1}=\G(q_{1})$.  Now, for each $1\leq i\leq k$ we
      construct below by induction valuations $v'_{2},\ldots,v'_{k+1}$ such that
      $v'_i\xrightarrow{g_i;[R_i]}v'_{i+1}$ and $v_{i+1}\preceq_{G_{i+1}}v'_{i+1}$.
      With $v'=v'_{k+1}$ we get $(q_1,v'_1)\xra{a,\prog} (q,v')$ and 
      $(q,v)\preceq_\A(q,v')$ as desired.
      
      For $i=1$, we have $v_{1}\preceq_{G_{1}}v'_{1}$ by hypothesis.
      Now, assume that $v_{i}\preceq_{G_{i}}v'_{i}$ for some $1\leq i\leq k$.
      Since $\asplit(g_{i})\subseteq G_{i}$ and $v_{i}\models g_{i}$, we deduce that 
      $v'_{i}\models g_{i}$. Now, let $v'_{i+1}$ be defined by 
      $v'_{i+1}\da_{R_{i}}=v_{i+1}\da_{R_{i}}$ and 
      $v'_{i+1}\da_{X\setminus R_{i}}=v'_{i}\da_{X\setminus R_{i}}$.
      We have $v'_{i+1}\in[R_{i}]v'_{i}$ and since $v'_{i}\models g_{i}$
      we deduce that $v'_i\xrightarrow{g_i;[R_i]}v'_{i+1}$.
      Notice that $\pre{[R_{i}]}{G_{i+1}}\subseteq G_{i}$.
      Hence, using Lemma~\ref{lem:pre-computation-general},
      $v'_{i+1}\da_{R_{i}}=v_{i+1}\da_{R_{i}}$ and $v_{i}\preceq_{G_{i}}v'_{i}$ we can
      conclude that $v_{i+1}\preceq_{G_{i+1}}v'_{i+1}$,
      which completes the proof.  \qedhere
  \end{description}
\end{proof}

\subsection{Zones for \GTA\ and a symbolic reachability algorithm}\label{sec:gta-zones}
The most widely used approach for checking reachability in a timed automaton (and more recently in event-clock automata) is based on reachability in a graph called the \emph{zone graph} of a timed automaton~\cite{Daws}.
Roughly, \emph{zones}~\cite{Bengtsson:LCPN:2003} are sets of valuations that can be represented efficiently using constraints between differences of clocks.  
In this section, we introduce an analogous notion for \GTAfull.  We consider \emph{\GTA\ zones}, which are special sets of valuations of \GTAfull.

\begin{definition}[\textbf{\GTA\ zones}]\label{defn:gen-zones}
  A \GTA\ zone is a set of valuations satisfying a conjunction of constraints of the form
  $y-x \leqlt c$, where $x,y \in X\cup\{0\}$,
  $c\in\overline{\mathbb{Z}}$ and
  ${\leqlt}\in\{\leq,<\}$.
\end{definition}

Thus zones are an abstract representation of sets of valuations. Then, an abstract configuration, also called a {\em node}, is a pair consisting of a state and a zone. Firing a transition $t:= (q,a,\prog,q')$ in a \GTA\  $\A$ from node $(q,Z)$ will result in another node following a sequence of operations that we now define.

\begin{definition}[\textbf{Operations on \GTA\ zones}]
  \label{def:oponzones}
  Let $g$ be a guard, $R\subseteq X$ be a set of clocks and $Z$ be a \GTA\ zone.
  \begin{itemize}
  \item Guard intersection:
    $Z \wedge g := \{v \mid v \in Z \text{ and } v \models g\}$

  \item Release/Reset: $[R]Z= \bigcup_{v\in Z} [R]v$ (as defined in 
  Section~\ref{sec:prelims})

  \item Time elapse:
    $\elapse{Z} = \{v+\d \mid v\in Z, \d\in\Rpos \text{ s.t.\ }
    v+\d\models (X_{F} \leq 0) \}$
  \end{itemize}
\end{definition}
From the above definition, it is easy to see that starting from a \GTA\ zone $Z$, the successors after the above operations are also \GTA\ zones. A guard $g$ can be seen as yet another \GTA\ zone and hence guard intersection is just an intersection operation between two \GTA\ zones. 
Similarly, the change operation preserves \GTA\ zones.  Finally, as is usual with timed automata, zones are closed under the time elapse operation.

Thus, for a transition $t:=(q,a,\prog,q')$ and a node $(q,Z)$, we can define the
successor node $(q',Z')$, and we write $(q, Z) \xra{t} (q',Z')$, where $Z'$ is the zone computed by the following sequence of operations: Let $\prog=\prog_1;\ldots;\prog_n$, where each $\prog_i$ is an atomic program, i.e., a guard $g$ or a change $R\subseteq X$. Then
we define zones $Z_1,\ldots, Z_{n+1}$ where, $Z_1=Z$, $Z'=\elapse{Z_{n+1}}$, and for each $1\leq i\leq n$,
\[Z_{i+1}= \begin{cases}
  Z_i\land \prog_i & \text{ if }\prog_i \text{ is a guard}\\
  [\prog_i]Z_{i} & \text{ if }\prog_i \text{ is a change}\\
\end{cases}
\]

Now, we can lift zone graphs, simulations from TA to \GTA\ and obtain a symbolic reachability algorithm for \GTA.

\begin{definition}[\textbf{\GTA\ zone graph}]
  Given a \GTA\ $\A$, its {\em \GTA\ zone graph, denoted \gzg($\A$)}, is defined as follows:
  Nodes are of the form $(q, Z)$ where $q$ is a state and $Z$ is a \GTA\ zone.  
  The initial node is $(q_0, \elapse{Z_0})$ where $q_0$ is the initial state and $Z_0$ is
  given by $g_{0}\wedge\big(X_{F} \leq 0\big)\wedge \big(X_H \geq 0\big)$
  ($Z_{0}$ is the set of all valuations which satisfy the \emph{initial constraint} $g_0$).
  For every node $(q, Z)$ and every
  transition $t := (q, a, \prog, q')$ there is a transition $(q, Z) \xra{t} (q', Z')$ in the \GTA\ zone graph.
  A node $(q,Z)$ is accepting if $q\in Q_f$ and $Z\cap g_{f}$ is non-empty, 
  i.e., there exists a valuation in $Z$ satisfying the final constraint.
\end{definition}

Similar to the case of zone graphs for timed automata and event zone graphs for event-clock automata, the \GTA\ zone graph can be used to decide reachability for \GTAfull. A node $(q,Z)$ is said to be
reachable (in $\A)$ if there is a path from the initial node
$(q_0,\elapse{Z_0})$ to $(q,Z)$ in $\gzg(A)$. Thus, reachability of a
final state in $\A$ reduces to checking reachability of an accepting
node in $\gzg(A)$. However, as in the case of zone graphs for timed
automata, $\gzg(A)$ is also not guaranteed to be finite. Hence, we
need to compute a finite truncation of the \GTA\ zone graph, which is
still sound and complete for reachability.

\begin{definition}[\textbf{Simulation on \GTA\ zones and finiteness}]
  Let $\preceq$ be a simulation relation on $\TS(\A)$. For two \GTA\
  zones $Z, Z'$, we say $(q, Z) \preceq (q, Z')$ if for every
  $v \in Z$ there exists $v' \in Z'$ such that
  $(q, v) \preceq (q, v')$. The simulation $\preceq$ is said to be
  finite if for every sequence $(q, Z_1), (q, Z_2), \dots$ of
  \emph{reachable} nodes, there exists $j > i$ such that
  $(q, Z_j) \preceq (q, Z_i)$.
\end{definition}

Now, the reachability algorithm, as in TA, enumerates the nodes of the \GTA\ zone graph and uses $\preceq$ to truncate nodes that are smaller with respect to the simulation.

\begin{definition}[\textbf{Reachability algorithm}] \label{def:reach-algo}
  Let $\Aa$ be a \GTA\ and $\preceq$ a simulation relation on $\TS(\A)$. 
  Add the initial node of the \GTA\ zone graph $(q_0, \elapse{Z_0})$ to a Waiting list. 
  Repeat the following until Waiting list is empty:
  \begin{itemize}
    \item Pop a node $(q, Z)$ from the Waiting list and add it to the
    Passed list.
    \item For every $(q, Z) \xra{t} (q_1, Z_1)$: if there exists a
    $(q_1, Z'_1)$ in the Passed or Waiting lists such that
    $(q_1, Z_1) \preceq (q_1, Z'_1)$, discard $(q_1, Z_1)$; else add
    $(q_1, Z_1)$ to the Waiting list.
  \end{itemize}
  If some accepting node is reached, the algorithm terminates and returns a Yes.
  Else, it continues until there are no further nodes to be explored and returns a No answer.
\end{definition}

The correctness of the above algorithm, follows from the correctness of the simulation approach in timed automata, with termination guaranteed when the simulation used is finite~\cite{HerbreteauSW12}. Thus, the following theorem is a straightforward adaptation of the corresponding proof~\cite{Daws,HerbreteauSW12} from timed automata.

\begin{theorem} Given $\GTA$ $\A$,
  \begin{enumerate}
  \item $\gzg(A)$ is sound and complete for reachability, i.e., an accepting node is reachability in $\gzg(A)$ iff an accepting state is reachable in $\A$.
  \item $\A$ has an accepting run iff the reachability algorithm returns Yes.
  \item The reachability algorithm is guaranteed to terminate, if the simulation $\preceq$ used is finite.
    \end{enumerate}
\end{theorem}

For the simulation, we will use $\preceq_\A$ as defined in the previous section. But we still need to show that it is finite. This is the hardest and most technical part of this paper and will form the bulk of Sections~\ref{sec:dagger},~\ref{sec:finiteness}. Before that, in Section~\ref{sec:DBM}, we first address the question of implementability of the above algorithm for $\GTA$ and the data structures needed for it, in particular the notion of distance graphs. Importantly, the properties that we show on the distance graphs will also be used in showing finiteness later.

\section{Computing with \GTA\ zones using distance graphs}\label{sec:DBM}

One of the main innovations towards implementability of timed automata was the development of Difference-Bound-Matrices (DBMs) as efficient data structures to represent and manipulate zones~\cite{Bengtsson:LCPN:2003}. For this, the central step was to view zones as {\em distance graphs} that could be immediately represented as DBMs. To have a practical implementation of \GTA, and to use the vast repertoire of existing tools and techniques for DBMs, a vital step is to be able to encode \GTA\ zones as distance graphs. 

We now show that \GTA\ zones can be represented using Difference-Bound-Matrices (DBMs) and the operations required for the reachability algorithm can be implemented using DBMs. 
The first hurdle is that for normal timed automata, each edge of the distance graph (i.e., entry in a DBM) encodes a constraint of the form $x - y \leqlt c$, where edges/entries are $(<,\infty)$ or $(\leqlt, c)$ with $c \in \mathbb{R}$ and ${\leqlt} \in \{ <, \leq \}$.  But for \GTA, we need to deal with valuations $+\infty$ or $-\infty$. 
For this purpose, we use more general weights as introduced in Definition~\ref{def:weights} and we extend the algebra of weights to the new entries in a natural way.
Before discussing the representation of \GTA\ zones as distance graphs, we briefly recall the extended algebra and some of the results we will use.  

\begin{definition}[\cite{eca-simulations-arxiv}][Order and sum of weights]\label{def:order-sum-weights}
  Let $(\leqlt,c), (\leqlt',c')\in\mathcal{C}$ be weights.
  
  \noindent\textbf{Order.} Define
  $(\leqlt, c) < (\leqlt', c')$ when either (1) $c < c'$,
  or (2) $c = c'$ and $\leqlt$ is $<$ while $\leqlt'$ is $\leq$. 
  This is a total order with
  $(<,-\infty) < (\le, -\infty) < (\leqlt, c) < (<, \infty) < (\le, \infty)$ 
  for all $c \in \mathbb{R}$.
 
  \smallskip\noindent\textbf{Sum.}
  We define the \emph{commutative} sum operation as follows.
  \begin{align*}
    (<,-\infty)+\alpha &= (<,-\infty) 
    &&\text{if } \alpha\in\mathcal{C}
    \\
    (\leq,\infty)+\alpha &= (\leq,\infty) 
    &&\text{if } \alpha\in\mathcal{C}\setminus\{(<,-\infty)\}
    \\
    (\leq,-\infty)+\alpha &= (\leq,-\infty) 
    &&\text{if } \alpha\in\mathcal{C}\setminus\{(<,-\infty),(\leq,\infty)\}
    \\
    (<,\infty)+\alpha &= (<,\infty)
    &&\text{if } \alpha\in\mathcal{C}\setminus\{(<,-\infty),(\leq,-\infty),(\leq,\infty)\}
    \\
    (\leqlt,c) + (\leqlt',c') &= (\leqlt'',c+c')
    &&\text{if } c,c'\in\mathbb{R} \text{ and } 
    {\leqlt''}={\leq} \text{ if } {\leqlt} = {\leqlt}' = {\leq}
    \text{ and } {\leqlt''}={<} \text{ otherwise.}
  \end{align*}
  Notice that sum of weights is an associative operation and $\alpha+(\leq,0)=\alpha$ for
  all $\alpha\in\mathcal{C}$.
\end{definition}

The intuition behind the above definition of order is that when $(\leqlt,c)<(\leqlt',c')$, the set of valuations that satisfies a constraint $x - y \leqlt c$ is contained in the solution set of $x - y \leqlt' c'$.  For the sum, the following lemma gives the idea behind our choice of definition.
 
\begin{lemma}[\cite{eca-simulations-arxiv}]\label{lem:sum-weights}
  Let $x, y, z\in X\cup\{0\}$ be clocks, $(\leqlt_1, c_1), (\leqlt_2, c_2)\in\mathcal{C}$ be weights and $(\leqlt, c) = (\leqlt_1, c_1) + (\leqlt_2, c_2)$.
  For all valuations $v\in\V$, if $v\models y-x \leqlt_1 c_1$ and $v\models z-y \leqlt_2 c_2$, then $v\models z-x \leqlt c$.
\end{lemma}

Equipped with the weights and the arithmetic over it, we can now define the representation of \GTA\ zones as distance graphs.  

\subsection{Distance graphs over the extended algebra}

\begin{definition}[Distance graphs]\label{def:distance-graph}
  A distance graph $\GG$ is a weighted directed graph without self-loops, with vertex set
  being $X\cup\{0\}=X_{F} \cup X_{H} \cup \{0\}$, edges being labeled with weights from
  $\mathcal{C}\setminus\{(<,-\infty)\}$.\footnote{
    If we allowed an edge with weight $\GG_{xy}=(<,-\infty)$ then we would get
    $\sem{\GG}=\emptyset$ since the constraint $y-x<-\infty$ is equivalent to false.}
  We define $\sem{\GG} := \{ v\in\V \mid v \models y - x \leqlt c \text{ for all edges } 
  x \xra{\leqlt\,c} y \text{ in } \GG \}$.
  The weight of edge $x\to y$ is denoted $\GG_{xy}$ and we set 
  $\GG_{xy}=(\leq,\infty)$ if there is no edge $x\to y$.
  The weight of a path is the sum of the weights of its edges.  A cycle in $\GG$ is said to be negative if its weight is strictly less than $(\le, 0)$.

  We say that $\GG$ is in standard form if it satisfies the following conditions;
  \begin{enumerate}
    \item $\GG_{0x}\leq(\leq,0)$ for all $x\in X_{F}$ and
    $\GG_{x0}\leq(\leq,0)$ for all $x\in X_{H}$.

    \item For all $x,y\in X$, if $\GG_{xy}\neq(\leq,\infty)$ 
    then $\GG_{x0}\neq(\leq,\infty)$ and $\GG_{0y}\neq(\leq,\infty)$.
  \end{enumerate}
  
  We extend the order on weights to distance graphs pointwise:
  Let $\GG$, $\GG'$ be distance graphs, we write $\GG\leq\GG'$ when 
  $\GG_{xy}\leq\GG'_{xy}$ for all edges $x\to y$. Notice that this implies 
  $\sem{\GG}\subseteq\sem{\GG'}$.

\end{definition}

The intuition of the standard form comes from the semantics of a distance graph.  
In classical timed automata the distance graph has no negative cycles iff its semantics is non-empty.  However, for distance graphs over the extended algebra, we will see that this is true only when it is in standard form.
To illustrate the need for general form, we provide an example from~\cite{eca-simulations-arxiv} here.

  \begin{example}\cite{eca-simulations-arxiv}~\label{ex:standardization}
    Suppose that valuations are finite and that
  we have constraints: $y - x \le 1$ and $-y \le 2$. 
    From these constraints, we can infer 
    $-x\le 3$ just by adding the inequalities.
    If there was another constraint $x\le -4$, 
    we will get unsatisfiability. In the language of
    distance graphs, the two initial constraints correspond to edges $x \xra{\le
    1} y$ and $y \xra{\le 2} 0$. The derived constraint $-x\le 3$ is
  obtained as the edge $x \xra{\le 3} 0$. The constraint $x \le -4$
  corresponds to $0 \xra{\le -4} x$ and the unsatisfiability is witnessed by
  a negative cycle $0 \xra{\le -4} x \xra{\le 3} 0$. Basically, adding
  the weights of $x \to y$ and $y \to 0$, we get the strongest possible constraint
  about $x \to 0$ resulting from the two constraints $x \to y$ and $y
  \to 0$.
  
  This holds no more in the extended algebra, due to the fundamental 
  difference while adding weight $(\le, \infty)$.
  Consider the constraints $y-x\le 1$ and $0-y\leq\infty$, corresponding to the edges
  $x\xra{\leq 1}y$ and $y\xra{\leq\infty}0$ in the distance graph.
  Adding the two weights gives the edge $x\xra{\leq\infty}0$, corresponding to the 
  constraint $0-x\leq\infty$. This is \emph{not} the 
  strongest possible constraint on $x$ induced by the constraints
  $y-x\le 1$ and $0-y\leq\infty$. Indeed, a valuation $v$ with $v(x)=-\infty$ satisfies the 
  constraint $0-x\leq\infty$. But, any valuation $v$ satisfying $y-x\le 1$ should have
  $v(x) \neq -\infty$, irrespective of the value of $v(y)$. 
  Now, if we also had a constraint $x-0\le -\infty$, corresponding to edge
  $0\xra{\leq-\infty}x$, we have no negative cycle in the corresponding distance graph.  
  But the set of constraints $y - x \le 1$, $0 - y \le \infty$ and $x-0 \le -\infty$ is not
  feasible.  In order to get a correspondence between negative cycles and empty solution
  sets, we propose the standard form. The standard form 
  equips the graph with the additional information that when $y - x$ is
  bounded by a finite value, that is, edge $x \to y$ does not have weight
  $(\le, \infty)$, the constraint $0 - x$ (and also $y-0$) is at most $(<,\infty)$. 
  With this information, we get negative cycles whenever there is a contradiction.
  \end{example}

Fortunately, it turns out that each distance graph $\GG$ can be transformed into an equivalent distance graph $\GG'$ which is in standard form. By equivalent, we mean $\sem{\GG}=\sem{\GG'}$.
First, we set $\GG'_{0x}=\min(\GG_{0x},(\leq,0))$ for $x\in X_{P}$ and
$\GG'_{x0}=\min(\GG_{x0},(\leq,0))$ for $x\in X_{H}$.
Moreover, if $x\in X_{P}$ then we set $\GG'_{x0}=\min(\GG_{x0},(<,\infty))$ if
$\GG_{xy}\neq(\leq,\infty)$ for some $y\neq x$, otherwise we keep $\GG'_{x0}=\GG_{x0}$.
Similarly, if $y\in X_{H}$ then we set $\GG'_{0y}=\min(\GG_{0y},(<,\infty))$ if
$\GG_{xy}\neq(\leq,\infty)$ for some $x\neq y$, otherwise we keep $\GG'_{0y}=\GG_{0y}$.
Finally, for $x,y\in X$ with $x\neq y$ we set $\GG'_{xy}=\GG_{xy}$. 
The graph $\GG'$ constructed above is called the standardization of $\GG$.

\begin{lemma}\label{lem:standard-form}~\cite{eca-simulations-arxiv}
  The standardization $\GG'$ of a distance graph $\GG$ is in standard form and
  $\sem{\GG'}=\sem{\GG}$. Moreover, $\sem{\GG}\neq\emptyset$ iff
  $\GG$ has no negative cycles.
\end{lemma}

  Now, suppose $\GG'$ (in standard form) has no negative cycles, then we construct $\GG''$ by replacing the weight of an edge $x\to y$ by the minimum of the weights of the paths from $x$ to $y$ in $\GG'$. 
  Such a $\GG''$ is called the \emph{normalization} of $\GG'$ and has several useful properties.

  \begin{lemma}[Normalization]\label{lem:normalization}~\cite{eca-simulations-arxiv}
    Let $\GG$ be a standard distance graph with no negative cycles.
    The normalization $\GG'$ of $\GG$ is normalized and $\sem{\GG}=\sem{\GG'}$.
  \end{lemma}

Let $Z$ be a nonempty zone.  Writing the constraints in $Z$ as a distance
graph, followed by standardizing and normalizing it, results in {\em its canonical distance graph $\graph{Z}$}: $\sem{\graph{Z}}=Z$ and $\graph{Z}$ is minimal among the standard graphs $G$ with $\sem{G}=Z$.
We denote by $Z_{xy}$ the weight of the edge $x\to y$ in $\graph{Z}$.
Formally,
  
\begin{lemma}\label{lem:canonical}~\cite{eca-simulations-arxiv}
  Let $\GG,\GG'$ be two distance graphs where $\GG$ is normalized. If 
  $\sem{\GG}\subseteq\sem{\GG'}$ then $\GG\leq\GG'$.   
  In particular, if both $\GG,\GG'$ are normalized and if $\sem{\GG}=\sem{\GG'}$ then 
  $\GG=\GG'$.
\end{lemma}

\subsection{Successor computation for \GTA\  zones}

Next, we show how we can perform \GTA\ zone operations on the respective distance graphs
of the zones.  Thanks to the algebra over the new weights, the arguments are very similar
to the cases for normal timed automata and event-clock
automata~\cite{Concur22,eca-simulations-arxiv}.  One important technical difference from
these earlier works is that due to the presence of diagonal constraints among future
clocks, after a guard intersection, we need to explicitly standardize the zone in order to
check its emptiness by looking for a negative cycle.  When we had only non-diagonal
guards, this was not necessary, as non-diagonal guards cannot change the weight of $x \to
y$ edges.

\begin{definition}[\textbf{Operations on distance graphs}]\label{defn:dist-graph-ops}
  Let $\GG$ be a normalized distance graph, let $g$ be a guard and let $R\subseteq X$ 
  be a set of clocks.
  \begin{itemize}
    \item Guard intersection: a distance graph $\GG_g$ is
    obtained from $\GG$ as follows,
    \begin{itemize}
      \item for each constraint $y - x \leqlt c$ in $g$, replace
      weight of edge $x \to y$ with $\min(\GG_{xy},(\leqlt, c))$,
      \item standardize the graph obtained in the above step,
      \item normalize the resulting graph if it has no negative cycles.
    \end{itemize}
    
    \item Release/Reset: a distance graph $[R]\GG$ is obtained
    from $\GG$ by 
    \begin{itemize}
      \item removing all edges involving clocks $x\in R$ and then
      \item adding the edges $0 \xra{(\leq,0)} x$ and $x \xra{(\leq,\infty)} 0$ for all
      $x\in R_{F}$, 
      \item adding the edges $0 \xra{(\leq,0)} x$ and $x \xra{(\leq, 0)} 0$ for all
      $x\in R_{H}$, and then
      \item normalizing the resulting graph.
    \end{itemize}
    
  \item Time elapse: the distance graph $\elapse{\GG}$ is
    obtained by the following transformation:
    \begin{itemize}
    \item for all history clocks $x$, if $\GG_{0x}\neq(\leq,\infty)$ then replace it with $(<,\infty)$,
      
    \item for all future clocks $x$, if $\GG_{0x}\neq(\leq,-\infty)$ then replace it with $(\leq,0)$,

    \item normalize the resulting graph.
    \end{itemize}
  \end{itemize}
\end{definition}

The theorem below says that the operations on \GTA\ zones translate easily to operations on distance graphs and that the successor of a \GTA\ zone is a \GTA\ zone. Except for the release operation $[R_{F}]\GG$, the rest of the operations are standard in timed automata, but they do not use weights $(\le, -\infty), (\le, +\infty)$.  
We can perform all these operations in the new algebra with quadratic complexity, matching the best-known complexity for timed
automata without diagonal constraints~\cite{DBLP:journals/ipl/ZhaoLZ05}.

\begin{theorem}\label{thm:distance-graph-operations}
  Let $\GG$ be a normalized distance graph,
  $g$ be a guard and $R\subseteq X$ be a set of clocks.  
  We can compute, in $\mathcal{O}(|X|^2)$ time, normalized distance graphs 
  $\GG_g$, $[R]\GG$ 
  and $\elapse{\GG}$, such that
  $\sem{\GG} \land g = \sem{\GG_g}$,
  $[R]\sem{\GG} = \sem{[R]\GG}$, and
  $\elapse{\sem{\GG}} = \sem{\elapse{\GG}}$.
\end{theorem}

The proof of the theorem follows from Lemmas~\ref{lem:guard},~\ref{lem:release-reset}
and~\ref{lem:time-elapse-prop} which we state below for completeness. 
We give only the proof of Lemma~\ref{lem:guard} which is a bit more involved than the non-diagonal case handled in~\cite{eca-simulations-arxiv}. 
The proofs of Lemma~\ref{lem:release-reset}
and~\ref{lem:time-elapse-prop} can be found in \cite{eca-simulations-arxiv}.

\begin{lemma}\label{lem:guard}
  Let $\GG$ be a normalized distance graph and let $g$ be a guard.  
  Then $\sem{\GG_g}=\sem{\GG}\wedge g$, 
  and $\GG_{g}$ can be computed in time 
    $\mathcal{O}(|g|+ |X|^{2})$.
\end{lemma}

\begin{proof}
  According to the definition, we first construct an intermediate graph $\GG'$ by replacing weights of edges of the form 
  $x \to y$ depending on the atomic constraints in $g$.  
  It is easy to see that $\sem{\GG'}=\sem{\GG}\wedge g$ 
  and that $\GG'$ is computed from $\GG$ in time $\mathcal{O}(|g|+|X|^{2})$.  The
  standardization process computes in time $\mathcal{O}(|X|^{2})$ a graph $\GG''$ is
  standard form with the same solution set. If $\GG''$ has no negative cycle,
  the normalization process does not change the solution set.
  
  In general, checking for negative cycle and normalization of $\GG''$ may take time
  $\mathcal{O}(|X|^{3})$.  Alternatively, we can start by handling the modification of
  non-diagonal edges as we did in the ECA paper: see below how to check for negative
  cycles and normalize in time $\mathcal{O}(|X|^{2})$.  Then, for each diagonal constraint
  $y-x\leqlt c$ in $g$, we reduce the weight of each edge $x'\to y'$ to
  $\min(\GG''_{x'y'},\GG''_{x'x}+(\leqlt ,c)+\GG''_{yy'})$.
\end{proof}

\begin{lemma}\label{lem:release-reset}
  Let $\GG$ be a normalized distance graph and $R\subseteq X$.  
  Then, $\sem{[R]\GG}=[R]\sem{\GG}$,
  and $[R]\GG$ can be computed in time $\mathcal{O}(|X|^{2})$.
  
  Moreover, the weight $\GG'_{xy}$ of edge $x\to y$ in $[R]\GG$ is given by
  $$
  \GG'_{xy}=
  \begin{cases}
    (\leq,\infty) & \text{if } x\in R_{F} \\
    (\leq,0) & \text{if } x\in R_{H}, y\in R \\
    \GG_{0y} & \text{if } x\in R_{H}, y\notin R \\
    \GG_{x0} & \text{if } x\notin R, y\in R \\
    \GG_{xy} & \text{if } x,y\notin R 
  \end{cases}
  $$
\end{lemma}

\begin{lemma}\label{lem:time-elapse-prop}~\cite{eca-simulations-arxiv}
  Let $\GG$ be a normalized distance graph.
  Then, $\sem{\elapse{\GG}}=\elapse{\sem{\GG}}$,
  and $\elapse{\GG}$ can be computed in time $\mathcal{O}(|X|^{2})$.
\end{lemma}

\section{Safely reachable \GTA\ zones and their properties}\label{sec:dagger}
Till now, we have shown properties of distance graphs for \GTA\ zones in general. In this section, we show that \GTA\ zones that are reachable from the initial zone in an $X_D$-safe \GTA\ have additional special properties. As in normal TA, we also use the fact the maximal constant occurring in the programs (in the transitions) of a \GTA. 

Let $M\in\mathbb{N}$.  We say that a constraint $x - y \leqlt c$ is \emph{$M$-bounded} if either $c \in \mathbb{R}$ is such that $-M\leq c\leq M$ or $(\leqlt,c) \in \{(\leq,-\infty),(<,\infty),(\leq,\infty)\}$. We say that a program is $M$-bounded if each of its constraints is $M$-bounded. Recall the definition of $X_{D}$-safe programs from Definition~\ref{def:safe-program}. We say that a program is $(X_{D},M)$-safe if it is both $M$-bounded and $X_{D}$-safe.

\begin{definition}[$(X_{D},M)$-safe operations]\label{def:safe-operations}
  The following zone operations are $(X_{D},M)$-safe:
  \begin{itemize}
    \item Guard intersection with a safe guard:
    $Z \wedge g$, where $g$ is $(X_{D},M)$-safe.

    \item Reset of a history clock $x$ or release of a future clock $x \notin X_{D}$: $[x]Z$, where $x \notin X_{D}$.
    
    \item Release of a future clock $x \in X_{D}$ when its value is $0$ or $-\infty$: $[x](Z \wedge (x = c))$, where $x \in X_{D}$ and $c \in \{0,-\infty\}$.

    \item Time elapse: $\elapse{Z}$.
  \end{itemize}  
\end{definition}

We say that a zone $Z$ is \emph{$(X_{D},M)$-safely reachable} if 
\begin{itemize}
  \item the initialization guard $g_0$ sets each history clock to either $0$ or $\infty$. 
  \item if $Z$ can be obtained starting from the initial zone $Z_{0}$ and applying only $(X_D,M)$-safe zone operations.
\end{itemize}

\begin{lemma}
  If $\Aa$ is an $X_{D}$-safe \GTA\ in which the maximum constant used is $M$, then its reachable zones are \emph{$(X_{D},M)$-safe}.   
\end{lemma}
In other words, for these systems, we need to only reason about \emph{$(X_{D},M)$-safely reachable} zones.
When $X_{D}$ and $M$ are clear from the context, we will sometimes abuse notation and just say \emph{safely reachable} zone instead of $(X_{D},M)$-safely reachable zone, and use \emph{safe} programs, constraints, accordingly.

\subsection{Valuations of safely reachable zones}

Next, we define an equivalence relation $\simeq$ between valuations that relates valuations that agree on value of history clocks, and satisfy the same set of safe constraints involving non-history clocks.
\begin{definition}\label{defn:eq-diagonals}
  $v_1 \simeq v_2$ if $v_1 \da_{X_H} = v_2 \da_{X_H}$ and, for all $x,y \in X_{F} \cup \{0\}$ and for all $(X_{D},M)$-safe constraints $y - x \leqlt c$, we have $v_{1} \models y - x \leqlt c$ if and only if $v_{2} \models y - x \leqlt c$.
\end{definition}

Let $x\in X_{F}$ and $v_{1}\simeq v_{2}$. 
Then $v_{1}(x)=-\infty$ iff $v_{2}(x)=-\infty$. 
This is because $x-0\leq-\infty$ is an $(X_{D},M)$-safe constraint.
Further, $v_{1}(x)=v_2(x)$ if $-M\leq v_1(x)\leq 0$.
This is because $x-0\leq v_{1}(x)$ and $0-x\leq -v_{1}(x)$ are both $(X_{D},M)$-safe.
It follows that $-\infty < v_1(x) < -M$ if and only if 
$-\infty < v_2(x) < -M$ as well.

Definition~\ref{defn:eq-diagonals} requires that $v_1$ and $v_2$ satisfy the
same set of $(X_{D},M)$-safe constraints involving non-history clocks.  
We will now show that if $v_1 \simeq v_2$, then $v_1$ and $v_2$ satisfy the same set of $(X_{D},M)$-safe constraints involving any pair of clocks.

\begin{lemma}\label{lem:simeq-all-clocks}
  If $v_1 \simeq v_2$ then, for all $x,y \in X \cup \{0\}$ and for all $(X_{D},M)$-safe constraints $y - x \leqlt c$, we have $v_{1} \models y - x \leqlt c$ if and only if
  $v_{2} \models y - x \leqlt c$.
\end{lemma}

\begin{proof}
  The claim follows from Definition~\ref{defn:eq-diagonals} for $(X_{D},M)$-safe constraints
  involving non-history clocks.  Since $v_1 \simeq v_2$ implies $v_1 \da_{X_H} = v_2
  \da_{X_H}$, the claim is easy to see for $(X_{D},M)$-safe constraints (in fact all
  safe constraints, not just $(X_{D},M)$-safe constraints) not involving future clocks.
  Finally, we consider $M$-bounded constraints involving a history clock $y$ and a future clock $x$.
  \begin{itemize}
    \item Suppose that $v_1(x) = v_2(x)$.
    In this case, it is easy to see that $v_1$ and $v_2$ satisfy the same constraints involving $y$ and $x$.  

    \item Suppose that $v_1(x) \neq v_2(x)$.
    Then, since $v_1 \simeq v_2$, this implies that $-\infty < v_1(x) < -M$ and $-\infty < v_2(x) < -M$.
    \begin{itemize}
      \item $y - x \leqlt c$.
      Suppose $v_1(y) = v_2(y) = \infty$.  
      Then, $v_{1}(y)-v_{1}(x)=\infty=v_{2}(y)-v_{2}(x)$.
      Otherwise, $0\leq v_1(y)=v_2(y)<\infty$. 
      Then, $M<v_{1}(y)-v_{1}(x)<\infty$ and 
      $M<v_{2}(y)-v_{2}(x)<\infty$. In both cases, we obtain
      $v_{1}\models y - x \leqlt c$ if and only if 
      $v_{2}\models y - x \leqlt c$.

      \item $x - y \leqlt c$.  We argue similarly, distinguishing two cases 
      depending on whether $v_1(y) = v_2(y)$ is finite or not.
      \qedhere
    \end{itemize}
  \end{itemize}
\end{proof}

We will now state a lemma which highlights an important property of future clocks in safely reachable \GTA\ zones - namely, that safely reachable zones are closed under $\simeq$-equivalence.  
The proof follows from the observation that the property is true in the initial zone, and is invariant under the zone operations.

\begin{lemma}\label{lem:FutureClocks-diagonals}
  For all safely reachable zones $Z$, if $v \in Z$ and $v \simeq v'$, then $v' \in Z$.
\end{lemma}

\begin{proof}
  We will prove that the statement of the lemma is an invariant over safely reachable zones.
  The property is true if $Z=\V$ is the set of all valuations.
  We now show that the property is invariant under all the safe zone operations given in
  Definition~\ref{def:safe-operations}.
  Notice that the initial zone is $Z_{0}=\elapse{\V\cap g_{0}}$ and $g_{0}$ is 
  safe.
  Assume that $Z$ is a zone that satisfies the property of the lemma.
  \begin{description}
    \item[Guard intersection.] Let $g$ be a guard, which is in general a
    conjunction of (possibly diagonal) $(X_{D},M)$-safe constraints.
    We get directly from Lemma~\ref{lem:simeq-all-clocks} that the property continues to
    hold in the zone $Z \land g$.
    
    \item[Release of a clock $x\in X_{F}\setminus X_{D}$.]  Let $v\in [x]Z$ and $v'\simeq v$.  
    We need to show that $v'\in [x]Z$.  By definition of the release operation, we have
    $v=u[x\mapsto\beta]$ for some $u\in Z$ and $-\infty \le \beta \le 0$. 
    Let $u'=v'[x\mapsto u(x)]$. 
    Since $Z$ is closed under $\simeq$-equivalence (by assumption), it suffices to show
    that $u'\simeq u$.  We then have $u'\in Z$ and $v'=u'[x\mapsto v'(x)]$, which implies
    $v'\in[x]Z$.
    
    First, we have $u\da_{X_{H}}=v\da_{X_{H}}=v'\da_{X_{H}}=u'\da_{X_{H}}$.
    Next, consider a safe constraint $y-z\leqlt c$ with $y,z\in X_{F}\cup\{0\}$.
    \begin{itemize}
      \item Suppose that $y,z \neq x$. 
      We have $u(y)-u(z)=v(y)-v(z)$ and $u'(y)-u'(z)=v'(y)-v'(z)$. 
      Using $v\simeq v'$, we deduce that
      $u\models y-z\leqlt c$ iff $v\models y-z\leqlt c$ iff
      $v'\models y-z\leqlt c$ iff $u'\models y-z\leqlt c$.

      \item Suppose $y=x\neq z$ (resp.\ $y\neq x=z$). 
      Since the constraint is $X_{D}$-safe and $x\in X_{F}\setminus X_{D}$ we deduce that 
      $z=0$ (resp.\ $y=0$). We have $u(x)=u'(x)$. We deduce that 
      $u\models y-z\leqlt c$ iff $u'\models y-z\leqlt c$.
    \end{itemize}
    
    \item[Release of a clock $x\in X_{D}$.]  Let $v\in[x](Z\wedge(x=a))$ with $a=0$ or 
    $a=-\infty$ and let $v'\simeq v$.  We need to show that $v'\in[x](Z\wedge(x=a))$.  
    Note that we have $u=v[x\mapsto a]\in Z$. Let $u'=v'[x\mapsto a]$.
    Since $Z$ is closed under $\simeq$-equivalence (by assumption), it suffices to show
    that $u'\simeq u$.  We then have $u'\in Z$ and we get $v'\in[x](Z\wedge(x=a))$.
    
    First, we have $u\da_{X_{H}}=v\da_{X_{H}}=v'\da_{X_{H}}=u'\da_{X_{H}}$.
    Next, consider a $M$-bounded constraint $y-z\leqlt c$ with $y,z\in X_{F}\cup\{0\}$.
    We proceed as above if $y,z \neq x$, or if $y=x$ and $z=0$, or if $y=0$ and $z=x$.
    \begin{itemize}
      \item Suppose $y\neq 0$ and $z=x$.
      We have $u(z)=u'(z)=a$, $u(y)=v(y)$ and $u'(y)=v'(y)$.
      We deduce that $u\models y-z\leqlt c$ iff $v(y)-a\leqlt c$ and
      $u'\models y-z\leqlt c$ iff $v'(y)-a\leqlt c$.
      Finally, we have $v(y)-a\leqlt c$ iff $v'(y)-a\leqlt c$.
      This is clear when $a=-\infty$ and it follows from $v\simeq v'$ when $a=0$ 
      ($y-0\leqlt c$ is a safe constraint).

      \item We proceed similarly when $y=x$ and $z\neq 0$.
      We have $u\models y-z\leqlt c$ iff $a-v(z)\leqlt c$ and
      $u'\models y-z\leqlt c$ iff $a-v'(z)\leqlt c$.
      Notice that $v(z)=-\infty$ iff $v'(z)=-\infty$ since $v\simeq v'$ and 
      $z\leq-\infty$ is a safe constraint. We deduce that $a-v(z)\leqlt c$ iff
      $a-v'(z)\leqlt c$.
    \end{itemize}
        
    \item[Reset.] The reset operation of a history clock $x$ takes each valuation in $Z$ and sets $x$ to $0$. 
    
    Let $v \in [x]Z$. 
    This implies that there exists $u \in Z$ such that $v = [x]u$.
    Then, $v(x) = 0$, and $u(y) = v(y)$ for all $y \neq x$. 
    
    Let $v' \simeq v$.  We need to show that $v'\in [x]Z$.
    Notice that $v'(x)=v(x)=0$.
    Let $u'=v'[x\mapsto u(x)]$. We have $v'=[x]u'$.    
    Since $Z$ is closed under $\simeq$-equivalence (by assumption), it suffices to show
    that $u' \simeq u$. We then have $u' \in Z$ and $v' = [x]u'\in [x]Z$.
    
    Since $v\da_{X_{H}}=v'\da_{X_{H}}$, we first get $u\da_{X_{H}}=u'\da_{X_{H}}$.
    It remains to show that $u$ satisfies an $(X_{D},M)$-safe constraint $y-z \leqlt c$ with
    $y,z\in X_{F}\cup\{0\}$ if and only if $u'$ also satisfies it.  
    Since $y,z\neq x$, we have $u(y)-u(z)=v(y)-v(z)$ and
    $u'(y)-u'(z)=v'(y)-v'(z)$. Using $v\simeq v'$, we deduce that
    $u\models y-z\leqlt c$ iff $v\models y-z\leqlt c$ iff
    $v'\models y-z\leqlt c$ iff $u'\models y-z\leqlt c$.
    
    \item[Time elapse.]  Time elapse increases the value of all clocks in $X$ in a
    synchronous manner, without affecting the differences between clocks in $X$.  We will
    now show that our property is not affected by time elapse.
    
    Suppose that $v \in \elapse{Z}$, i.e., $v = u + \delta$ for some $u \in Z$ and
    $\delta\geq0$.  Note that this means $v(z)=(u+\delta)(z) \le 0$ for all
    future clocks $z$.
    Let $v' \simeq v$. Take $u' = v' - \delta$. We show that $u' \in Z$, which implies
    $v'=u'+\delta\in\elapse{Z}$.
    
    Since $Z$ is closed under $\simeq$-equivalence (by assumption), it suffices to show
    that $u' \simeq u$.  
    Since $v\da_{X_{H}}=v'\da_{X_{H}}$, we first get $u\da_{X_{H}}=u'\da_{X_{H}}$.
    We consider the possible cases for a safe constraint $y - z \leqlt c$ with 
    $y,z\in X_{F}\cup\{0\}$.
    \begin{itemize}
      \item If $y,z\in X_{F}$. 
      We have $u(y)-u(z)=v(y)-v(z)$ and $u'(y)-u'(z)=v'(y)-v'(z)$. 
      Using $v\simeq v'$, we deduce that
      $u\models y-z\leqlt c$ iff $v\models y-z\leqlt c$ iff
      $v'\models y-z\leqlt c$ iff $u'\models y-z\leqlt c$.
      
      \item $y=0\neq z$ (the case where $z=0\neq y$ follows by a similar argument.)
      
      Suppose that $u\models 0-z\leqlt c$, i.e., $-u(z)\leqlt c$.
      Since $u=v-\delta$, we get $-v(z)\leqlt c-\delta$. Recall that $v(z)\leq 0$.
      Hence, we have $0\leq c-\delta \leq c \leq M$. 
      Further, since $v \simeq v'$ and $0-z\leqlt c-\delta$ is a safe constraints,
      we get $-v'(z)\leqlt c-\delta$. Using $v'=u'+\delta$, we get
      $-u'(z)\leqlt c$, i.e., $u'\models 0-z\leqlt c$.      
      \qedhere
    \end{itemize}
  \end{description}
\end{proof}

\begin{remark}
  The proof crucially uses the fact that
  $\A$ is $X_{D}$-safe.  For the case of releasing a clock $x\in X_{F}\setminus X_{D}$, we
  use the fact that a diagonal constraint involving $x$ may not use another future clock. 
  For the case of releasing a clock $x\in X_{D}$, we use the fact that the value of the 
  clock must be $0$ or $-\infty$ just before the release. 
  
  We remark that the claim does not hold for all zones (which could be reached by releasing a clock in $X_{D}$ when its value is not necessarily $0$ or $-\infty$).
  As a non-example, consider Figure~\ref{fig:an-bn}.  Here, $X_D = \{y, z\}$ and $M = 1$. 
  After two iterations of $a$, the zone $Z_2$ reached is $x = 0 \wedge y = z = -2$.  
  Pick $v: x = 0, y = z = -2$ and $v': x = 0, y = z = -3$.  Notice that both of them satisfy the same set of $(X_D, M)$-safe constraints, but $v \in Z_2$, $v' \notin Z_2$.  
  Indeed, the automaton is not $X_D$-safe since $y$ and $z$ are released arbitrarily.
\end{remark}

\begin{corollary}~\label{cor:FutureClocks-0-Diagonals}
  Let $Z$ be a $(X_{D},M)$-safely reachable zone and let $v \in Z$ be a valuation.
  Let $n=\max(1,|X_{D}|)$.
  \begin{enumerate}
    \item Let $x\in X_{F}\setminus X_{D}$.  If $-\infty < v({x}) < -M$ then, for every
    $-\infty < \alpha < -M$, then the valuation $v'=v[x\mapsto\alpha]$ belongs to $Z$.

    \item Let $x,y\in X_{D}\cup\{0\}$, if $-\infty < v({x}) - v({y}) < -nM$ then,
    for every $-\infty < \alpha < -nM$, we have a valuation $v'\in Z$ with 
    $v'({x}) - v'({y}) = \alpha$.
  
    \item Let $x,y\in X_{F}\cup\{0\}$, if $-\infty < v({x}) - v({y}) < -n M$ then,
    for every $-\infty < \alpha < -n M$, we have a valuation $v'\in Z$ with 
    $v'({x})-v'({y})=\alpha$.
  \end{enumerate}
\end{corollary}

\begin{proof}
  \begin{enumerate}
    \item  We show that $v'\simeq v$, and we deduce by Lemma~\ref{lem:FutureClocks-diagonals} 
    that $v'\in Z$. So we have to show that $v,v'$ satisfy the same $(X_{D},M)$-safe constraints.
    This is clear for a constraint which does not involve clock $x$.
    Since $x\in X_{F}\setminus X_{D}$, a safe constraint involving clock $x$ must 
    be of the form $x\leqlt c$ or $-x\leqlt c$. 
    We conclude easily since the constraint is $M$-bounded and $-\infty<v(x),v'(x)<-M$.
  
    \item Let $x,y\in X_{D}\cup\{0\}$ be such that $-\infty < v({x}) - v({y}) < -nM$.
    We have $v(x)\neq-\infty\neq v(y)$. Hence $-\infty<v(x)<v(y)-nM<v(y)\leq 0$. 
    We first give a sufficient condition for a valuation $v'$ to be equivalent to $v$.
    Consider the following conditions on a valuation $v'$:
    \begin{enumerate}
      \item  $v'\da_{X_{H}}=v\da_{X_{H}}$ and
      $v'\da_{X_{F}\setminus X_{D}}=v\da_{X_{F}\setminus X_{D}}$,
      
      \item  for all $z\in X_{D}$, we have 
      \begin{itemize}
        \item $v'(z)=v(z)$ if $v(z)=-\infty$ or $v(y)\leq v(z)\leq 0$, and

        \item $v'(z)-v'(x)=v(z)-v(x)$ if $-\infty<v(z)\leq v(x)$,
      \end{itemize}
      
      \item  for all $x',y'\in X_{D}$ such that $v(x)\leq v(x')\leq v(y')\leq v(y)$, 
      we have $v'(x')-v'(y')=v(x')-v(y')$ or both $-\infty<v'(x')-v'(y')<-M$ 
      and $-\infty<v(x')-v(y')<-M$.
    \end{enumerate}    
    It is not hard to check that if a valuation $v'$ satisfies the above conditions then
    $v'\simeq v$. 
    We can also check that there is a valuation $v'$ satisfying the conditions 
    above and such that $v'({x}) - v'({y}) = \alpha$. 
    The property follows.
    
    \item This follows from (2) if $x,y\in X_{D}\cup\{0\}$. 
    We assume below that $x\in X_{F}\setminus X_{D}$ or $y\in X_{F}\setminus X_{D}$.
    As above, we have $-\infty<v(x)<v(y)-n M<v(y)\leq 0$.
    
    Assume that $x\in X_{F}\setminus X_{D}$. Then $-\infty<v(x)<-M$. We apply (1) with 
    $\alpha'=\alpha+v(y)$. We get $v'=v[x\mapsto\alpha']\in Z$ and $v'({x})-v'({y})=\alpha$.
    
    Finally, assume that $x\in X_{D}$ and $y\in X_{F}\setminus X_{D}$.
    We have $-\infty<v(x)-0<-nM$. We apply (2) to the pair of clocks $x,0$ and 
    $\alpha'=\alpha+v(y)$. We get $v'\in Z$ with $v'({x})-0=\alpha+v(y)$.
    Notice that from the construction above (2.a) we have $v'(y)=v(y)$.
    Therefore, $v'({x})-v'({y})=\alpha$.
    \qedhere%
  \end{enumerate}
\end{proof}

\begin{remark}
  Note that in the second and third parts of Corollary~\ref{cor:FutureClocks-0-Diagonals}, we do not maintain the valuation of all the
  other clocks while changing the particular difference that we are interested in. This is in contrast with the first part, where we change the value of the future clock $x \in X_{F}\setminus X_{D}$, while keeping the valuation of the other clocks unchanged.
\end{remark}

\subsection{The dagger lemma: from finiteness to boundedness in safely reachable zones}

We will now use Corollary~\ref{cor:FutureClocks-0-Diagonals} to 
prove the main invariants
satisfied by the zones obtained during the enumeration.  Essentially, the weights
of edges involving non-history clocks come from a finite set which depends on the number
of future clocks in $X_{D}$ and the maximum constant $M$ of the automaton.  This also
induces an invariant on the constraint between a history clock and a future clock.

Before proving Lemma~\ref{lem:dagger-diagonals}, we first state two technical lemmas from \cite{eca-simulations-arxiv}.

\begin{lemma}[\cite{eca-simulations-arxiv}]\label{lem:weight-properties}
  \begin{enumerate}
    \item  Let $(\leqlt,c)$ be a weight and $\alpha\in\overline{\mathbb{R}}$. Then,
    \begin{itemize}
      \item  $\alpha\leqlt c$ iff $(\leq,\alpha)\leq(\leqlt,c)$ iff 
      $(\leq,0)\leq(\leq,-\alpha)+(\leqlt,c)$,
      
      \item $\alpha\not\leqlt c$ iff $(\leqlt,c)<(\leq,\alpha)$ iff
      $(\leq,-\alpha)+(\leqlt,c)<(\leq,0)$ iff $(\leq,-\alpha)+(\leqlt,c)\leq(<,0)$.
    \end{itemize}
    
    \item  Let $(\leqlt,c),(\leqlt',c'),(\leqlt'',c'')$ be weights with
    $(\leq,0)\leq(\leqlt,c)+(\leqlt',c')$.
    Then, there exists $\alpha\in\overline{\mathbb{R}}$ such that $\alpha\leqlt c$ and
    $-\alpha\leqlt' c'$.
    If in addition we have $(\leqlt'',c'')<(\leqlt,c)$ then there exists such an 
    $\alpha$ with $\alpha\not\leqlt'' c''$.
  \end{enumerate}
\end{lemma}

\begin{lemma}[\cite{eca-simulations-arxiv}]\label{lem:fixing-values-distance-graph}
  Let $\GG=\graph{Z}$ for a non-empty \GTA\ zone $Z$, and let $x,y\in X\cup\{0\}$ be a
  pair of distinct nodes and $\alpha\in\overline{\mathbb{R}}$.  There is a valuation
  $v\in\sem{\GG}$ with $v(y)-v(x)=\alpha$ if and only if
  \begin{enumerate}
    \item $(\leq,\alpha)\leq\GG_{xy}$ and $(\leq,-\alpha)\leq\GG_{yx}$, and
  
    \item if $x,y\in X$ and $\alpha\in\mathbb{R}$ is finite then the weights 
    $\GG_{x0},\GG_{0x},\GG_{y0},\GG_{0y}$ are all different from $(\leq,-\infty)$, and
  
    \item if $x,y\in X$ and $\alpha=-\infty$ then $\GG_{0x}\neq(\leq,-\infty)\neq\GG_{y0}$.
  \end{enumerate}
\end{lemma}

The following lemma extends the corresponding property of \cite{eca-simulations-arxiv} by 
taking into account the initial guard $g_{0}$ of a safe $\GTA$.

\begin{lemma}\label{lem:reachable-properties}
  Let $Z$ be a nonempty reachable zone and let $\GG$ be its canonical distance graph.
  \begin{enumerate}
    \item  For all $x\in X_{H}$, we have $\GG_{x0}=(\leq,-\infty)$ or 
    $\GG_{0x}\leq(<,\infty)$.
  
    \item  For all $x,y\in X$, if $\GG_{xy}=(\leq,-\infty)$ then 
    $\GG_{x0}=(\leq,-\infty)$ or $\GG_{0y}=(\leq,-\infty)$.
  \end{enumerate}
\end{lemma}

\begin{proof}
  Let $x\in X_{H}$ be a history clock.  
  Since $\A$ is safe, the initial guard $g_{0}$ induces either the weight
  $(\leq,-\infty)$ for edge $x\to 0$ or the weight $(\leq,0)$ for edge $0\to x$.
  If the weight of $x\to 0$ is $(\leq,-\infty)$, it stays unchanged until we first
  apply the reset operation on $x$, resulting in the weight $(\leq,0)$ for edge $0\to x$.
  Then, the weight of edge $0\to x$ may only be increased by the time elapse operation,
  which sets it to $(<,\infty)$.  This proves the first property.

  For the second property, consider $x,y\in X$ with $\GG_{xy}=(\leq,-\infty)$ and
  $\GG_{x0}\neq(\leq,-\infty)$. We have to show that $\GG_{0y}=(\leq,-\infty)$.
  If $x\in X_{H}$ then we get $\GG_{0x}\leq(<,\infty)$ by the first property. 
  If $x\in X_{F}$ then we have $\GG_{0x}\leq(\leq,0)$.
  In both cases, since $\GG$ is normal, we obtain
  $\GG_{0y}\leq\GG_{0x}+\GG_{xy}=(\leq,-\infty)$ and we are done.
\end{proof}

We next state the following central lemma that give the $(\dagger)$ conditions, that says that for all safely reachable zones, the weight of edges of the form $0 \to x$, $x \to 0$ and $x_1 \to x_2$ belong to the finite set $\{(\leq,-\infty),(<,\infty),(\leq,\infty)\}\cup \{(\leqlt,c)\mid c\in\mathbb{Z} \wedge -nM\leq c\leq nM\}$, for all future clocks $x,x_1,x_2 \in X_{F}$.
In other words, for safely reachable zones, the constraints between non-history clocks come from a finite set.

\begin{lemma}~\label{lem:dagger-diagonals}
  Let $Z$ be a nonempty $(X_{D},M)$-safely reachable zone and let $n=\max(1,|X_{D}|)$.
  Then, the normalized distance graph $\graph{Z}$ satisfies the following $(\dagger)$ conditions:
  \begin{enumerate}

    \item[$\dagger_{1}$] For all $x \in X_{F}$, if $Z_{xy}$ is finite for some $y \in X_{H} \cup \{0\}$, then $(\leq,0)\leq Z_{x0}\leq(\leq,nM)$.

    \item[$\dagger_{2}$] For all $x \in X_{F}$, 
    if $Z_{0x}$ is finite, then
    $(<,-nM)\leq Z_{0x}\leq(\leq,0)$.

    \item[$\dagger_{3}$] For all $x \in X_{H}$ and $y \in X_{F}$, if $Z_{0y}$ is finite, then $Z_{x0} + (<,-nM) \leq Z_{xy}$.    
    
    \item[$\dagger_{4}$] For $x,y \in X_{F}$, 
    if $Z_{xy}$ is finite, then     
    $(<,-nM) \leq Z_{xy} \leq (\leq,nM)$.
  \end{enumerate}
\end{lemma}

\begin{proof}

\medskip
\item[$\dagger_{1}$] For all $x \in X_{F}$, if $Z_{xy}$ is finite for some $y \in X_{H} \cup \{0\}$, then $(\leq,0)\leq Z_{x0}\leq(\leq,nM)$.
In other words, if $Z_{xy}<(<,\infty)$ for some $y \in X_{H} \cup \{0\}$, then
$(\leq,0)\leq Z_{x0}\leq(\leq,nM)$.
    
    First, we consider the case where $y = 0$. 
    So we assume that $(\leq,0)\leq Z_{x0}<(<,\infty)$ is finite.
    Towards a contradiction, suppose that $(\leq, nM) < Z_{x0} < (<,\infty)$.
    Since $Z$ is non-empty, we know that $(\leq,0) \leq Z_{x0} + Z_{0x}$.
    Then, using Lemma~\ref{lem:weight-properties}, we can find 
    $\alpha\in\overline{\mathbb{R}}$ such that $(\leq,\alpha)\leq Z_{x0}$,
    $(\leq,-\alpha)\leq Z_{0x}$, and $nM < \alpha$.
    Notice that $\alpha<\infty$ since $Z_{x0}<(<,\infty)$.
    Further, using Lemma~\ref{lem:fixing-values-distance-graph}, 
    we can get a valuation $v \in Z$ such that $0 - v(x) = \alpha$.   
    Since $nM<\alpha<\infty$, this implies $-\infty<v(x)<-nM$.    
    Let $Z_{x0}=(\leqlt,c)$. We have $nM<c<\infty$.
    Using Corollary~\ref{cor:FutureClocks-0-Diagonals}(3),
    we can get a valuation $v' \in Z$, such that $-\infty<v'(x)<-c$,
    a contradiction as it violates the constraint $0-x\leqlt c$ in the zone.

    Next, assume that $Z_{xy}<(<,\infty)$ for some $y \in X_{H}$.
    Since $Z$ is normal, we have $Z_{x0} \leq Z_{xy} + Z_{y0} < (<,\infty)$ as $Z_{xy} <
    (<,\infty)$ and $Z_{y0} \leq (\leq,0)$.
    We now conclude from the first case that $(\leq,0)\leq Z_{x0}\leq(\leq,nM)$.
    
    \medskip
    \item[$\dagger_{2}$] For all $x \in X_{F}$, if $Z_{0x}$ is finite, then
    $(<,-nM)\leq Z_{0x}\leq(\leq,0)$.  This means that either $Z_{0x}=(\leq,-\infty)$ or
    $(<,-nM)\leq Z_{0x}\leq(\leq,0)$.

    Let $Z_{0x}=(\leqlt,c)$.  Suppose $(\leq,-\infty) < Z_{0x} < (<,-nM)$.  
    We have $-\infty<c<-nM$.  
    By Lemma~\ref{lem:weight-properties}, we can find $\alpha$ such that $(\leq,\alpha)
    \leq Z_{0x}$, $(\leq,-\alpha) \leq Z_{x0}$ and $\alpha \neq -\infty$.
    Then, by
    Lemma~\ref{lem:fixing-values-distance-graph}, we can find $v\in Z$ with $v(x)=\alpha$.
    We have $-\infty < v(x) \leqlt c < -nM$.
    Now, using Corollary~\ref{cor:FutureClocks-0-Diagonals}(3),
    we can get a valuation $v'\in Z$ such that $c<v'(x)<-nM$, which leads to a contradiction as it violates the constraint $x-0\leqlt c$ in the zone.

    \medskip
    \item[$\dagger_{3}$] For all $x \in X_{H}$ and $y \in X_{F}$, if $Z_{0y}$ is finite, then $Z_{x0} + (<,-nM) \leq Z_{xy}$.
    
    If $Z_{x0} = (\leq,-\infty)$ then the inequality trivially holds.
    So, we assume for the rest of the proof that $Z_{x0}\neq(\leq,-\infty)$.
    Since $Z_{0y}$ is finite, we know that $Z_{0y}\neq(\leq,-\infty)$. 
    By Lemma~\ref{lem:reachable-properties}, this implies $Z_{xy}\neq(\leq,-\infty)$. 
    Let $Z_{x0} = (\leqlt,-c)$ and $Z_{xy} = (\leqlt',e)$,
    as shown in Figure~\ref{fig:eq-dagger}. We have $0\leq c<\infty$ and 
    $-\infty<e\leq 0$.

    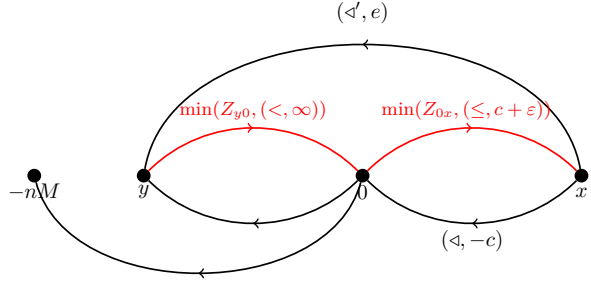
\begin{figure}[tbh]
      \centering
      \scalebox{0.8}{
      \begin{tikzpicture}[scale=0.9]
          
          \begin{scope}[state/.style={draw, thick, circle, inner sep=2pt}]
                    \node [state, fill=black] (s) at (-1, 6) {}; 
                    \node [state, fill=black] (p) at (1, 6) {}; 
            \node [state, fill=black] (q) at (5, 6) {}; 
            \node [state, fill=black] (r) at (9, 6) {}; 
          \end{scope}
          
                \begin{scope}[thick,decoration={markings, mark=at position 0.5 with {\arrow{>}}}]             
            \draw [draw=red,postaction={decorate},bend left=45] (p) to (q); 
            \draw [postaction={decorate},bend left=45] (q) to (p); 
            \draw [postaction={decorate},bend left=45] (r) to (q); 
            \draw [postaction={decorate},bend left=-80] (r) to (p); 
            \draw [draw=red,postaction={decorate},bend left=45] (q) to (r); 
            \draw [postaction={decorate},bend left=75] (q) to (s); 
          \end{scope}
        
          \node at (-1,5.7)  {$-nM$};
                \node at (1,5.7)  {$y$};
          \node at (5,5.7)  {$0$};
          \node at (9,5.7)  {$x$};
          
                \node at (5,9)  {$(\triangleleft', e)$};
          \node at (3,7.2)  {\small \textcolor{red}{$\min(Z_{y0}, (<,\infty))$}};
                \node at (6.9,7.2)  {\small \textcolor{red}{$\min(Z_{0x}, (\le, c+ \varepsilon))$}};
                
          \node at (7,4.8)  {$(\triangleleft, -c)$};
                
      \end{tikzpicture}
      }
      \captionof{figure}{Distance graph $\graph{Z}$ (without the red edges) and $\GG'$ (with the red edges).}
      \label{fig:eq-dagger}
    \end{figure}

    Fix $\varepsilon>0$.  Consider the distance graph $\GG'$ obtained from $\graph{Z}$ by
    setting the weight of $0 \to x$ to $\min(Z_{0x}, (\le, c+ \varepsilon))$,
    and the weight of $y\to 0$ to $\min(Z_{y0}, (<,\infty))$, as shown in
    Figure~\ref{fig:eq-dagger}.
    It is easy to see that $\GG'$ is also in standard form.

    Next, we show that there are no negative cycles in this graph.
    Since $Z\neq\emptyset$, the candidates for being negative must use the new weight
    $(\leq,c + \varepsilon)$ of $0\to x$ or the new weight $(<,\infty)$ of
    $y\to 0$ or both.  Then the possible negative cycles are:
    \begin{itemize}
      \item $0\to x \to 0$ with weight $(\leq,c+\varepsilon)+Z_{x0}=(\leq,c+\varepsilon)+
      (\leqlt, -c) = (\leqlt,\varepsilon)$, which is not negative, since
      $\varepsilon > 0$.

      \item $0 \to y \to 0$ with weight $Z_{0y}+(<,\infty)$ which is not
      negative since $Z_{0y}\neq(\leq,-\infty)$,

      \item $y \to 0 \to x \to y$ with weight $(<,\infty)+(\leq,
      c+\varepsilon)+Z_{xy}$ which is not negative since
      $Z_{xy}\neq(\le,-\infty)$.
    \end{itemize}

    Since $\GG'$ has no negative cycles, Lemma~\ref{lem:standard-form} implies
    $\sem{\GG'}\neq\emptyset$.  Note that $\sem{\GG'}\subseteq\sem{\graph{Z}}=Z$. 
    We know that for all $v \in \sem{\GG'}$, we have $c \leqlt v(x) \leq c + \varepsilon$.

    We will now show that there exists a valuation $v' \in \sem{\GG'}$ such that $-nM - \varepsilon \leq v'(y)$. 
    Let $v\in\sem{\GG'}$.  If $-nM \leq v(y)$, we let $v' = v$ and we are done.
    Otherwise, $-\infty<v(y)<-nM$, where the first inequality is due to 
    $\GG'_{y0}\leq(<,\infty)$.
    Using Corollary~\ref{cor:FutureClocks-0-Diagonals}(3),
    there exists a valuation
    $v'\in\sem{\GG'}$ such that $v'(y) = -nM - \varepsilon$ since $\varepsilon>0$.

    Since $v'\in\sem{\GG'}$, we have $c\leqlt v'(x)\leq c+\varepsilon$ and we obtain
    $-nM-c-2\varepsilon\leq v'(y)-v'(x)\leqlt' e$, where the last inequality uses 
    again $v'\in\sem{\GG'}$ and $\GG'_{xy}=(\leqlt',e)$.     
    Since this is true for all $\varepsilon>0$ we deduce that $-nM-c\leq e$. We deduce 
    that $(<,-nM-c)\leq(\leqlt',e)=Z_{xy}$. We conclude using 
    $(<,-nM-c)=(<,-nM)+(\leqlt,-c)$.

    \medskip
    \item[$\dagger_{4}$] For $x,y \in X_{F}$, if $Z_{xy}$ is finite, then     
    $(<,-nM) \leq Z_{xy} \leq (\leq,nM)$.

    Suppose that $x,y \in X_{F}$, and $Z_{xy}=(\leqlt,c)\not\in\{(\leq,-\infty),(<,\infty),(\leq,\infty)\}$ is finite. 
    Notice that, since $\graph{Z}$ is standard, this implies $Z_{x0}\neq(\leq,\infty)$.
    
    The proof proceeds by application of Lemma~\ref{lem:fixing-values-distance-graph}, and for this, when $x,y\in X$ and $\alpha\in\mathbb{R}$ is finite, we need to first show that the weights $Z_{x0},Z_{0x},Z_{y0},Z_{0y}$ are all different from $(\leq,-\infty)$. We will now show this. 
    \begin{itemize}
      \item We get this for free for weights $Z_{x0}, Z_{y0}$, as $x$ and $y$ are future clocks.

      \item Suppose that $Z_{0x} = (\leq,-\infty)$.  Then, since $Z$ is non-empty,
      we get $Z_{x0} = (\leq,\infty)$, a contradiction.
      
      \item Suppose $Z_{0y} = (\leq,-\infty)$.  Since $\graph{Z}$ is normal, we
      have $Z_{xy} \leq Z_{x0} + Z_{0y}=(\leq,-\infty)$ (since
      $Z_{x0} \neq (\leq,\infty)$).  Again this is a contradiction with
      $Z_{xy}\neq(\leq,-\infty)$.

    \end{itemize}
    Thus, we have shown that $Z_{x0},Z_{0x},Z_{y0},Z_{0y}$ are all different from $(\leq,-\infty)$.

    Next, we consider the two possibilities for violation of the $\dagger_4$ condition. We will show that both of them lead to a contradiction.
    \begin{enumerate}
      \item $(\leq, nM) < Z_{xy} = (\leqlt,c) < (<,\infty)$. 
      This implies that $nM<c<\infty$.  

      Using Lemma~\ref{lem:weight-properties}, we can find $\alpha\in\mathbb{R}$ such that
      $(\leq,\alpha)\leq Z_{xy}$, $(\leq,-\alpha)\leq Z_{yx}$, and
      $nM<\alpha$.  Notice that $\alpha\leq c<\infty$.  Further, using
      Lemma~\ref{lem:fixing-values-distance-graph}, we know that there exists a valuation
      $v\in Z$ with $v(y)-v(x)=\alpha$.  We get
      $-\infty<-\alpha=v(x)-v(y)<-nM$ and by 
      Corollary~\ref{cor:FutureClocks-0-Diagonals}(3), we can find a valuation $v'\in Z$ with
      $-\infty<v'(x)-v'(y)=\beta<-c$ (for instance, $\beta=-c-1$), as illustrated in Figure~\ref{fig:dagger-7-a}.  This is a
      contradiction with the constraint $y-x \leqlt c$ in $Z$.

      \begin{figure}[h]
        \centering
        \includegraphics[width=.5\linewidth,scale=.1]{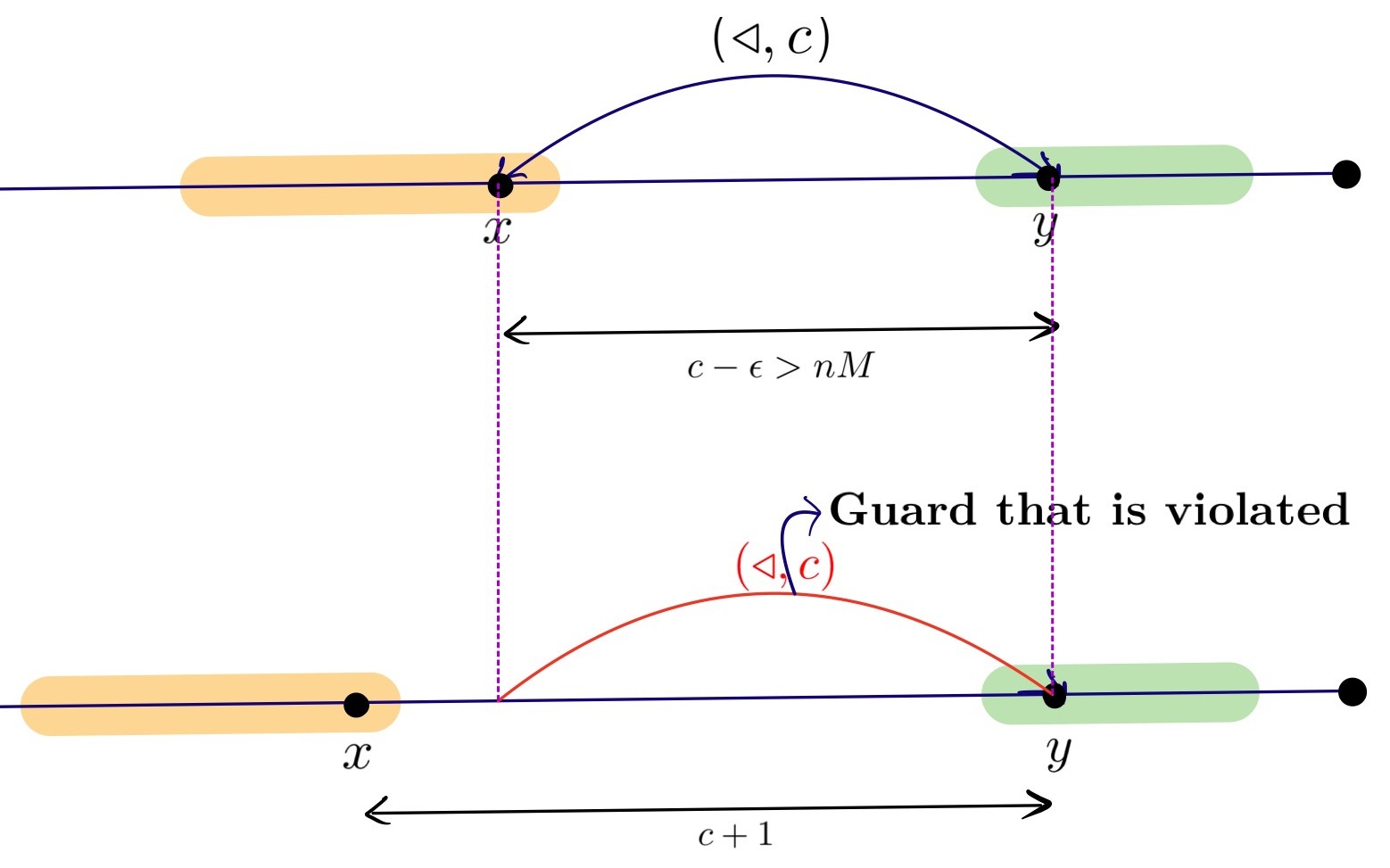}
        \captionof{figure}{Increasing the difference between $x$ and $y$ using 
        $\simeq$-equivalence.}
        \label{fig:dagger-7-a}
      \end{figure}
      
      \item $(\leq,-\infty) < Z_{xy} = (\leqlt,c) < (<, -nM)$.  
      This implies that $-\infty<c<-nM$.  

      Using Lemma~\ref{lem:weight-properties}, we can find $\alpha\in\mathbb{R}$ such that
      $(\leq,\alpha)\leq Z_{xy}$, $(\leq,-\alpha)\leq Z_{yx}$, and $-\infty<\alpha$.
      Notice that $\alpha\leq c<-nM$.  Further, using
      Lemma~\ref{lem:fixing-values-distance-graph}, we know that there exists a valuation
      $v\in Z$ with $v(y)-v(x)=\alpha$.  Since
      $-\infty<\alpha<-nM$, we use
      Corollary~\ref{cor:FutureClocks-0-Diagonals}(3) to find a valuation $v'\in Z$ with
      $c<v'(y)-v'(x)=\beta<-nM$ (for instance, $\beta=\frac{c-nM}{2}$), as illustrated in Figure~\ref{fig:dagger-7-b}.  This is a contradiction with the constraint $y-x \leqlt c$ in $Z$.
      
      \begin{figure}[h]
        \centering
        \includegraphics[width=.5\linewidth,scale=.1]{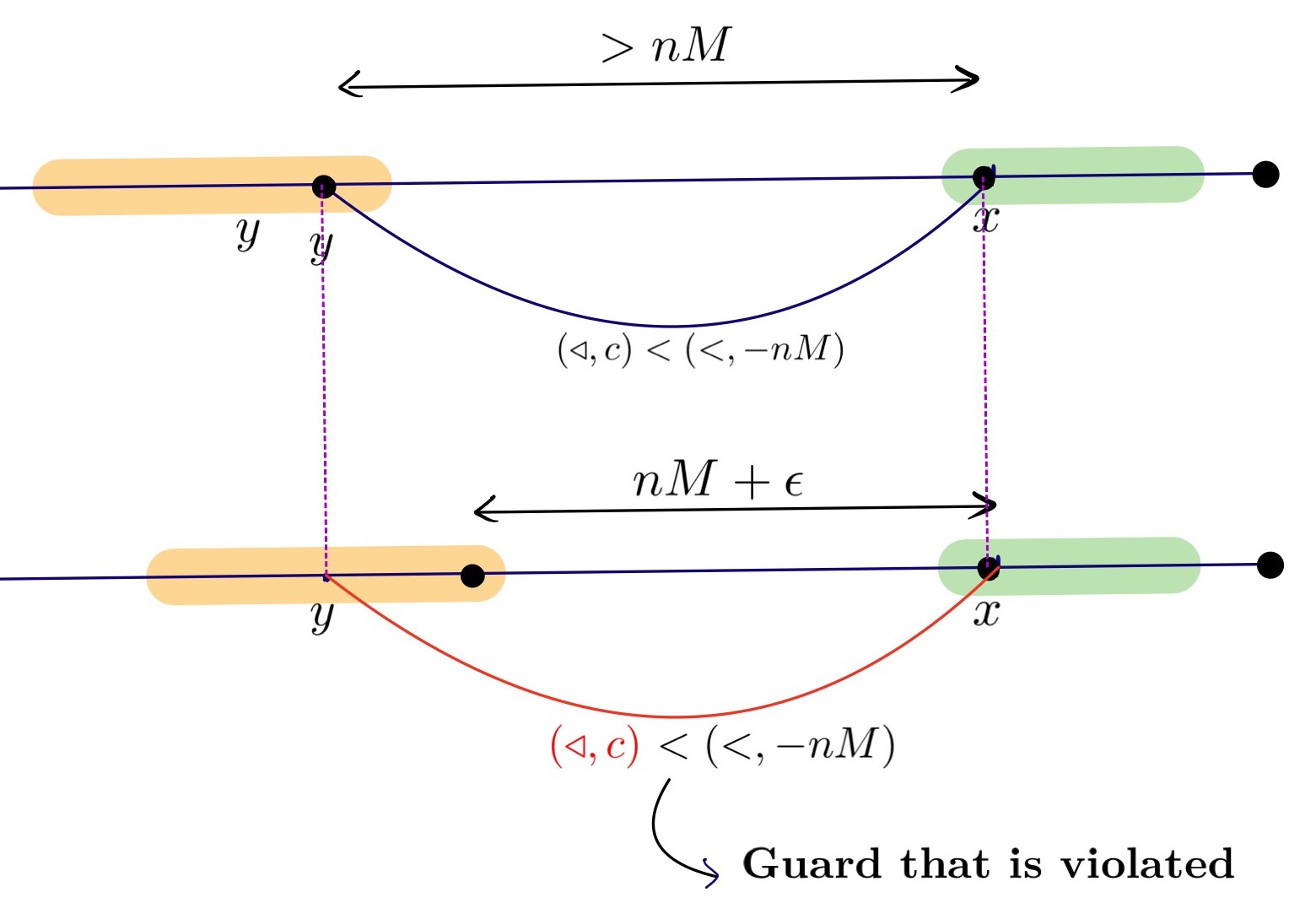}
        \captionof{figure}{Shrinking the difference between $x$ and $y$ using
        $\simeq$-equivalence.}  
        \label{fig:dagger-7-b}
      \end{figure}
    \end{enumerate}
    Therefore, if $Z_{xy}$ is finite, then $(<,-nM) \leq Z_{xy} \leq (\leq,nM)$.\qedhere 
\end{proof}

\begin{remark}
  Note that in each of the dagger conditions in Lemma~\ref{lem:dagger-diagonals}, we can differentiate the cases where the future clock belongs to the set $X_{D}$ or not. 
  In particular, when the future clock $x$ being considered is not in $X_{D}$, the
  bound can be restricted to $M$ (instead of $nM$).  
\end{remark}

Thus, we obtain as a corollary that, for event-predicting automata, we do not even need simulation to obtain finiteness of its zone graph.

\begin{corollary}
  Let $\Aa$ be an event-predicting automata with diagonal constraints. Then, the zone graph of $\Aa$ is finite.    
\end{corollary}

\section{Finiteness of the simulation relation} \label{sec:finiteness}

In this section, we will show that the simulation relation $\preceq_{\A}$ defined in Section~\ref{sec:simulation} is finite, which implies that the reachability algorithm 
terminates.  Recall that given a \GTA\ $\A$, we have an associated map $\G$ from states of $\A$ to sets of atomic constraints.
Let $M=\max\{|c| \mid c\in\mathbb{Z} \text{ is used in some constraint of } \A\}$, the maximal constant of $\A$.  We have $M\in\mathbb{N}$ and constraints in the sets $\G(q)$ use constants in $\{-\infty,\infty\}\cup\{c\in\mathbb{Z}\mid|c|\leq M\}$.
We will refer to such constraints as \emph{$M$-bounded integral constraints}.

Recall that the simulation relation $\preceq_{\A}$ was defined on nodes of the  zone graph of $\A$ by $(q,Z)\preceq_{\A}(q',Z')$ if $q=q'$ and $Z\preceq_{\G(q)}Z'$.
This simulation relation $\preceq_{\A}$ is \emph{finite} if for any infinite sequence $(q,Z_{0}),(q,Z_{1}),(q,Z_{2}),\ldots$ of \emph{safely reachable} nodes in the zone graph of $\A$ we find $i<j$ with $(q,Z_{j})\preceq_{\A}(q,Z_{i})$, i.e., $Z_{j}\preceq_{\G(q)}Z_{i}$.
Notice that we restrict to \emph{safely reachable} zones in the definition above.  
Our goal now is to prove that the relation $\preceq_{\A}$ is finite.  
The structure of the proof is as follows.  
\begin{enumerate}
  \item We proved in Lemma~\ref{lem:dagger-diagonals} of Section~\ref{sec:dagger} that for any \emph{safely reachable} node $(q,Z)$ of the zone graph of $\A$, the canonical distance graph $\graph{Z}$ satisfies a set of conditions, that we call $(\dagger)$ conditions, which depend only on the maximal constant $M$ of $\A$ and the number of future clocks in $\A$.

  \item We will now introduce an equivalence relation $\sim_{M}$ of \emph{finite index} on valuations (depending on $M$ only) and show in Lemma~\ref{lem:sim_M main property} of Section~\ref{sec:eqreln} that, if $G$ is a set of atomic constraints using \emph{$M$-bounded integral constraints} 
  and if $Z$ is a zone such that its canonical distance graph $\graph{Z}$ satisfies $(\dagger)$ conditions, then $\da_{G}Z$ is a union of $\sim^{n}_{M}$ equivalence classes.
\end{enumerate}

\medskip\noindent \textbf{An equivalence relation of finite index on valuations.}~\label{sec:eqreln}
We first define an equivalence relation of finite index $\sim_{M}$ on valuations.
First, we define $\sim_{M}$ on 
$\alpha,\beta\in\overline{\mathbb{R}}=\mathbb{R}\cup\{-\infty,\infty\}$
by $\alpha\sim_{M}\beta$ if
$(\alpha\leqlt c \Longleftrightarrow \beta\leqlt c)$ for all
$(\leqlt,c)$ with ${\leqlt}\in\{<,\leq\}$ and
$c\in\{-\infty,\infty\}\cup\{d\in\mathbb{Z}\mid|d|\leq M\}$.  In
particular, if $\alpha\sim_{M}\beta$ then
$(\alpha=-\infty \Longleftrightarrow \beta=-\infty)$ and
$(\alpha=\infty \Longleftrightarrow \beta=\infty)$.

Next, for valuations $v_{1},v_{2}\in\V$, we define
$v_{1}\sim^{n}_{M}v_{2}$ by two conditions: $v_{1}(x)\sim_{nM}v_{2}(x)$ and
$v_{1}(x)-v_{1}(y)\sim_{(n+1)M}v_{2}(x)-v_{2}(y)$
for all clocks $x,y\in X$. Notice that we use $(n+1)M$ for differences of values.
Clearly, $\sim^{n}_{M}$ is an equivalence relation of finite index on
valuations.
Using this, we can show that the zones that are reachable in
a safe \GTA\ are unions of  $\sim^{n}_{M}$-equivalence classes.

\medskip\noindent \textbf{Distance graph for valuations that simulate a given valuation.}~\label{sec:distance graph for up G v}
For a valuation $v$, we let $\ua_{G}v=\{v'\in\V \mid v\preceq_{G} v'\}$, i.e., the set of valuations $v'$ which simulate $v$.  
We will define a distance graph, denoted $\graphv{G}{v}$, such that $\sem{\graphv{G}{v}}=\ua_{G}v$.  
We remark that $\sem{\graphv{G}{v}}$ is not really a zone since it may use constants that are not integers.

We will now define the distance graph $\graphv{G}{v}$ which denotes the set $\ua_{G}v$. 
We will define $\graphv{G}{v}$ as the intersection of a distance graphs $\GG_{v}^{G}$ and a guard $g_{v}^{G}$.

  \begin{definition}
    The distance graph $\GG_{v}^{G}$ is defined as follows.
    \begin{itemize}
      \item For each future clock $x \in X_{F}$, we have the edges 
      $x \xra{(\leq,-v(x))} 0$ and $0 \xra{(\leq,v(x))} x$.

      \item For each history clock $y \in X_{H}$, we have 
      \begin{itemize}
        \item the edge $0 \to y$ with weight $(\leq,v(y))$ if there is a constraint $y \leqlt c \in G$ with $c < \infty$ and $v \models y \leqlt c$.

        \item the edge $y \to 0$ with weight $(\leq, -v(y))$ if there is a constraint $c \leqlt y \in G$ with $c < \infty$ and $v \not\models c \leqlt y$.
      \end{itemize}
    \end{itemize}
  \end{definition}

\begin{definition}
  The guard $g_{v}^{G}$ is given by the set of all constraints of the form $y - x \leqlt
  c$ in $G$ where $x,y \in X \cup \{0\}$ and $v \models y - x \leqlt c$.
\end{definition}

With this definition, we can show 
that if $G$ is a set of atomic constraints containing both $x\leq 0$ and $0\leq x$ for
each clock $x\in X_{F}$, then $\ua_{G}v = \sem{\GG_{v}^{G}}\cap\sem{g_{v}^{G}}$.

\begin{lemma}\label{lem:upset-v}
  Let $G$ be a set of 
  constraints such that for all future clock 
  $x\in X_{F}$ we have both $x\leq 0$ and $0\leq x$ in $G$.
  We have $\ua_{G}v = \sem{\GG_{v}^{G}}\cap g_{v}^{G}$.
\end{lemma}

\begin{proof}
  $\subseteq$: Let $v'$ be such that $v\preceq_{G}v'$. 
  % First, 
  By definition of the 
  simulation relation, for all $g'=y-x\leqlt c$ in $G$ such that $v\models g'$, we have 
  $v'\models g'$. Hence, $v'\models g_{v}^{G}$. 
  Next, let $x\in X_{F}$ be a future clock. If $v(x)=-\infty$ then for all 
  $0\leq\delta<\infty$ we have $v+\delta\models x\leq0$. Since $v\preceq_{G}v'$ we get 
  $v'+\delta\models x\leq 0$, which implies $v'(x)=-\infty=v(x)$. Otherwise,  
  let $0\leq\delta=-v(x)<\infty$. Since $v+\delta\models x\leq 0 \wedge 0\leq x$ and 
  $v\preceq_{G}v'$, we get $v'+\delta\models x\leq 0 \wedge 0\leq x$. We deduce that 
  $v'(x)=v(x)$. Therefore, $v'$ satisfies the edges $x\xra{\leq,-v(x)}0$ and 
  $0\xra{\leq,v(x)}x$ of $\GG_{v}^{G}$.
  
  Now, let $x\in X_{H}$ be a history clock. Assume that $v\models x\leqlt c$ for some 
  $x\leqlt c$ in $G$ with $0\leq c<\infty$. 
  Using $v\preceq_{x\leqlt c}v'$, we get $v'(x)\leq v(x)$. 
  Hence, $v'$ satisfies the edge $0\xra{\leq,v(x)}x$ of $\GG_{v}^{G}$.
  Assume that $v\not\models c\leqlt x$ for some $c\leqlt x$ in $G$ with $0\leq c<\infty$. 
  Again, we obtain $v(x)\leq v'(x)$ from $v\preceq_{c\leqlt x}v'$. 
  Hence, $v'$ satisfies the edge $x\xra{\leq,-v(x)}0$ of $\GG_{v}^{G}$.
  Thus, $v'$ satisfies all constraints of $\GG_{v}^{G}$, i.e., 
  $v'\in\sem{\GG_{v}^{G}}$.
  
  \medskip\noindent $\supseteq$: %
  Let $v\in\sem{\GG_{v}^{G}}$ with $v\models g_{v}^{G}$.
  Let $g'=y-x\leqlt c$ be a diagonal constraint in $G$ with $x,y\in X$.
  If $v\models g'$ then $g'$ is in $g_{v}^{G}$ and $v'\models g'$.
  Therefore, $v\preceq_{g'}v'$.
  
  Now, let $g'$ be a non-diagonal constraint on a future clock, i.e., $x\leqlt c$ or 
  $c\leqlt x$ with $x\in X_{F}$. Since $v\in\sem{\GG_{v}^{G}}$ we get $v'(x)=v(x)$ and we 
  deduce that $v\preceq_{g'}v'$.
  Let $g'$ be an upper non-diagonal constraint $x\leqlt c$ on a history clock $x\in X_{H}$.
  If $v\not\models g'$ then $v\preceq_{g'}v'$.
  If $v\models g'$ and $c$ is finite then we get $v'(x)\leq v(x)$ from the edge
  $0\xra{\leq,v(x)}x$ of $\GG_{v}^{G}$. Hence, $v\preceq_{g'}v'$.
  If $g'$ is $x<\infty$ and $v\models g'$ then $g'$ is in $g_{v}^{G}$ and we get 
  $v'(x)<\infty$ from $v'\models g_{v}^{G}$. We deduce that $v\preceq_{g'}v'$.
  If $g'$ is $x\leq\infty$ then $g'$ is equivalent to \emph{true} and $v\preceq_{g'}v'$.
  Let $g'$ be a lower non-diagonal constraint $c\leqlt x$ on a history clock $x\in X_{H}$.
  If $v\models g'$ then $g'$ is in $g_{v}^{G}$ and we get $v'\models g'$. 
  Therefore, $v\preceq_{g'}v'$. Assume now that $v\not\models g'$.
  If $c$ is finite then  we get $v(x)\leq v'(x)$ from the edge
  $x\xra{\leq,-v(x)}0$ of $\GG_{v}^{G}$. We deduce that $v\preceq_{g'}v'$.
  If $g'$ is $\infty<x$ then $g'$ is equivalent to \emph{false} and $v\preceq_{g'}v'$.
  Lastly, when $g'$ is $\infty\leq x$ and $v(x)$ is finite. Then, for all 
  $0\leq\delta<\infty$ we have $v+\delta\not\models g'$. Therefore, $v\preceq_{g'}v'$.
\end{proof}

\begin{remark}
  \begin{enumerate}
    \item $\GG_{v}^{G}$ is in standard form, but not necessarily in normal form.
    \item $\sem{\GG_{v}^{G}}$ is non-empty, since $v \in \sem{\GG_{v}^{G}}$.
    \item $g_{v}^{G}$ is a conjunction of atomic constraints, each of which is $(X_{D},M)$-safe.
  \end{enumerate}
\end{remark}

Further, we show that if $\GG_{v}^{G} \cap Z'$ is empty and $\GG'$ is the normalized
distance graph of $Z'$, then there is a small witness, i.e., a negative cycle in
$\min(\GG_{v}^{G},\GG')$ containing at most three edges, and belonging to one of three
specific forms.
This also gives us an efficient simulation
check for \GTA\ zone graphs.

\begin{lemma}\label{lem:negative-cycles-Gv2-Z}
  Let $v$ be a valuation, $Z'$ a non-empty \emph{reachable} event zone with canonical distance graph $\GG'$ and $G$ a set of atomic constraints. 
  Then, $\GG_{v}^{G} \cap Z'$ is empty iff there is a
  negative cycle in one of the following forms:
  \begin{enumerate}
  \item $0 \to x \to 0$ with $0\to x$ from $\GG_{v}^{G}$ and
    $x \to 0$ from $\GG'$,
  \item $0 \to y \to 0$ with $0 \to y$ from $\GG'$ and $y \to 0$
    from $\GG_{v}^{G}$, and
  \item $0 \to x \to y \to 0$, with weight of $x \to y$ from
    $\GG'$ and the others from $\GG_{v}^{G}$.
  \end{enumerate}
\end{lemma}

\begin{proof}
  Since the distance graph $\GG'$ is in normal form, it has no negative cycle. 
  Similarly, $\GG_{v}^{G}$ has no negative cycle since $v\in \GG_{v}^{G} \neq\emptyset$.  
  We know that $\GG_{v}^{G} \cap Z'=\emptyset$ iff there is a (simple) negative cycle in $\min(\GG_{v}^{G},\GG')$.  
  Since $\GG'$ is in normal form, we may restrict to negative cycles which do not use two consecutive edges from $\GG'$.  
  Further, note that all edges of $\GG_{v}^{G}$ are adjacent to node $0$.  Hence, if a simple cycle uses an edge from $\GG'$ which is adjacent to $0$, it
  consists of only two edges $0\to x\to 0$, one from $\GG'$ and one from $\GG_{v}^{G}$.
  Otherwise, the simple cycle is of the form $0\to x\to y\to 0$ where the edge $x\to y$ is from $\GG'$ and the other two edges are from $\GG_{v}^{G}$.  
\end{proof}

\begin{lemma}~\label{lem:sim_M and graphv}
  Let $v \sim^{n}_{M} v'$ and $G$ be a set of $M$-bounded integral constraints. 
  Then, we have the following
  \begin{enumerate}
  \item $g_{v'}^{G} = g_{v}^{G}$. 
  \item The graph $\GG_{v'}^{G}$ is obtained by replacing the weights $(\leq,v(x))$ (resp.\ $(\leq,-v(x))$) by $(\leq,v'(x))$ (resp.\ $(\leq,-v'(x))$) in the graph $\GG_{v}^{G}$.
\end{enumerate}
\end{lemma}
  
\begin{proof}
\begin{enumerate}
  \item $g_{v'}^{G} = g_{v}^{G}$ is easy to see from the definition of $\GG_{v'}^{G}$ and $\GG_{v}^{G}$, and the fact that $v \sim_{(n+1)M} v'$.

  \item For a future clock $x \in X_{F}$, this is easy to see from the definition for edges $x \to 0$ and $0\to x$ adjacent to $x$.

  We consider now edges adjacent to history clocks $y \in X_{H}$.
  \begin{itemize}
  \item Consider the edge $0\to y$.
    If its weight is $(\leq,v(y))$ in $\GG_{v}^{G}$ 
    then there is some $y \leqlt c\in G$ with $c<\infty$ and $v(y)\leqlt c$.
    Since $v \sim_{(n+1)M} v'$, we deduce that $v'(y)\leqlt c$ and the edge $0\to y$ has
    weight $(\leq,v'(y))$ in $\GG_{v'}^{G}$. 
    
  \item Consider the edge $y\to 0$.
  If its weight is $(\leq,-v(y))$ in $\GG_{v}^{G}$, then there is some $c\leqlt y \in
  G$ with $c<\infty$ and $c\not\leqlt v(y)$.
  Since $v \sim_{(n+1)M} v'$, we deduce that $c\not\leqlt v'(y)$ and the edge $y\to0$
  has weight $(\leq,-v'(y))$ in $\GG_{v'}^{G}$.
  \qedhere
      
  \end{itemize}
\end{enumerate}
\end{proof}

Using all the results above, we can now show that the zones that are reachable in a safe \GTA\ are unions of  $\sim^{n}_{M}$-equivalence classes.

\begin{remark}
  Before we state the lemma, we list some properties that we will use extensively in the proof of the lemma.
  \begin{enumerate}
    \item $-b \leqlt a$ iff $-a \leqlt b$ iff $(\leq,0) \leq (\leqlt, a+b)$.
    \item $a \leqlt b$ iff $\neg(b ~\tilde{\leqlt}~ a)$  where $\tilde{\leq} = <$ and $\tilde{<} = \leq$. 
    \item $\alpha \sim_{M} \beta$ and $c \in \mathbb{R}$ is such that $-M\leq c\leq M$ or $(\leqlt,c) \in \{(\leq,-\infty),(<,\infty),(\leq,\infty)\}$, then, $c \leqlt \alpha$ iff $c \leqlt \beta$. 
    This is because 
    \begin{itemize}
      \item $c \leqlt \alpha$ iff $\neg(\alpha ~\tilde{\leqlt}~ c)$ by (2) above.
      \item $\neg(\alpha ~\tilde{\leqlt}~ c)$ iff $\neg(\beta ~\tilde{\leqlt}~ c)$ by definition of $\sim_{M}$ equivalence.
      \item $\neg(\beta ~\tilde{\leqlt}~ c)$ iff $c \leqlt \beta$  by (2) above.
    \end{itemize}    
  \end{enumerate}
\end{remark}

\begin{lemma}\label{lem:sim_M main property}
  Let $G$ be a set of $X_{D}$-safe $M$-bounded integral constraints which contains
  both $x\leq 0$ and $0\leq x$ for each future clock $x\in X_{F}$.
  Let $Z$ be a zone with a canonical distance graph $\graph{Z}$ satisfying the 
  $(\dagger)$ conditions of Lemma~\ref{lem:dagger-diagonals}.  
  Let $v_{1},v_{2}\in\V$ be valuations with $v_{1}\sim^{n}_{M}v_{2}$.
  Then, $v_{1}\in\da_{G}Z$ iff $v_{2}\in\da_{G}Z$.
\end{lemma}

\begin{proof}
  Notice that $v\in\da_{G}Z$ iff $\ua_{G}v\cap Z\neq\emptyset$.
  We need to show that $\ua_{G}v_{1}\cap Z\neq\emptyset$ iff $\ua_{G}v_{2}\cap Z\neq\emptyset$. 
  Using the characterization of up-sets 
  given by Lemma~\ref{lem:upset-v}, 
  this amounts to 
  $Z\cap g_{v_1}^{G} \cap \sem{\GG_{v_1}^{G}} \neq\emptyset$ iff 
  $Z\cap g_{v_2}^{G} \cap \sem{\GG_{v_2}^{G}} \neq\emptyset$.

  Further, since $v_{1}\sim^{n}_{M}v_{2}$, 
  using Lemma~\ref{lem:sim_M and graphv}, it follows that
  $g_{v_2}^{G} = g_{v_1}^{G}$.
  Let $Z' = Z \cap g_{v_2}^{G} = Z \cap g_{v_1}^{G}$. If $Z'$ is empty then the 
  equivalence holds. Otherwise,  
  let $\graph{Z'}$ be the normalized distance graph of $Z'$.
  Note that since $Z$ was an $(X_{D},M)$-safely reachable zone and $g_{v_1}^{G}$ is a
  conjunction of atomic constraints, each of which is $(X_{D},M)$-safe, it follows that
  $Z'$
  is an $(X_{D},M)$-safely reachable zone. 
  As a consequence, the $\dagger$ conditions of Lemma~\ref{lem:dagger-diagonals} apply to $Z'$.

  In the rest of the proof, we will now work with the zone $Z'$ (using its normalized
  distance graph representation $\graph{Z'}$) and the standard distance graphs
  $\GG_{v_1}^{G}$ and $\GG_{v_2}^{G}$.
  The proof proceeds by contradiction.
  We assume that $\ua_{G}v_{1}\cap Z\neq\emptyset$ and
  $\ua_{G}v_{2}\cap Z=\emptyset$.  
  This is equivalent to $Z'\cap\sem{\GG_{v_1}^{G}}\neq\emptyset$ and
  $Z'\cap\sem{\GG_{v_2}^{G}}=\emptyset$.
  By Lemma~\ref{lem:negative-cycles-Gv2-Z},
  we can find a negative cycle $C_{2}$ using one edge from
  $\graph{Z'}$ and one or two edges from $\GG_{v_2}^{G}$.  
  By Lemma~\ref{lem:sim_M and graphv}, 
  we have a corresponding cycle
  $C_{1}$ using the same edge from $\graph{Z'}$ and the same one or two
  edges from $\GG_{v_1}^{G}$ (with weights using $v_{1}$ instead of $v_{2}$).  
  The cycle $C_{1}$ is not negative since $Z'\cap\sem{\GG_{v_1}^{G}}\neq\emptyset$
  
  The rest of the proof involves a case analysis of the various forms that the cycle $C_2$
  can take, which we provide below.
  We consider the different cases.
  \begin{enumerate}
  \item Cycle
    $C_{2}=0 \xra{(\leq,v_{2}(y))} y \xra{Z'_{y0}} 0$ for some history clock $y \in X_{H}$. 
    
    We have $C_{1}=0 \xra{(\leq,v_{1}(y))} y \xra{Z'_{y0}} 0$.

    Since we have the edge $0 \xra{(\leq,v_{1}(y))} y$ in
    $\GG_{v_1}^{G}$, there is a constraint $y \leqlt' c'$ in
    $G$ with $c'<\infty$ and $v_{1}(y)\leqlt' c'$.
    We deduce that $0\leq v_{1}(y)\leq M$.

    Let $Z'_{y0}=(\leqlt,c)$.
    Since $C_{1}$ is not a negative cycle, we get
    $(\leq,0)\leq(\leqlt,c+v_{1}(y))$, which is equivalent to
    $-c\leqlt v_{1}(y)$. Using $0\leq v_{1}(y)\leq M$ and
    $v_{1}\sim^{n}_{M}v_{2}$ we deduce that $-c\leqlt v_{2}(y)$.  This
    is equivalent to $(\leq,0)\leq(\leqlt,c+v_{2}(y))$, a
    contradiction with $C_{2}$ being a negative cycle.
    
  \item Cycle
    $C_{2}=0 \xra{Z'_{0y}} y \xra{(\leq,-v_{2}(y))}
    0$ for some history clock $y \in X_{H}$.
    
    We have $C_{1}=0 \xra{Z'_{0y}} y \xra{(\leq,-v_{1}(y))} 0$.
    
    Since we have the edge $y \xra{(\leq,-v_{1}(y))} 0$ in
    $\GG_{v_1}^{G}$, there is a constraint $c'\leqlt'y$ in
    $G$ with $c'<\infty$ and $c'\not\leqlt' v_{1}(y)$.  We deduce that
    $0\leq v_{1}(y)\leq M$.

    Let $Z'_{0y}=(\leqlt,c)$.
    Since $C_{1}$ is not a negative cycle, we get
    $(\leq,0)\leq(\leqlt,c-v_{1}(y))$, which is equivalent to
    $v_{1}(y)\leqlt c$.  Using $v_{1}\sim^{n}_{M}v_{2}$ and
    $0\leq v_{1}(y)\leq M$, we deduce that
    $v_{2}(y)\leqlt c$.  This is equivalent to
    $(\leq,0)\leq(\leqlt,c-v_{2}(y))$, a contradiction with
    $C_{2}$ being a negative cycle.
    
  \item Cycle
    $C_{2}=0 \xra{(\leq,v_{2}(x))} x \xra{Z'_{x0}}
    0$ for some future clock $x \in X_{F}$.    
    
    We have $C_{1}=0 \xra{(\leq,v_{1}(x))} x \xra{Z'_{x0}} 0$.
    
    Since $C_{2}$ is negative, we have $Z'_{x0}\neq(\leq,\infty)$.  Also, if
    $Z'_{x0}=(<,\infty)$ then we must have $v_{2}(x)=-\infty$, which implies
    $v_{1}(x)=-\infty$ since $v_{1}\sim^{n}_{M}v_{2}$, a contradiction with $C_{1}$
    being non-negative.  Hence, $Z'_{x0}=(\leqlt,c)$ is finite and by
    $(\dagger_{1})$, we infer $0\leq c\leq nM$.
    
    Since $C_{1}$ is not negative, we get
    $(\leq,0)\leq(\leqlt,c+v_{1}(x))$, which is equivalent to
    $-c\leqlt v_{1}(x)$.  Using $v_{1}\sim^{n}_{M}v_{2}$ and
    $0\leq c\leq nM$ we deduce that $-c\leqlt v_{2}(x)$.  This is
    equivalent to $(\leq,0)\leq(\leqlt,c+v_{2}(x))$, a
    contradiction with $C_{2}$ being a negative cycle.
    
  \item Cycle
    $C_{2}=0 \xra{Z'_{0x}} x \xra{(\leq,-v_{2}(x))}
    0$  for some future clock $x \in X_{F}$.      
    
    We have $C_{1}=0 \xra{Z'_{0x}} x \xra{(\leq,-v_{1}(x))} 0$.
    
    Let $Z'_{0x}=(\leqlt,c)$.  Since $C_{2}$ is negative, we
    deduce that $v_{2}(x)\neq-\infty$.  Using
    $v_{1}\sim^{n}_{M}v_{2}$, we infer $v_{1}(x)\neq-\infty$.  Since
    $C_{1}$ is not negative, we get $Z'_{0x}\neq(\leq,-\infty)$.
    From
    $(\dagger_{2})$, we infer
    $(<,-nM)\leq Z'_{0x}\leq(\leq,0)$ and $-nM\leq c\leq 0$.
    
    Since $C_{1}$ is not a negative cycle, we get
    $(\leq,0)\leq(\leqlt,c-v_{1}(x))$, which is equivalent to
    $v_{1}(x)\leqlt c$.
    Using $v_{1}\sim^{n}_{M}v_{2}$ and $-nM\leq c\leq 0$, we deduce that
    $v_{2}(x)\leqlt c$.  
    This is equivalent to
    $(\leq,0)\leq(\leqlt,c-v_{2}(x))$, a contradiction with
    $C_{2}$ being a negative cycle.
  
  \item Cycle
    $C_{2}=0 \xra{(\leq,v_{2}(y))} y
    \xra{Z'_{yx}} x \xra{(\leq,-v_{2}(x))} 0$ for some history clock $y \in X_{H}$ and future clock $x \in X_{F}$.    
        
    We have
    $C_{1}=0 \xra{(\leq,v_{1}(y))} y
    \xra{Z'_{yx}} x \xra{(\leq,-v_{1}(x))} 0$.

    Let $Z'_{yx}=(\leqlt,c)$.  As in case 1 above, we get
    $0 \leq v_{1}(y) \leq M$.  From the fact that the cycle
    $0 \xra{(\leq,v_{1}(y))} y \xra{Z'_{y0}} 0$ is
    not negative, we get $(\leq,-M)\leq Z'_{y0}$.  Since $C_{2}$
    is negative, we get $v_{2}(x)\neq-\infty$.  Using
    $v_{1}\sim^{n}_{M}v_{2}$, we infer $v_{1}(x)\neq-\infty$.  From
    the fact that the cycle
    $0 \xra{Z'_{0x}} x \xra{(\leq,-v_{1}(x))} 0$ is
    not negative, we deduce $Z'_{0x}\neq(\leq,-\infty)$.  Using
    $(\dagger_{3})$
    we obtain
    $$
    (\leq,-M)+(<,-nM)\leq Z'_{y0}+(<,-nM) \leq Z'_{yx}
    =(\leqlt,c)
    $$
    and we deduce that $-(n+1)M\leq c\leq 0$.
    
    Since $C_{1}$ is not a negative cycle, we get
    $(\leq,0)\leq(\leqlt,c+v_{1}(y)-v_{1}(x))$, which is
    equivalent to $-c\leqlt v_{1}(y)-v_{1}(x)$.  Using
    $v_{1}\sim^{n}_{M}v_{2}$ and $-(n+1)M\leq c\leq 0$ we deduce that
    $-c\leqlt v_{2}(y)-v_{2}(x)$.  We conclude as in the
    previous cases.
    
  \item Cycle
    $C_{2}=0 \xra{(\leq,v_{2}(x))} x
    \xra{Z'_{xy}} y \xra{(\leq,-v_{2}(y))} 0$ for some history clock $y \in X_{H}$ and future clock $x \in X_{F}$.    

    We have
    $C_{1}=0 \xra{(\leq,v_{1}(x))} x
    \xra{Z'_{xy}} y \xra{(\leq,-v_{1}(y))} 0$.

    Since $C_{2}$ is negative but not $C_{1}$, we get first 
    $Z'_{xy}\neq(\leq,\infty)$ and then $v_{1}(x)\neq-\infty$.
    As in case 2 above, we get $0\leq v_{1}(y)\leq M$. 
    We deduce that $Z'_{xy}=(\leqlt,c)<(<,\infty)$ and $c\neq\infty$.
    From $(\dagger_{1})$ we obtain $Z'_{x0}\leq(\leq,nM)$.
    Since $0 \xra{(\leq,v_{1}(x))} x \xra{Z'_{x0}} 0$ is
    not a negative cycle, we get $-nM\leq v_{1}(x)\leq 0$.
    Finally, we obtain $0\leq v_{1}(y)-v_{1}(x)\leq (n+1)M$.

    Since $C_{1}$ is not a negative cycle, we get
    $(\leq,0)\leq(\leqlt,c+v_{1}(x)-v_{1}(y))$, which is
    equivalent to $v_{1}(y)-v_{1}(x)\leqlt c$.  Using
    $v_{1}\sim^{n}_{M}v_{2}$ and
    $0\leq v_{1}(y)-v_{1}(x)\leq (n+1)M$, we deduce that
    $v_{2}(y)-v_{2}(x)\leqlt c$.  We conclude as in the
    previous cases.

  \item Cycle
    $C_{2}=0 \xra{(\leq,v_{2}(x))} x
    \xra{Z'_{xy}} y \xra{(\leq,-v_{2}(y))} 0$
    with $x\neq y$ for future clocks $x,y \in X_{F}$.    

    We have
    $C_{1}=0 \xra{(\leq,v_{1}(x))} x
    \xra{Z'_{xy}} y \xra{(\leq,-v_{1}(y))} 0$.

    Since $C_{2}$ is negative but not $C_{1}$, using $v_{1}\sim^{n}_{M}v_{2}$ we get
    successively $Z'_{xy}\neq(\leq,\infty)$, $v_{2}(y)\neq-\infty\neq
    v_{1}(y)$, $v_{1}(x)\neq-\infty\neq v_{2}(x)$, and finally
    $(\leq,-\infty)< Z'_{xy}<(<,\infty)$.
    
    Let $Z'_{xy}=(\leqlt,c)$. 
    From $(\dagger_{4})$, we deduce that $-nM\leq c \leq nM$.

    Since $C_{1}$ is not a negative cycle, we get
    $(\leq,0)\leq(\leqlt,c+v_{1}(x)-v_{1}(y))$, which is
    equivalent to $v_{1}(y)-v_{1}(x)\leqlt c$.  Using
    $v_{1}\sim^{n}_{M}v_{2}$ and $-nM\leq c \leq nM$, we deduce that
    $v_{2}(y)-v_{2}(x)\leqlt c$.  We conclude as in the
    previous cases.
    
  \item Cycle
    $C_{2}=0 \xra{(\leq,v_{2}(x))} x
    \xra{Z'_{xy}} y \xra{(\leq,-v_{2}(y))} 0$
    with $x\neq y$ for history clocks $x,y \in X_{H}$.    

    We have
    $C_{1}=0 \xra{(\leq,v_{1}(x))} x
    \xra{Z'_{xy}} y \xra{(\leq,-v_{1}(y))} 0$.
    
    As in case 1 above, we get $0\leq v_{1}(x)\leq M$.  As in
    case 2 above, we get $0\leq v_{1}(y)\leq M$.  We obtain
    $-M\leq v_{1}(y)-v_{1}(x) \leq M$.

    Let $Z'_{xy}=(\leqlt,c)$.  Since $C_{1}$ is not
    negative, we get
    $(\leq,0)\leq(\leqlt,c+v_{1}(x)-v_{1}(y))$, which is
    equivalent to $v_{1}(y)-v_{1}(x)\leqlt c$.  Using
    $v_{1}\sim^{n}_{M}v_{2}$ and
    $-M\leq v_{1}(y)-v_{1}(x) \leq M$, we deduce that
    $v_{2}(y)-v_{2}(x)\leqlt c$.  We conclude as in the
    previous cases.
    
  \end{enumerate}
  Notice that we have crucially used the ``$(n+1)M$'' occurring in  the
  definition of $v_1 \sim^{n}_M v_2$ (as $v_1(x) - v_1(y)
  \sim_{(n+1)M} v_2(x) - v_2(y)$) in the cases where we deal with
  cycles containing one future clock and one history clock (Cases
  5 and 6). 
\end{proof}

Finally, from Lemmas~\ref{lem:dagger-diagonals} and \ref{lem:sim_M main property}, we obtain our
main theorem of the section.

\begin{theorem}\label{thm:simulation finite}
  The simulation relation $\preceq_{\A}$ is finite if $\Aa$ is safe.
\end{theorem}

\begin{proof}
  Let $(q,Z_{0}),(q,Z_{1}),(q,Z_{2}),\ldots$ be an infinite sequence
  of \emph{reachable} nodes in the zone graph of $\A$.  
  By Lemma~\ref{lem:dagger-diagonals}, for all $i$, the distance graph $\graph{Z_{i}}$ in
  canonical form satisfies conditions $(\dagger)$.

  The set $\G(q)$ contains only $X_{D}$-safe and $M$-bounded integral constraints.
  Let $G$ be $\G(q)$ together with the constraints $x\leq 0$ and $0\leq x$ for each 
  future clock $x\in X_{F}$.
  From Lemma~\ref{lem:sim_M main property} we deduce that for all $i$, 
  $\da_{G}Z_{i}$ is a union of $\sim^{n}_{M}$-classes.
  Since $\sim^{n}_{M}$ is of finite index, there are only finitely many
  unions of $\sim^{n}_{M}$-classes.  Therefore, we find $i<j$ with
  $\da_{G}Z_{i}=\da_{G}Z_{j}$, which implies $Z_{j}\preceq_{G}Z_{i}$.
  Since $\G(q)\subseteq G$, this also implies $Z_{j}\preceq_{\G(q)}Z_{i}$.
\end{proof}

\section{Experimental evaluation}\label{sec:experiments}

We have implemented a prototype that takes as input a \GTA, as given in Definition~\ref{defn:gta}, and applies our reachability algorithm,
in the open source tool \textsc{Tchecker}~\cite{TChecker}. To do so, we extend
\textsc{Tchecker} to allow clocks to be declared as one of \emph{normal}, \emph{history}, \emph{prophecy}, or \emph{timer}, and extend the syntax of edges to allow arbitrary interleaving of guards and clock changes (reset/release).
Our tool, along with the benchmarks used in this paper, is publicly available and can be downloaded from \href{https://github.com/anirjoshi/GTA-Model}{https://github.com/anirjoshi/GTA-Model}.
We present selected results in Table~\ref{table:exp-results}, with
further details in Appendix~\ref{sec:app:experiments}.

\begin{table}[t]
  \centering
  \begin{tabular}{|l|l|r|r|r|r|r|r|}
    \hline
    Sl. 
    & \multicolumn{1}{|c|}{Models}
    & \multicolumn{3}{|c|}{$\G$-Sim}
    & \multicolumn{3}{|c|}{\GTA\ Reach}\\
    \cline{3-8}
    No. 
    & 
    & Visited & Stored & Time in
    & Visited & Stored & Time in\\
    & 
    & nodes & nodes & sec.
    & nodes & nodes & sec.\\
    \hline
    1 & Dining Phi. (6) & 5480 & 5480 & 4.911 & 5480 & 5480 & 6.410\\ 
    \hline
    2 & FDDI (10) & 10219 & 459 & 10.139 & 10219 & 459 & 16.797 \\ 
    \hline
    3 & Fischer (10) & 447598 & 260998 & 29.1574 & 447598 & 260998 & 34.6517 \\ 
    \hline
    \hline
    4 & $\toyeca(10000,4)$ &  150049 & 49 & 4.22 & 3 & 3 & 0.0003 \\ \hline
    5 & $\toyeca(5000,6)$ & 315193 & 193 & 15.572 & 3 & 3 & 0.0006 \\
    \hline
    6 & $\toyeca(1000,100)$ & \multicolumn{3}{|c|}{TIMEOUT}& 3 & 3 & 0.877 \\ \hline
    7 & $\toyeca(50000,120)$ & \multicolumn{3}{|c|}{TIMEOUT} & 3 & 3 & 1.52 \\ 
    \hline
    \hline
    8 & $\firealarmpattern$(5) & \multicolumn{3}{|c|}{\_} & 46 & 46 & 0.027\\
    \hline
    9 & $\csmacdbounded$(1) & \multicolumn{3}{|c|}{\_} & 34 & 26 & 0.0054\\
    \hline
    10 & $\csmacdbounded$(4) & \multicolumn{3}{|c|}{\_} & 4529 & 2068 & 2.597\\
    \hline
    11 & $\abpp$(1) & \multicolumn{3}{|c|}{\_} & 114 & 114 & 0.038\\
    \hline
    12 & $\abpq$(1) & \multicolumn{3}{|c|}{\_} & 168 & 168 & 0.026\\
    \hline
  \end{tabular}
  \caption{Experimental results obtained by running our prototype
  implementation and, when possible, the standard reachability
  algorithm using $\G$-simulation implemented in \textsc{Tchecker}.  
  Both implementations use a breadth-first search with simulation. 
  For each model, we give the parameters in parenthesis - for $\toyeca$, we explain the parameterization in Appendix~\ref{sec:app:experiments}, while for others, we report the number of concurrent processes.
  All experiments were run on an Ubuntu machine with an Intel-i5 7th Generation processor and 8GB RAM, and timeout set to 60 seconds.}
  \label{table:exp-results}
\end{table}

First, we consider timed automata models from standard benchmarks\cite{TripakisY01,FDDI,Dining-philosopher}.
Despite the overhead induced by our framework (e.g., maintaining general programs on transitions), we are only slightly worse off w.r.t.\ running time than the standard algorithm, while visiting and storing the same number of nodes.
We illustrate this in rows 1-3 of Table~\ref{table:exp-results} by providing a
comparison of our tool with the implementation of the state-of-the-art zone-based reachability algorithm using $\G$-simulation introduced in \cite{GastinMS18,GastinMS19, GastinMS20}.

Next, we consider 
models belonging to the class of ECA without diagonal constraints.  We
remark that ours is the first implementation of a reachability algorithm that can operate
on the whole class of ECA directly.
We compare against an implementation that first translates the ECA into
a timed automaton using the translation proposed in \cite{AFH99}, and then runs the
state-of-the-art reachability algorithm of \cite{GastinMS18,GastinMS19, GastinMS20} on
this timed automaton.
From rows 4-7 of Table~\ref{table:exp-results}, we observe significant improvements, both
in terms of running time as well as number of visited nodes and stored nodes w.r.t. the
standard approach.

Finally, in Rows 8-12, we consider the unified model \GTA.
As already pointed out, model-checking an event-clock specification
$\varphi$ over a timed automaton model $\Aa$ can be reduced to the reachability on the
product of the TA $\Aa$ and the ECA representing $\neg \varphi$.
In this spirit, our implementation allows the model to use any combination of
\emph{normal} clocks, \emph{history} clocks, \emph{prophecy} clocks or \emph{timers} and
moreover, permits diagonal guards between any of these clocks.
To the best of our knowledge, no existing tool allows all these features. 
We emphasize this by the $-$ in the $\G$-Sim column of Table~\ref{table:exp-results}.

We model simple but useful properties using event-clocks, and check these properties on some standard models from literature such as 
$\altbp$($\abp$)~\cite{KuroseRoss}, $\csmacd$~\cite{TripakisY01} and $\firealarm$~\cite{fire-alarm}.  
Note that for the benchmark $\firealarmpattern$, the specification is modelled using an ECA with diagonals.  
As a consequence, the product automaton that we check reachability on contains
normal clocks and event-clocks.  Here, we consider the following ECA specification:
no three $a$'s occur within $k$ time units.  The negation of this property can be
easily modeled by an ECA with two states and a transition on $a$ with the diagonal constraint $\history{a} - \prophecy{a} \le k$, where $\history{a}$ is the history clock recording time since the previous occurrence of $a$, and $\prophecy{a}$ is a future clock predicting the time to the next $a$ occurrence. 
When reading an $a$, the quantity $\history{a} - \prophecy{a}$ gives the distance between the next and the previous occurrence.  
This language is used in~\cite{DBLP:conf/fsttcs/BozzelliMP19a} to observe that ECA with diagonals are more expressive than ECA. 
Finally, we remark that the model of $\abp$ contains timers. 
For a more detailed discussion of the model and specifications in these
benchmarks, see Appendix~\ref{sec:app:experiments}.

In conclusion, as can be seen from the experimental results in Table~\ref{table:exp-results}, we are able to demonstrate the full power of our reachability algorithm for the unified model of \GTAfull.

\section{Conclusion}\label{sec:conclusion}

The success of timed automata verification can safely be attributed to
the advances in the zone-based 
technology over the last three decades. In fact,~\cite{Dill89}, the
precursor to the seminal works~\cite{DBLP:conf/icalp/AlurD90,AD94}, 
already laid the foundations for zones by describing the
Difference-Bounds-Matrices (DBM) data structure. Our goal in this
work has been to unify timing features defined in different timed models,
while at the same time retain the ability to use efficient state-of-the-art
algorithms for reachability. To do so, we have
equipped the model with two kinds of 
clocks, history and
future, and modified the transitions to contain a program that
alternates 
between a guard and a change to the variables. For the algorithmic
part, we have adapted the $\Gg$-simulation framework to this powerful
model. The main challenge was to show finiteness of the simulation in
this extended setting. To aid the practical use of this generic model,
we have developed a prototype implementation that can answer
reachability for \GTA. We remark that decidability for \GTA\ comes via zones, and not through regions. In fact, since we
generalize event-clock automata, we do not have a finite region
equivalence for \GTA~\cite{GeeraertsRS14}. 

We conclude with some interesting avenues for future work.
An immediate future work is to use \GTAfull\ for model-checking timed specifications over real-time systems.
Further, the complexity and expressivity of safe GTA are natural intersting theoretical open questions, but we believe they are not obvious. Both these questions are answered in the timed automata literature using regions. However, we cannot have a region equivalence for our model, since even for the subclass of ECA, it was shown that no finite bisimulation is possible.
In particular, it would be interesting to investigate if is possible to have a translation from safe GTA to timed automata. Note that even if such a translation exists, it is likely to incur an exponential blowup since even the translation from ECA to TA costs an exponential.
Coming to the complexity of the reachability problem for safe GTA, it is easy to see that our procedure runs in EXPSPACE, as we have shown that each reachable zone is a union of equivalence classes of a finite index (see Lemma~\ref{lem:sim_M main property}). On the other hand, PSPACE-hardness is inherited from timed automata~\cite{DBLP:conf/icalp/AlurD90,AlurPhD}. Closing the complexity gap is open. We note that even in timed automata, the precise complexity of the simulation based reachability algorithm is difficult to analyze, but its selling point is that it works well in practice.
Finally, we would also like to investigate liveness verification for \GTA, in
particular what future clocks bring us when we consider the setting of
$\omega$-words.

\bibliography{gta.bib}
\appendix

\section{Appendix for Section~\ref{sec:experiments}}\label{sec:app:experiments}

\subsection*{Benchmarks for \GTA\ in Table~\ref{table:exp-results}}

In each of the benchmarks, we consider a model for which we check a property. 
For each of these properties, we propose an event-clock automaton modelling the negation of the property. 
Then, whether the model satisfies the property may be checked by checking reachability on the product where the model synchronizes with the ECA on the actions of the ECA.
 
Note that we only provide here the ECA modelling the negation of the property that we want to check, and not the full product of the model and the ECA. 
We provide the model for the $\altbp$ ($\abp$) (Figure~\ref{fig:alt:bit}). 
The models for $\firealarm$ and $\csmacd$ are the standard models as given in~\cite{fire-alarm}  and~\cite{CSMACD}, respectively.

While depicting event-clock automata, we will use $\history{e}$ to denote the history clock recording time since the previous occurrence of event \textit{e}, and $\prophecy{e}$ to denote the prophecy clock predicting the negative of the time to the next \textit{e}.

\begin{figure}[!htbp]
  \centering
    \begin{subfigure}[b]{0.45\textwidth}
      \centering
      \scalebox{0.8}{
        \begin{tikzpicture} [node distance = 1cm]
          \node (p0) [state] {$P_0$};
          \node (p1) [state, accepting, right= 4.5cm of p0] {$P_1$};
          \node (init) [left= 0.5cm of p0] {};

           \path [-stealth, thick]
           
            (init) edge (p0)
             
             (p0) edge [loop above] node[above, inner sep=1pt] {$\alive$}()
                          
             (p0) edge node[below] {$\alive$} node[above, inner sep=2pt] {$\history{\alive} - \prophecy{\alive} \le k$} (p1)
      
             (p1) edge[loop above] node[above, inner sep=1pt] {$\alive$}()

             ;
              
          \end{tikzpicture}
        }
        \caption{\small ECA model for $\firealarmpattern$}
    \label{fig:fire-alarm-prop}
    \end{subfigure}%
  \hspace{0.2in}
  \begin{subfigure}[b]{0.45\textwidth}
    \centering
    \scalebox{0.8}{
      \begin{tikzpicture} [node distance = 1cm]
        \node (p0) [state] {$P_0$};
        \node (p1) [state, accepting, right= 4.5cm of p0] {$P_1$};
        
        \node (init) [left= 0.5cm of p0] {};
          
        \path [-stealth, thick]
        
         (init) edge (p0)
                      
           (p0) edge[loop above] node[above, inner sep=1pt] {$begin_1, cd_1$}(p0)
           
           (p0) edge node[below] {$cd_1$} node[above, inner sep=2pt] 
           {$\overrightarrow{begin_1}<-30 \land -\infty<\overrightarrow{cd_1}$} (p1)
    
           (p1) edge[loop above] node[above, inner sep=1pt] {$begin_1, cd_1$}()
           ;
            
        \end{tikzpicture}
        }
        \caption{\small ECA model for $\csmacdbounded$}
        \label{fig:CSMACD:Prop}
  \end{subfigure}%
\end{figure}    

\medskip\noindent \textbf{$\firealarmpattern$.} 
  We consider the $\firealarm$ model from~\cite{fire-alarm}.
  The model is a network consisting of $n$ processes, referred to as $\Sensor$ processes, and a server process.
  Each process in the model is modelled using a timed automaton.
  Here, we check the property that no three $\alive$ actions are executed by
  the process Sensor$_{1}$ in $k$ time units.  
  The negation of this property can be modeled by the ECA
  in Figure~\ref{fig:fire-alarm-prop} with two states and a transition on $\alive$
  with the diagonal constraint $\history{\alive} - \prophecy{\alive} \le k$.
  When reading an action $\alive$, the quantity $\history{\alive} - \prophecy{\alive}$
  gives the distance between the next and the previous occurrence.

\medskip\noindent \textbf{$\csmacdbounded$.} 
We consider the $\csmacd$ model given in~\cite{CSMACD}.  
The model is a network consisting of $n$ processes, referred to as $\Station$ processes, and a central $\Bus$ process.
The property that we check here is: after each \emph{detected collision} (modelled using a $cd$ action), except the last one, $\Station_1$ sends a message (modelled using a $begin_1$ action) in $30$ time units.  
The negation of this property can be modeled by the
ECA of Figure~\ref{fig:CSMACD:Prop}.
When reading an action \textit{cd$_1$}, the constraints (1) $\prophecy{begin_1} < -30$
says that $begin_1$ (which denotes Process$_1$ sending a message) cannot be seen within
$30$ time units, (2) $-\infty<\prophecy{cd_1}$ says that this is not the last $cd$ event
(and therefore, at least one more collision will be detected in the future).

\begin{figure}[!htbp]
  \centering
  \begin{subfigure}[c]{0.5\textwidth}
    \centering
    \scalebox{0.8}{
      \begin{tikzpicture} [node distance = 1cm]
        % States
        \node (s0) [state] {$s_0$};
        \node (s1) [state, right=3cm of s0] {$s_1$};
        \node (s2) [state, below=2cm of s1] {$s_2$};
        \node (s3) [state, left=3cm of s2] {$s_3$};
        
        \node (a) [below left= 1.5cm of s2] {$Sender$};

        \node (init) [left= 0.5cm of s0] {};

        \path [-stealth, thick]
        
         (init) edge (s0)
  
            (s0) edge node[above, inner sep=1pt] {\scriptsize$send_0$} node[below, inner sep=0.1, pos=0.48] {\scriptsize $\langle \mathtt{timeout}(t); \mathtt{set(t=3)}; \rangle$} (s1)
            
            (s1) edge node[right,inner sep=1.5pt] {\scriptsize$ack_0$} node [right,inner sep=1.5pt, pos=0.63] {\scriptsize$\langle \mathtt{stop(t)} \rangle$} (s2)
            
            (s2) edge node[below, inner sep=1pt] {\scriptsize$send_1$} node[above,inner sep=0.1, pos=0.48] {\scriptsize $\langle \mathtt{timeout}(t); \mathtt{set(t=3)}; \rangle$} (s3)

            (s3) edge node[left, inner sep=1pt] {\scriptsize$ack_1$} node [left, inner sep=1pt, pos=0.35] {\scriptsize$\langle \mathtt{stop(t)} \rangle$}  (s0)
            
            (s0) edge [in=30,out=60,loop]  node[above] {\scriptsize$ack_1$} ()

            (s0) edge [in=120,out=150,loop]  node[above] {\scriptsize$ack_0$} ()
            
            (s1) edge [loop right]  node[right] {\scriptsize$ack_1$} ()
            
            (s1) edge [loop above]  node[above] {\scriptsize $\langle \mathtt{timeout}(t); \mathtt{set(t=3)}; \rangle$} node[right] {\scriptsize$send_0$} ()
            
            (s2) edge [loop right]  node[below] {\scriptsize$ack_0$} ()

            (s2) edge [loop below]  node[left, pos=0.8, inner sep=0.1] {\scriptsize$ack_1$} ()
            
            (s3) edge [loop left]  node[above, outer sep=3.2pt] {\scriptsize$ack_0$} ()
            
            (s3) edge [loop below]  node[left, inner sep=0.1, pos=0.8] {\scriptsize$send_1$} node[below, inner sep=0.2, pos=0.5] {\scriptsize $\langle \mathtt{timeout}(t); \mathtt{set(t=3)}; \rangle$} ()
            
        ;
        \end{tikzpicture}
    }
  \label{fig:alt:bit:sender}
  \end{subfigure}
  \hfil
  \begin{subfigure}[c]{0.35\textwidth}
    \centering
    \scalebox{0.8}{
    \begin{tikzpicture} [yshift=1.0cm, node distance = 1cm]
      \node (c0) [state] {\scriptsize$Initial$};
      \node (a) [below = 1.2cm of c0] {$Channel$};

      \node (c1) [state, above right=1.2cm of c0] {$c_1$};
      \node (c2) [state, below right=1.2cm of c0] {$c_2$};
      \node (c3) [state, above left=1.2cm of c0] {$c_3$};
      \node (c4) [state, below left = 1.2cm of c0] {$c_4$};
  
      \node (init) [above= 1cm of c0] {};

      \path [-stealth, thick]
      
       (init) edge (c0)

          (c0) edge node[sloped, above, inner sep=1pt] {\scriptsize$send_0$} (c1)
          (c0) edge node[sloped,above,inner sep=1pt] {\scriptsize$send_1$} (c2)
          (c0) edge node[sloped, anchor=center,above,inner sep=1pt] {\scriptsize$rack_0$} (c3)
          (c0) edge node[sloped, anchor=center,above, inner sep=1pt] {\scriptsize$rack_1$} (c4)
          (c1) edge [in=0,out=-90]  node[right] {\scriptsize$lost$} (c0)
          (c2) edge [in=0,out=90]  node[right] {\scriptsize$lost$} (c0)
          (c4) edge [in=180,out=90]  node[left] {\scriptsize$lost$} (c0)
          (c3) edge [in=180,out=-90]  node[left] {\scriptsize$lost$} (c0)
          
          (c1) edge [in=90,out=180]  node[sloped,  inner sep=1pt, pos=0.2,above] {\scriptsize$rpkt_0$} (c0)
          (c2) edge [in=-90,out=180]  node[sloped, anchor=center,above] {\scriptsize$rpkt_1$} (c0)
  
          (c4) edge [in=-90,out=0]  node[sloped, anchor=center,above,inner sep=1pt] {\scriptsize$ack_1$} (c0)
          (c3) edge [in=90,out=0]  node[sloped, anchor=south,pos=0.2,inner sep=0.1] {\scriptsize$ack_0$} (c0)
      ;
    \end{tikzpicture}
    }
    \label{fig:alt:bit:channel}
  \end{subfigure}
  
  \vspace{-2mm}
  
  \begin{subfigure}[b]{0.45\textwidth}
    \centering
    \scalebox{0.7}{
      \begin{tikzpicture} [node distance = 1cm]
        % States
        \node (r0) [state] {$r_0$};
        \node (r1) [state, right=2cm of r0] {$r_1$};
        \node (r2) [state, left=2cm of r0] {$r_2$};
    
        \node (a) [below = 1cm of r0] {$Receiver$};

         \path [-stealth, thick]
            (r0) edge node[below] {$rpkt_0$} (r1)
            (r1) edge[bend right=60] node[below] {$rack_0$} (r0)
            (r1) edge[out=60, in=30, loop] node [above] {$rack_0$}()
            (r2) edge[out=120, in=150, loop] node [above] {$rack_1$}()
            
            (r0) edge node[below] {$rpkt_1$} (r2)
            (r2) edge[bend left=60] node[below] {$rack_1$} (r0)
        ;
      \end{tikzpicture}
      }
    \label{fig:alt:bit:receiver}
  \end{subfigure}
  \hfil
  \begin{subfigure}[b]{0.45\textwidth}
    \centering
    \scalebox{0.7}{
    \begin{tikzpicture} [yshift=0cm, node distance = 1cm]
      
      \begin{scope}[xshift=-2.7cm]
        \node (f0) [state] {$f_0$};
        \node (f1) [state, right=2cm of f0] {$f_1$};
        \node (f2) [state, right=2cm of f1] {$f_2$};

        \node (a) [below = 1cm of f1] {$Scheduler$};

        \node (init) [left= 0.5cm of f0] {};
        
        \path [-stealth, thick]
        
         (init) edge (f0)

            (f0) edge[loop below] node[below] {$\Sigma \setminus \{lost\}$}()
            
            (f0) edge node[below] {$lost$} (f1)
            (f1) edge node[below] {$lost$} (f2)
            (f1) edge[loop below] node[below] {$\Sigma_B$}()
            (f2) edge[loop below] node[below] {$\Sigma_B$}()
            (f1) edge[in=30,out=150]  node[above] {$\Sigma_P$} (f0)
            (f2) edge[in=60,out=150]  node[above] {$\Sigma_P$} (f0)
        ;
      \end{scope}
    \end{tikzpicture}
    }
    \label{fig:alt:bit:scheduler}
  \end{subfigure}
  \caption{$\altbp$}
  \label{fig:alt:bit}
  \vspace{0.1in}
\end{figure}    

\medskip\noindent \textbf{$\altbp$.}
We consider a variant of the $\altbp$~\cite{KuroseRoss} as depicted in
Figure~\ref{fig:alt:bit}.
We model sending a packet with identifier $i\in\{0,1\}$ in Sender with the action
$send_i$, and receiving an acknowledgement with identifier $i \in\{0,1\}$ in Sender with
the action $ack_i$.
The Sender uses a timer $t$.
Recall that the timer operations are (1) $\mathtt{set(t=3)}$ that sets timer $t$ to value $c$, 
(2) $\mathtt{timeout}(t)$ that checks whether $t$ is 0,
(3) $\mathtt{stop}(t)$ that forgets the value of the timer and sets it 
$-\infty$ (to indicates that it is unused.) 
Note that in the automaton $Scheduler$ in Figure~\ref{fig:alt:bit}, $\Sigma_B=\{send_0,
send_1, rack_0, rack_1\}$, $\Sigma_P=\{rpkt_0, rpkt_1, ack_0, ack_1\}$, $\Sigma = \Sigma_B
\sqcup \Sigma_P \sqcup \{lost\}$.

\begin{figure}[!htbp]
  \centering
  \scalebox{0.8}{
  \begin{minipage}{.5\textwidth}
			\centering    
			\begin{figure}[H]
				  \centering
          \begin{tikzpicture} [node distance = 1cm,scale=1.5]
            \node (p0) [state] {$P_0$};
            \node (p1) [state, accepting, right= 3cm of p0] {$P_1$};
            \node (init) [left= 0.5cm of p0] {};

            \path [-stealth, thick]
            
             (init) edge (p0)
        
               (p0) edge[loop above] node[above, inner sep=1pt] {$ack_0, send_0, send_1$}(p0)
               (p0) edge node[below] {$send_0$} node[above, inner sep=2pt] {$\overrightarrow{ack_0}-\overrightarrow{send_1}<0$} (p1)
        
               (p1) edge[loop above] node[above, inner sep=1pt] {$ack_0, send_0, send_1$}(p1)
                ;
                
            \end{tikzpicture}
            \caption{$\abpp$}
            \label{fig:ABP:P1}
        \end{figure}	
        
        \vspace{-7mm}

        \begin{figure}[H]
				  \centering
          \begin{tikzpicture} [node distance = 1cm,scale=1.5]
          \node (p0) [state] {$P_0$};
          \node (p1) [state, accepting, right= 3cm of p0] {$P_1$};

           \path [-stealth, thick]
                          
             (p0) edge[loop above] node[above, inner sep=1pt] {$ack_0, send_0$}(p0)
             
             (p0) edge node[below] {$ack_0$} node[above, inner sep=2pt] {$\overleftarrow{send_0}>3$} (p1)
      
             (p1) edge[loop above] node[above, inner sep=1pt] {$ack_0, send_0$}(p1)
             ;

          \end{tikzpicture}
          \caption{$\abpq$}
          \label{fig:ABP:P2}      
      \end{figure}	
		\end{minipage}
		\begin{minipage}{.5\textwidth}
		\centering    
    \begin{figure}[H]
      \centering
      \begin{tikzpicture} [node distance = 1cm,
      dot/.style = {circle, fill, minimum size=#1,
              inner sep=0pt, outer sep=0pt},
      dot/.default = 1pt
      ]
      \node[dot] at (2.1,-2.5) {};
      \node[dot] at (2.1,-2.6) {};
      \node[dot] at (2.1,-2.7) {};
      
      \node (q0) [state] {$q_0$};
      \node (q1) [state, right = of q0] {$q_1$};
      \node (q2) [state, right = of q1] {$q_2$};
      
      \path [-stealth, thick]
        (q0) edge node[above] {\scriptsize$a$}  (q1)
        (q1) edge node[above] {\scriptsize$b$}    (q2)
        
        (q1) edge [loop above]  node[above] {\scriptsize$\overleftarrow{a}=1 \land \overrightarrow{b}\leq -K$} node[below] {\scriptsize$a$} ()
        
        (q1) edge [loop below]  node[below, inner sep=1pt] {\scriptsize$\overleftarrow{a}=1 \land$} node[below, outer sep=6.5pt] {\scriptsize$ \overrightarrow{c_1}\leq -K$} node[above] {\scriptsize$c_1$} ();

        \path[->,min distance=3.5cm](q1)edge[in=225,out=325,below] node[above]{\scriptsize$c_2$} node[below]{\scriptsize$\overleftarrow{a}=1 \land \overrightarrow{c_2}\leq -K$} (q1);
      
        \path[->,min distance=6cm](q1)edge[in=225,out=325,below] node[above]{\scriptsize$c_N$} node[below]{\scriptsize$\overleftarrow{a}=1 \land \overrightarrow{c_N}\leq -K$} (q1);
      
      \end{tikzpicture}
      \caption{$\toyeca(K,N)$}
      \label{fig:B3} 
      \end{figure}	         
      \end{minipage}%
	}
\end{figure}

\medskip\noindent \textbf{$\abpp$}: The property checks the following for the Sender process of $\abp$: after the sending $send_0$, the sender should receive an $ack_0$ before sending $send_1$.  
We model the negation of this property using an ECA as given in Figure~\ref{fig:ABP:P1}.

\medskip\noindent \textbf{$\abpq$}: The property that after sending a $send_0$, the sender must receive an $ack_0$ within 3 time units. We model the negation of this property with an ECA given in Figure~\ref{fig:ABP:P2}.
Note that this property does not hold for the sender in $\abp$.

\subsection*{Synthetic Benchmarks in Table~\ref{table:exp-results}}

\medskip\noindent $\toyeca(K,N)$: As depicted in Figure~\ref{fig:B3}, $\toyeca(K,N)$ has two parameters - $K$, which is the maximal constant and $N$, which is the number of $c_i$ loops (on state $q_1$), in the automaton.  

From $q_0$, on a transition $a$, the automaton goes to $q_1$.
From $q_1$, it can either take the loop $a$, after which the $b$ action taking it to state $q_2$ can be taken only after $K$ time units. 
Alternately,  from $q_1$ a sequence of distinct $c_i$ loops can be taken in zero-time (because of the $\overleftarrow{a}=1$ guard). Note that two $c_i$ actions can be taken only at an interval of greater than $K$ time units.

From Table~\ref{table:exp-results}, we observe an order of magnitude improvement, both in terms of running time as well as number of visited and stored nodes w.r.t. the standard approach.
Recall that there is a blow up (both in the number of states and clocks) while converting an ECA to a timed automaton. 
The effect caused by the blow up in clocks also affects the time taken for each zone operation.
Further, note that even when we increase the parameters $K, N$, despite an increase in runtime, the number of visited and stored nodes does not increase - this is because even though we explore more nodes of the zone graph (because of an increase in number of transitions), these explorations lead to nodes that are subsumed by nodes that have already been visited. 

\end{document}

%% file: declarations.tex
\newcommand{\Aa}{\mathcal{A}}

\newcommand{\Gg}{\mathcal{G}}

\newcommand{\RR}{\mathbb{R}}
\newcommand{\Rpos}{\RR_{\ge 0}}
\newcommand{\Rneg}{\RR_{\le 0}}

\newcommand{\gzg}{\ensuremath{\mathsf{GZG}}}

\newcommand{\xra}[1]{\xrightarrow{#1}}

\renewcommand{\d}{\delta}
\newcommand\sem[1]{{[\![ #1 ]\!]}}

\newcommand{\prophecy}[1]{\overrightarrow{#1}}
\newcommand{\history}[1]{\overleftarrow{#1}}

\newcommand{\A}{\mathcal{A}}
\newcommand{\G}{\mathcal{G}}
\newcommand{\asplit}{\mathsf{split}}
\newcommand{\pre}[2]{\mathsf{pre}(#1,#2)}

\newcommand{\leqlt}{\mathrel{\triangleleft}}
\newcommand{\ua}{{\uparrow}}
\newcommand{\da}{{\downarrow}}
\newcommand{\GG}{\mathbb{G}}
\newcommand{\graph}[1]{\GG(#1)}
\newcommand{\graphv}[2]{\GG_{#1}(#2)}
\newcommand{\V}{\mathbb{V}}

\newcommand{\elapse}[1]{\overrightarrow{#1}}

\newcommand{\prog}{\textrm{prog}}
\newcommand{\guard}{\mathrm{guard}}
\newcommand{\chg}{\mathrm{change}}

\newcommand{\Prog}{\textrm{Programs}}

\newcommand{\GTA}{\textrm{GTA}}
\newcommand{\GTAs}{\textrm{GTAs}}
\newcommand{\GTAFULL}{Generalized Timed Automata}

\newcommand{\GTAfull}{generalized timed automata}
\newcommand{\Atim}{A_{\ubt}}

\newcommand{\ubt}{\mathbin{\rotatebox[origin=c]{90}{$\bowtie$}}}

\newcommand{\timeout}{\mathsf{timeout}}

\newcommand{\toyeca}{\mathsf{ToyECA}}

\newcommand{\firealarm}{\operatorname{\mathsf{Fire}-\mathsf{alarm}}}
\newcommand{\firealarmpattern}{\operatorname{\mathsf{Fire}-\mathsf{alarm}-\mathsf{pattern}}}

\newcommand{\Sensor}{\mathsf{Sensor}}
\newcommand{\alive}{\mathsf{alive}}

\newcommand{\csmacd}{\mathsf{CSMACD}}
\newcommand{\csmacdbounded}{\operatorname{\mathsf{CSMACD}-\mathsf{bounded}}}

\newcommand{\Station}{\mathsf{Station}}
\newcommand{\Bus}{\mathsf{Bus}}

\newcommand{\altbp}{\operatorname{\mathsf{Alternating}-\mathsf{bit}-\mathsf{protocol}}}
\newcommand{\abp}{\mathsf{ABP}}
\newcommand{\abpp}{\operatorname{\mathsf{ABP}-\mathsf{prop1}}}
\newcommand{\abpq}{\operatorname{\mathsf{ABP}-\mathsf{prop2}}}